\documentclass{comnet}%%%%where comnet is the template name
\usepackage{amsfonts}
\usepackage{graphicx}
\usepackage{epstopdf}
\usepackage{algorithm}
\usepackage{algorithmic}
\usepackage{kotex}
\ifpdf
  \DeclareGraphicsExtensions{.eps,.pdf,.png,.jpg}
\else
  \DeclareGraphicsExtensions{.eps}
\fi

\usepackage{enumitem}
\setlist[enumerate]{leftmargin=.5in}
\setlist[itemize]{leftmargin=.5in}

\usepackage{amsopn}

\usepackage{newtxtext}
\usepackage{newtxmath}
\usepackage{bm}
\usepackage{subcaption}
\usepackage{multirow}
\usepackage{pdflscape}%for rotating tables
\usepackage{orcidlink}
\usepackage{xr-hyper}
\usepackage{cleveref}

\DeclareTextFontCommand{\textib}{%
  \fontseries\bfdefault % change series without selecting the font yet
  \itshape
}
\newcommand{\bA}{{\mathbf A}}

\newcommand{\bfd}{{\mathbf d}}

\newcommand{\bff}{{\mathbf f}}

\newcommand{\bv}{{\mathbf v}}

\newcommand{\bz}{{\mathbf z}}

\newcommand{\bTheta} {\bm{\Theta}}

\newcommand{\bxi} {\bm{\xi}}

\newcommand{\btau}{\bm{\tau}}

\newcommand{\bone}{{\mathbf 1}}

\newcommand{\mbR}{{\mathbb R}}

\DeclareMathOperator*{\argmin}{argmin}
\DeclareMathOperator*{\argmax}{argmax}
\DeclareMathOperator*{\cov}{Cov}
\DeclareMathOperator*{\var}{Var}

\newcommand{\I}{\mathbb{I}}
\newcommand{\E}{\mathbb{E}}

\newcounter{cnstcnt}

\newcommand{\R}{{\textsf{R}}}

\newcommand{\T}{\mathrm{T}}

\newcommand{\iid}{{\textit{i.i.d.} }}

\makeatletter
\def\ps@plain{%
  \let\@oddhead\@empty
  \let\@evenhead\@empty
  \def\@oddfoot{\hfill\thepage\hfill}%
  \let\@evenfoot\@oddfoot
}
\let\@received\@empty
\let\@revised\@empty
\let\@accepted\@empty
\makeatother

\begin{document}

\title{Joint Estimation of Sparse Multilayer Networks via Graph Limits}

\shorttitle{Joint Estimation of Sparse Multilayer Networks via Graph Limits} %%%for recto running head
% \shortauthorlist{Y. Song and S. C. Olhede} %%% for verso running head

\author{%%%% First author details
Youngseok Song~\orcidlink{0000-0003-3577-7509}
\address{School of Mathematical and Data Sciences, West Virginia University, Morgantown, West Virginia 26506, USA
\email{youngseok.song@mail.wvu.edu}}
%%%%%%% Second author details
\and
Sofia C. Olhede~\orcidlink{0000-0003-0061-227X}
\address{Chair of Statistical Data Science, EPFL, 1015 Lausanne, Switzerland
\email{sofia.olhede@epfl.ch}}
%%%%%%%
% \and
% %%%%%%% Third author details
% \name{Insert third author}
% \address{Third author address}}
}
\maketitle
\begingroup
\renewcommand{\thefootnote}{}
\footnotetext{* Submitted to the Journal of Complex Networks.}
\endgroup

\begin{abstract}
{Network datasets in modern applications often involve multiple types of interactions occurring over a shared set of individuals. Characterizing the generating mechanisms of these interactions can be enhanced by joint modelling, as shared vertices allow layers to help explain the structure of other layers. We model multiplex observations using graph limits, called a scaled set of graphons, and develop a nonparametric joint estimator based on blockmodel approximations, termed the multi-network histogram. This nonparametric framework captures each layer's varying sparsity and connection structure, accounting for heterogeneity via shared latent variables across all layers. We establish the theoretical properties of the multi-network histogram, providing an upper bound for the weighted mean integrated squared error and deriving the optimal bandwidth that minimizes this error. 
By leveraging information across layers, this joint modelling achieves a reduction in error and a smaller optimal bandwidth, which enables high-resolution estimation even in sparser layers. 
Its usefulness is demonstrated through simulation studies and an application to socioeconomic networks in an Indian village.
}
{multiplex network, joint graphon estimation, nonparametric statistics, network histograms, sparse networks, heterogeneous graphons
}
%%%% If classification number provided then
\\
2000 Math Subject Classification: 
% 34K30, 35K57, 35Q80,  92D25
05C80, 60G09, 62G05 
\end{abstract}

\section{Introduction}
% Motivation

This paper proposes methods for nonparametric estimation of the generative mechanism of sparse multilayer networks, where all layers are defined on a common set of vertices. A multilayer network consists of several distinct layers with both inter-layer and intra-layer connections \cite{bianconi2018Multilayer}. Such structures are particularly relevant for modelling systems characterized by several interaction types among entities, such as biological \cite{dedomenico2015MuxViz, vazquez2003Changes}, social \cite{banerjee2013Diffusion,padgett1993Robust,sapiezynski2019Interaction}, economic \cite{banerjee2013Diffusion,dedomenico2015Structuralb}, or transportation networks \cite{asgari2016CTMapper,cardillo2013Emergence,dedomenico2014Navigability}. 
Among multilayer networks, our focus lies on those in which all layers share an identical vertex set without inter-layer edges between distinct nodes. 
\begin{figure}[!ht]
\centering
\begin{subfigure}[b]{0.225\linewidth}
    \centering
    \includegraphics[width=\textwidth]{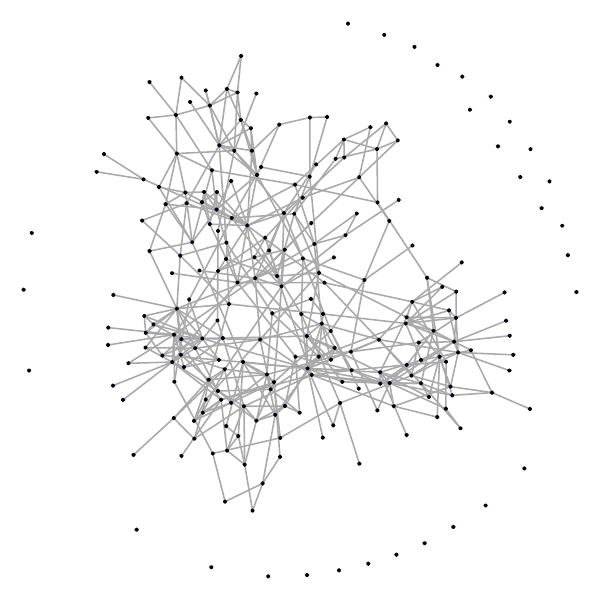}
    \caption{Borrow money}
    \label{subfig:igraph_borrowmoney}
\end{subfigure}
\hfill
\begin{subfigure}[b]{0.225\linewidth}
    \centering
    \includegraphics[width=\textwidth]{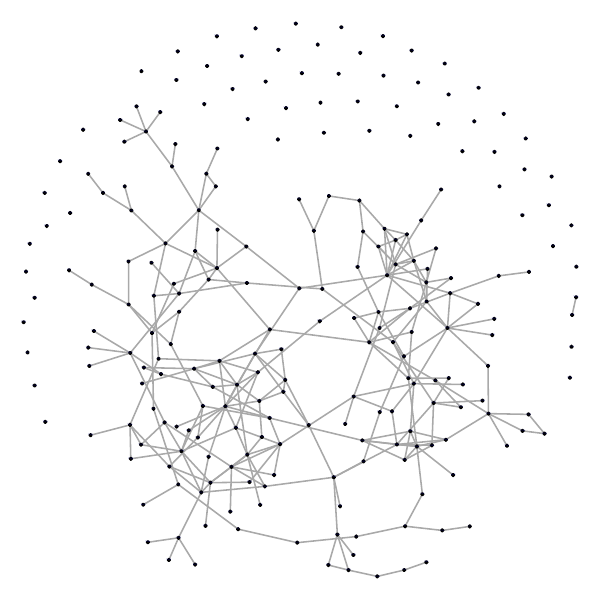}
    \caption{Give advice}
\end{subfigure}
\hfill
\begin{subfigure}[b]{0.225\linewidth}
    \centering
    \includegraphics[width=\textwidth]{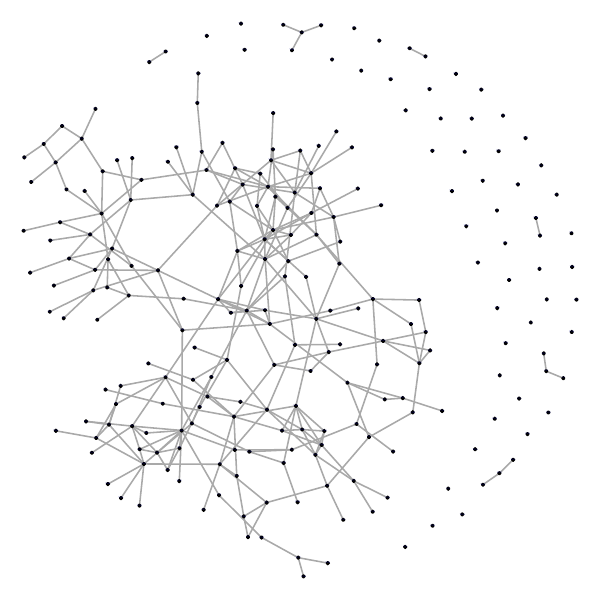}
    \caption{Help decision}
\end{subfigure}
\hfill
\begin{subfigure}[b]{0.225\linewidth}
    \centering
    \includegraphics[width=\textwidth]{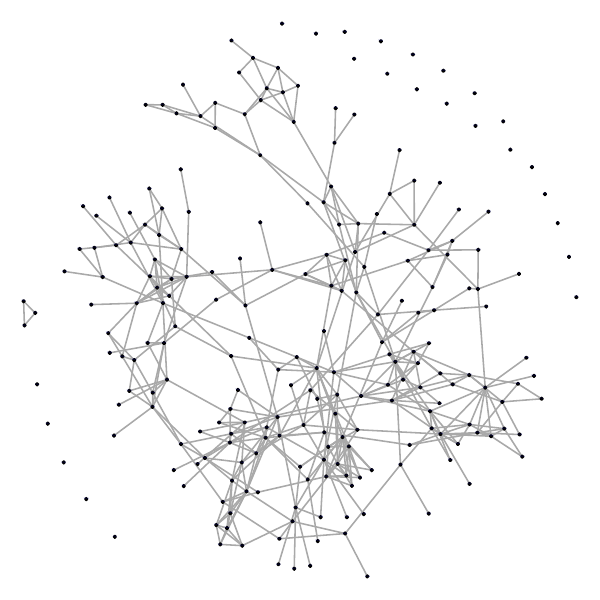}
    \caption{Kero rice come}
\end{subfigure}
\\
\begin{subfigure}[b]{0.225\linewidth}
    \centering
    \includegraphics[width=\textwidth]{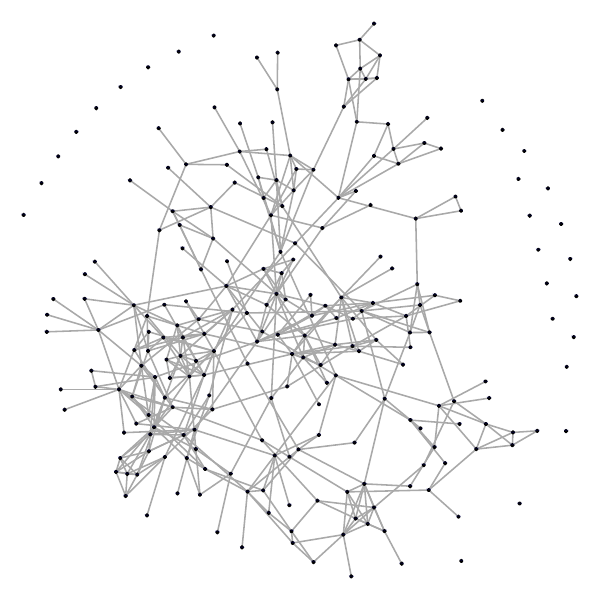}
    \caption{Kero rice go}
\end{subfigure}
\hfill
\begin{subfigure}[b]{0.225\linewidth}
    \centering
    \includegraphics[width=\textwidth]{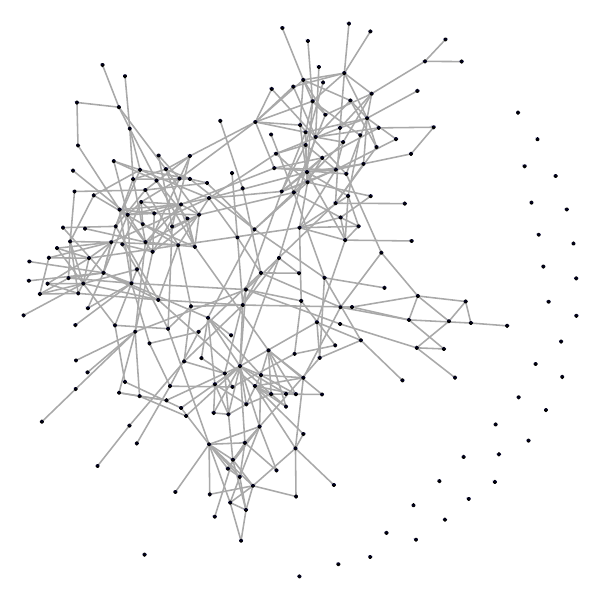}
    \caption{Lend money}
\end{subfigure}
\hfill
\begin{subfigure}[b]{0.225\linewidth}
    \centering
    \includegraphics[width=\textwidth]{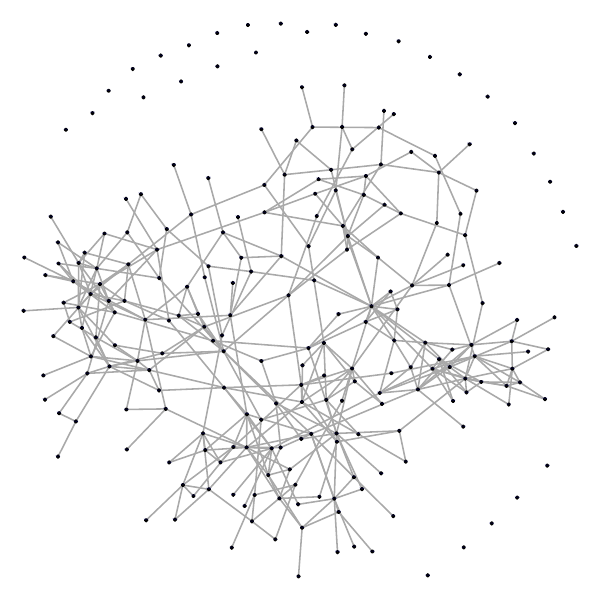}
    \caption{Medic}
\end{subfigure}
\hfill
\begin{subfigure}[b]{0.225\linewidth}
    \centering
    \includegraphics[width=\textwidth]{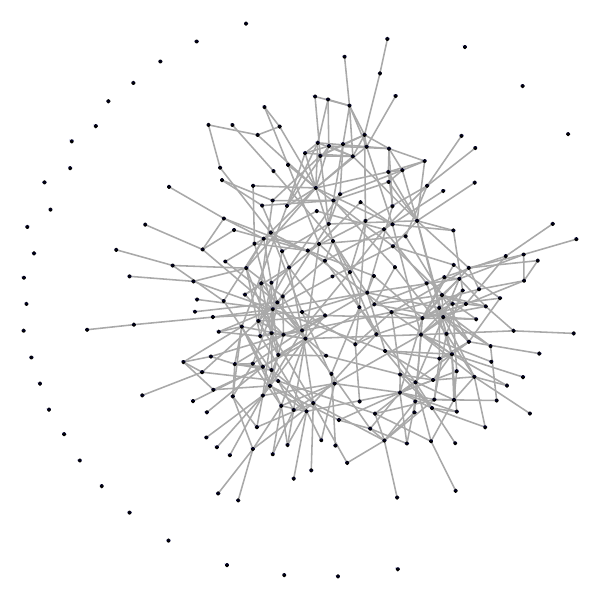}
    \caption{Nonrel}
    \label{subfig:vil_40_nonrel_igraph}
\end{subfigure}
\\
\begin{subfigure}[b]{0.225\linewidth}
    \centering
    \includegraphics[width=\textwidth]{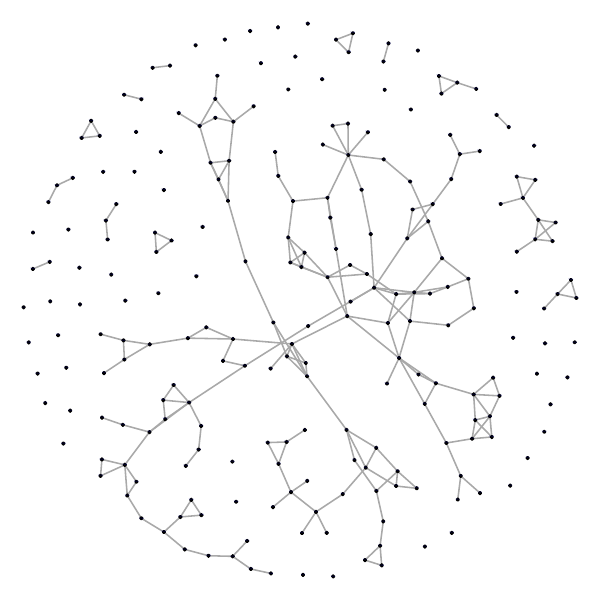}
    \caption{Rel}
    \label{subfig:vil_40_rel_igraph}
\end{subfigure}
\hfill
\begin{subfigure}[b]{0.225\linewidth}
    \centering
    \includegraphics[width=\textwidth]{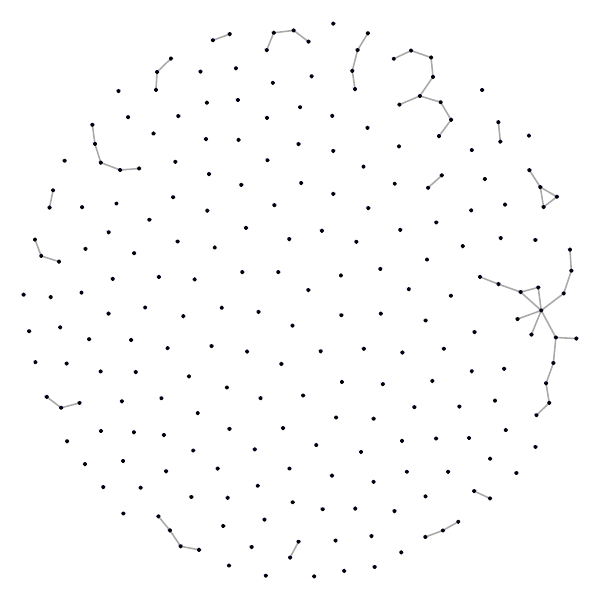}
    \caption{Temple company}
    \label{subfig:vil_40_temple_igraph}
\end{subfigure}
\hfill
\begin{subfigure}[b]{0.225\linewidth}
    \centering
    \includegraphics[width=\textwidth]{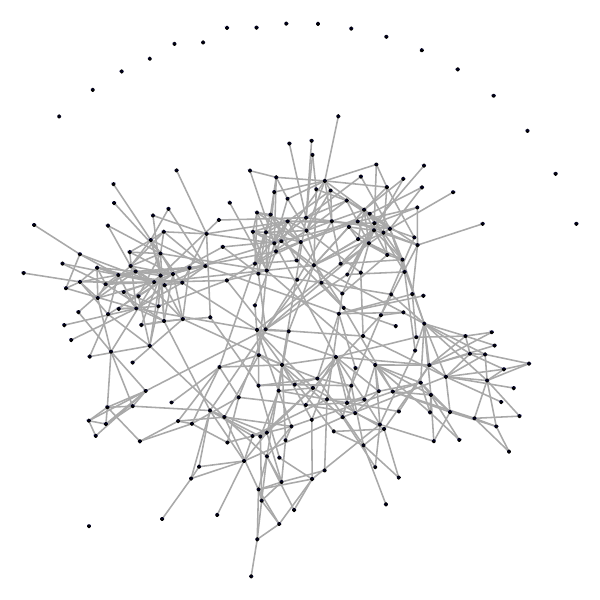}
    \caption{Visit come}
    \label{subfig:vil_40_visitcome_igraph}
\end{subfigure}
\hfill
\begin{subfigure}[b]{0.225\linewidth}
    \centering
    \includegraphics[width=\textwidth]{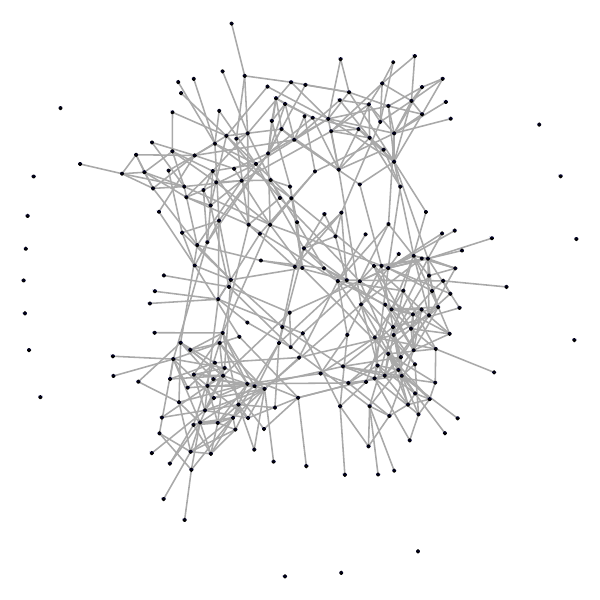}
    \caption{Visit go}
\end{subfigure}
\caption{Household network of village ID 40 in Indian village data \cite{banerjee2013Diffusion}. Each layer displays a socioeconomic relationship of 231 households. Note that the vertices have not been projected to the same location. For more details about the dataset, see Supplementary Material Section~\ref{sec_append:indian_vils}.}
\label{fig:indian_village_vil_40}
\end{figure}

Networks of this type present unique challenges in modelling their structures. 
To illustrate this, we consider Indian village network data \cite{banerjee2013Diffusion}, shown in Figure~\ref{fig:indian_village_vil_40}. In this dataset, all layers are defined on the same set of households, with each layer representing a distinct type of socioeconomic interaction and edges indicating the presence of interactions. The layers display varying levels of sparsity and structure. For example, the Temple company layer in Figure~\ref{subfig:vil_40_temple_igraph} is the sparsest layer, while the Visit come layer in Figure~\ref{subfig:vil_40_visitcome_igraph} is the densest. Some layers differ in structure, for example, the Rel layer in Figure~\ref{subfig:vil_40_rel_igraph} exhibits clear small-sized clusters, whereas the Nonrel layer  in Figure~\ref{subfig:vil_40_nonrel_igraph} appears to be randomly connected. To account for layer-wise heterogeneity, we develop a joint modelling framework based on graph limit models.

A graphon is a graph limit model for dense and vertex-exchangeable networks~\cite{lovasz2012large}. To accommodate sparse structures commonly observed in real-world networks, percolation can be applied to the underlying graphon. This leads to the scaled graphon framework~\cite{bickel2009nonparametric}, in which the generative mechanism is characterized by an underlying structural pattern along with a sparsity parameter that may vanish with network size.  
Introducing a sparsity parameter is crucial, as the sparsity level of a network determines, to first order, its sampling characteristics.
Specifically, networks with extremely few edges are tree-like in structure, whereas cycles become ubiquitous as edge density increases. Many network features, such as the counts of various cycles and other motifs, are primarily governed by the sparsity level, see for example, \cite{bickel2009nonparametric, bickel2011method, bollobas2007phase}.

For multilayer networks, sparsity levels often vary substantially, which further complicates network estimation. Applying a uniform sparsity scaling to all layers fails to capture these layer-specific structural differences and varying interaction costs. To address these challenges, we propose a {scaled set of graphons}, a multilayer network model that allows sparsity and structure to differ across layers. Central to our approach is the use of shared latent variables to represent a common vertex set, which facilitates identifying a shared group structure and joint modelling. This mechanism, incorporating layer-specific characteristics, captures coarse structural features and enables efficient joint analysis.

Graphon estimation has been extensively investigated in the single-layer setting. Several methods based on blockmodel approximation have been proposed, including vertex grouping by profile likelihood maximization \cite{olhede2014network,wolfe2013Nonparametric}, least-squares estimation \cite{gao2015rate}, and degree-based sorting \cite{chan2014Consistent}. In particular,  \cite{olhede2014network} introduced the network histogram, a nonparametric graphon estimator based on blockmodel approximations with equal-sized bins. This estimator has been further extended to accommodate missing edges \cite{gaucher2021Maximum}, non-rectangular block shapes \cite{verdeyme2024Hybrid}, and local linear estimation incorporating vertex covariates \cite{chandna2022Local}. 
Beyond blockmodel-based approaches, various alternative methodologies have been developed for single-layer network estimation, including universal singular value thresholding \cite{chatterjee2015Matrix} and neighbourhood smoothing \cite{zhang2017Estimating}.
While estimation methods for single-layer networks are well developed, those for multilayer networks remain relatively limited. 
Applying the aforementioned single-layer methods to each layer of a multilayer network independently leads to poor estimates for sparser layers, failing to take advantage of information available from other layers.
Moreover, existing multilayer approaches assume either a common sparsity parameter~\cite{chandna2020Nonparametric} or a common graphon~\cite{navarroGraphonaidedJointEstimation2022} across layers. These methods may not adequately capture layer-specific heterogeneity and are not based on blockmodel approximation. 
Recent studies have extended neighbourhood smoothing approaches to multilayer networks \cite{he2026Joint,guo2026Connection}, but they do not address settings where the layer-wise sparsity levels decay at different orders.
In addition, graphon-based approaches have been proposed for estimating time-varying networks \cite{lee2026Nonparametric}. This approach relies on the temporal ordering of layers, which is not available in general multilayer networks.

We shall therefore propose the {multi-network histogram} to jointly identify the underlying structure of a multilayer network defined on the same set of vertices. This method estimates a scaled set of graphons by leveraging the increased degrees of freedom available across layers while accounting for heterogeneous sparsity levels and structures. Specifically, it is constructed by applying blockmodel approximations to each layer based on the same group assignments and a shared bandwidth across all layers. This enables joint estimation of group assignments by maximizing the total profile likelihood, and it also allows even sparser layers to use smaller bandwidths, leading to a finer resolution of the network histogram. We establish theoretical properties regarding the weighted mean integrated squared error of the multi-network histogram using oracle group labels. These results characterize the effects of layer-specific sparsity and bandwidth choice on estimation and provide guidance for practical implementation. We also develop the homogeneous multi-network histogram, a special case in which all layers are generated from the same graphon but may differ in sparsity.

The rest of the paper is organized as follows. \Cref{sec:background} provides an overview of graphons and network histograms. \Cref{sec:method} introduces the scaled set of graphons and multi-network histograms, along with the corresponding algorithm. \Cref{sec:theory} presents the theoretical results. \Cref{sec:sim_study} demonstrates the performance of the proposed methods through simulations. In \Cref{sec:real_data}, we apply the method to a multilayer network of socio-economic relationships. \Cref{sec:discussion} concludes the paper with a brief discussion. All proofs from the theoretical study, along with additional tables and figures from simulations and real data analyses, are provided in the Supplementary Material.

\section{Background}
\label{sec:background}

In this section, we provide a brief overview of graphon models and a graphon estimation method known as {network histogram} for single-layer network \cite{olhede2014network}. These concepts are essential for understanding challenges and motivations behind our proposed methodology in \Cref{sec:method}.

We first introduce notations that will be used throughout the paper. We denote a finite index set $\{1,\ldots, n\}$ by $[n]$. The uniform distribution on the interval $[0,1]$ is denoted by $U(0,1)$. Bold capital letters represent matrices, for example, adjacency matrices $\bA=\{A_{ij}\}_{1\leq i,j\leq n}$, while bold lowercase letters denote vectors, for example, $\bv = (v_1,\ldots,v_n)$. We use superscript $(\ell)$ on matrices and vectors to denote that they correspond to the $\ell$th layer. For example, $\bA^{(\ell)}$ denotes the adjacency matrix of the $\ell$th layer of a multilayer network. The ``generalized" inverse of $c$ is defined by $c^{-}=c^{-1}$ if $c\neq 0$, otherwise, $c^{-}=0$.

\subsection{Graphon and Scaled Graphon}
\label{subsec:graphons}

{Graphons} and {scaled graphons} model vertex exchangeable networks, whose distribution is invariant under permutations of the vertex labels. We refer readers to \cite[Chapter 15 and 16]{izenman2023Network} for an introduction to graphons, scaled graphons, and their estimation, and to \cite{lovasz2012large} for a detailed mathematical review of graphons.

A {graphon} is defined as a symmetric measurable function, $\phi(x,y): [0,1]^2\to [0,1]$, representing the probability of an edge between the latent positions of vertices $x$ and $y$. An undirected, binary, and vertex exchangeable dense network can be represented by a random graphon, as outlined in the Aldous-Hoover Theorem \cite{aldous1981Representations,hoover1979Relations}. A graphon serves both as a limiting object for a sequence of networks and as a generative model for dense networks. 
However, since countable vertex exchangeable networks are almost surely either dense or empty  \cite[see Chapter 6.5.1]{crane2018probabilistic}, graphons are not well suited for modelling real-world networks that are commonly sparse.

A {scaled graphon}, also known as a sparse graphon, is used to model sparse networks \cite{bickel2009nonparametric,bollobas_riordan_2009}. Using a possibly vanishing term $\rho_n$, the scaled graphon accommodates sparsity in large networks, defined as $\phi(x,y) = \rho_n f(x,y) \in [0,1]$, where $\iint_{[0,1]^2}f(x,y)dxdy = 1$. Under this specification, $\rho_n$ controls the sparsity level. A dense network corresponds to $\rho_n = O(1)$, while a sparse network corresponds to $\rho_n=o(1)$ and $\rho_n=\omega(n^{-1})$, following the regime described in \cite{bollobas_riordan_2009}. To generate a network of size $n$ from a scaled graphon, \iid latent variables for $n$ vertices, $\xi_1,\ldots, \xi_n$, are sampled from $U(0,1)$. Then, for each pair $1\leq i < j \leq n$, edges are drawn from independent Bernoulli random variables with probabilities $p_{ij}=P(A_{ij}=1|\xi_i,\xi_j) = \rho_n f(\xi_i,\xi_j)$. In this sense, a network is instantiated from the scaled graphon. This can be interpreted as a percolation step, where edges in the dense underlying network defined by $f$ are randomly retained with probability $\rho_n$ to achieve the desired sparsity. In practice, $\rho_n$ should be interpreted as a scaling for finite networks, rather than as describing a mathematical limit for infinitely large networks.

\subsection{Network Histogram}
\label{subsec:nethist}

The network histogram is a nonparametric method for estimating a graphon \cite{olhede2014network}. It uses blockmodel approximations to a H\"{o}lder smooth graphon \cite{wolfe2013Nonparametric}, which involves grouping similar vertices and then computing blockwise edge densities within and between these groups. The vertex partition can be obtained using profile likelihood estimation \cite{olhede2014network,wolfe2013Nonparametric}. The group size, referred to as the bandwidth of the network histogram, is set equal across all groups, analogous to standard histograms. Using this consistent resolution reduces errors in estimating the underlying network structure.

We briefly formulate the network histogram for a single-layer network. Let $\bz = (z_1,\ldots, z_n) \in \mathcal{Z}_{k}$ be the vector of group labels, where $\mathcal{Z}_k\subset [k]^{n}$ is the set of all possible group labels of $n$ vertices into $k$ groups. When $n=hk+r$ for integer bandwidth $h$, the partition consists of $k-1$ groups of size $h$ and one group of size $h+r$.
Following \cite{wolfe2013Nonparametric}, the log-likelihood for the blockmodel approximation is given by
\begin{align*}
L(\bA; \bz, \bTheta)
= \sum_{i<j}\Big\{&A_{ij}\log \theta_{z_i z_j}
+ (1-A_{ij})\log(1-\theta_{z_i z_j})\Big\},
\end{align*}
where $\bTheta= \{\theta_{ab}\}_{a,b\in[k]}$ is the block probability matrix.
Given a group assignment $\bz$, the log-likelihood is maximized at $\theta_{z_i z_j}= \bar{A}_{z_iz_j}$, where $\bar{A}_{z_iz_j}$ is the edge density within the block containing vertices $i$ and $j$. 
The network histogram approach then identifies the optimal group label vector by maximizing the profile log-likelihood, $\widehat{\bz} = \argmax_{\bz\in\mathcal{Z}_{k}} L(\bA; \bz)$,
where
\begin{equation}
\label{eq:nethist}
\begin{aligned}
L(\bA; \bz)
:=\sum_{i<j}\Big\{&A_{ij}\log \bar{A}_{z_i z_j} + (1-A_{ij})\log (1-\bar{A}_{z_i z_j})\Big\}.
\end{aligned}
\end{equation} 
The optimal bandwidth $h^*$, which determines the group size, is selected as the one that minimizes the oracle mean integrated squared errors, see Theorem 1 of \cite{olhede2014network}. 
The selected $h^*$ increases as the number of vertices $n$ increases, as the network becomes sparser (smaller $\rho_n$), and as the underlying smooth graphon becomes flatter (smaller H\"{o}lder constant).
\begin{figure}[!tbp]
\centering
\begin{subfigure}[b]{0.45\linewidth}
    \centering
\includegraphics[width=\textwidth,trim={0.5cm 1.5cm 0.2cm 1cm},clip]{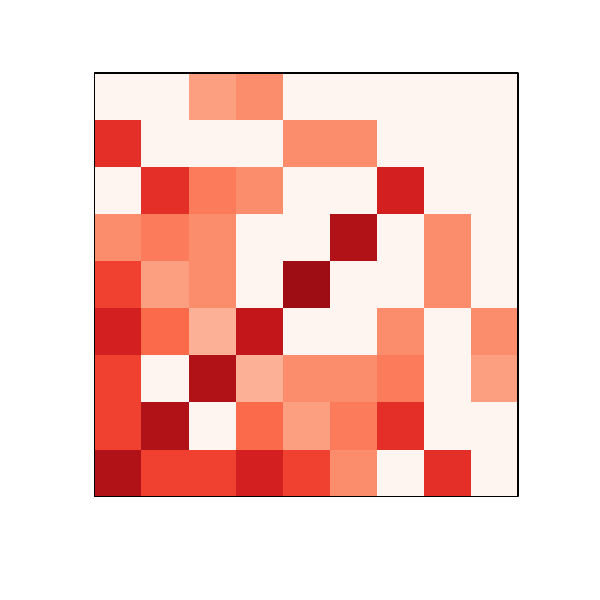}
    \caption{Borrow money (layer-wise)}
    \label{subfig:nethist_borrowmoney}
\end{subfigure}
% \hfill
\begin{subfigure}[b]{0.485\linewidth}
    \centering
    \includegraphics[width=\textwidth,trim={0.5cm 1.5cm 0.2cm 1cm},clip]{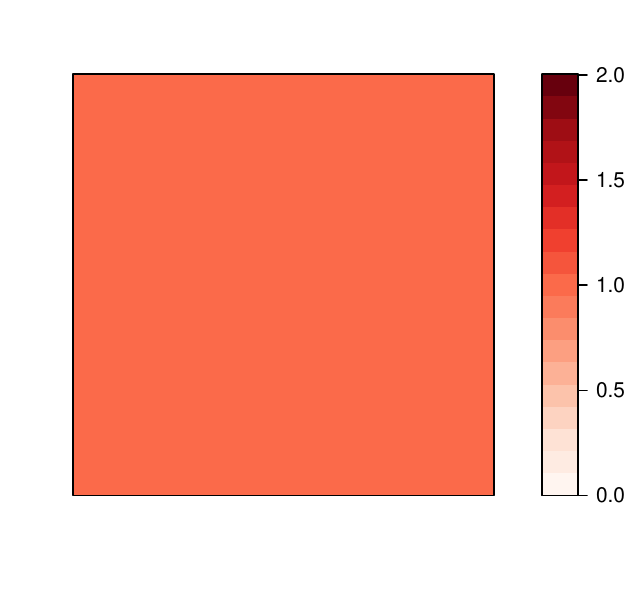}
    \caption{Temple company (layer-wise)}
    \label{subfig:nethist_templecompany}
\end{subfigure}
\caption{Heatmaps of fitted layer-wise network histograms from the household network of village ID 40 in Indian village data. Displayed legend values are $\{\widehat{f}^{(\ell)}(x,y)\}^{1/4}$ and are applied to both figures.}
\label{fig:layer_wise_nethist_indian_vil_40}
\end{figure}

Fitting network histograms to each layer separately in a multilayer network, however, has limitations. First, it produces different layer-wise bandwidths and group labels for each layer, failing to capture the joint structure of the multilayer network. 
Moreover, since the bandwidth in each layer depends on its sparsity, sparser layers are estimated with lower resolution and produce poor estimates. For example, the Borrow money layer in Figure~\ref{subfig:igraph_borrowmoney}  contains enough within-layer connections to separate groups into small blocks as in Figure~\ref{subfig:nethist_borrowmoney}, whereas the Temple company layer in Figure~\ref{subfig:vil_40_temple_igraph} has too few connections, causing it to collapse into a single large block as in Figure~\ref{subfig:nethist_templecompany}. To overcome these limitations, we introduce a joint estimation approach for multiplex networks in \Cref{sec:method}.

\section{Method}
\label{sec:method}

This section develops a nonparametric estimation framework for multiplex networks where all layers represent different types of interactions over an identical vertex set. Our approach is designed to accommodate layers that may differ in both their connection structures and levels of sparsity. We first define the notion of {a scaled set of graphons}. Then we introduce the {multi-network histogram}, which is a joint graphon estimation method using common latent variables across layers.

\subsection{Scaled Set of Graphons}
\label{subsec:SMG}

We introduce {a scaled set of graphons}, which is a modelling framework for multilayer networks that exhibit layer-wise heterogeneity in both sparsity and network structure. Its formal definition is as follows. 
\begin{definition}[A scaled set of graphons]
\label{def:scaled_multigraphon}
A set of non-negative integrable functions $\{f^{(\ell)}(x,y): [0,1]^2 \to \mbR_0^+, \ell\in[L]\}$ is called a {scaled set of graphons} if they are coupled with $L$ non-negative scaling sequences enumerated in index $n$, $\{\rho_n^{(1)},\dots, \rho_n^{(L)}\}$ where $\rho_n^{(\ell)} \in (0,1)$, so that $\rho_n^{(\ell)}f^{(\ell)}(x,y)$ defines a valid probability.
We set  $f^{(\ell)}(x,y) = f^{(\ell)}(y,x)$ for symmetry, and $\iint_{[0,1]^2} f^{(\ell)}(x,y)dxdy = 1$ for identifiability. 
\end{definition}

In this framework, a fixed set of functions, $\{f^{(1)},\ldots, f^{(L)}\}$, represents the underlying connection structures within each layer. Layer-specific sparsity parameters, $\{\rho_n^{(1)},\ldots, \rho_n^{(L)}\}$, are the proportions of within-layer connections, ranging from sparse to dense. This formulation facilitates modelling both diverse network patterns and different sparsity levels in a multilayer network.
As a special case of Definition~\ref{def:scaled_multigraphon}, when $L=1$, it reduces to the scaled graphon for single-layer network discussed in \Cref{subsec:graphons}.
When interpreting a scaled set of graphons, care must be taken in understanding the large sample properties of the multiple layers. As with standard scaled graphons for single-layer networks, see e.g.~\cite{bollobas_riordan_2009}, the model should be understood as a scaling when the networks consist of a large number of vertices. 

A scaled set of graphons can be used to generate multiplex networks. Each latent variable, $\xi_i \in [0,1]$, shared across all layers captures intrinsic characteristics of each vertex $i\in[n]$. 
Then, the edge probability between vertices $i$ and $j$ in the $\ell$th layer is determined by a scaled graphon corresponding to that layer.
Conditional on the latent variables, the edges are independent both within and across layers.
Example~\ref{ex:scaled_multigraphon} outlines the sampling scheme.
\begin{example}[Sampling an undirected multiplex network from a scaled set of graphon with the common latent variables]
\label{ex:scaled_multigraphon}
Let $\xi_1,\ldots, \xi_n$ be an \iid sample from $U(0,1)$.
Given a scaled set of graphon, we define the edge probability between the vertex positions $\xi_i$ and $\xi_j$ in the $\ell$th layer by $\rho_n^{(\ell)} f^{(\ell)}(\xi_i,\xi_j)
$. The adjacency matrix entries $\{A_{ij}^{(\ell)}\}_{i,j\in[n]}$ are modelled using independent Bernoulli trials,
\begin{align}
A_{ij}^{(\ell)}|\xi_i,\xi_j \sim \text{Bernoulli}\left(\rho_n^{(\ell)} f^{(\ell)}(\xi_i,\xi_j)\right),
\end{align}
for $i<j$. To ensure symmetry, we set $A_{ji}^{(\ell)} = A_{ij}^{(\ell)}$ for each $\ell\in[L]$.
\end{example} 

In addition, we consider a multiplex network in which all layers share the same interaction structure but differ only in sparsity.
In this case, we can generate a multiplex network using the procedure in Example~\ref{ex:scaled_multigraphon}, but applying the same function $f$ to all layers while allowing the layer-wise sparsity parameters $\rho_n^{(\ell)}$ to vary. Specifically, the adjacency matrix for the $\ell$th layer is generated independently given the latent variables as
\begin{equation}
\label{model:common_f}
\begin{aligned}
A_{ij}^{(\ell)}|\xi_i,\xi_j &\sim \textnormal{Bernoulli}\left(\rho_n^{(\ell)}f(\xi_i,\xi_j)\right).
\end{aligned}
\end{equation}

\subsection{Multi-Network Histogram}
\label{subsec:multinethist}

In this section, we propose a nonparametric estimator for a scaled set of graphons, referred to as {multi-network histogram}. It effectively estimates layer-specific functions using common latent variables, even when each layer's sparsity and interaction structure are highly heterogeneous. In addition, we introduce {homogeneous multi-network histograms} for layers that share an identical structure but have varying sparsity.

The multi-network histogram applies blockmodel approximations to each layer using the same group size and identical group assignments across all layers.
To begin with, we define $\mathcal{Z}_k\subset [k]^n$ as the set of all possible group label vectors of multi-network histograms with $k-1$ groups of size $h$ and a group of size $h+r$, when $n=hk+r$ and $0\leq r < h$. Let $R_{ab}(\bz)$ be the set of vertex pairs in the $(a,b)$-th block for the common group label vector $\bz= (z_1,\ldots, z_n)\in \mathcal{Z}_k$, given by
\begin{align*}
    R_{ab}(\bz) := \begin{cases}
        \{(i,j): i<j, z_i=z_j = a\}, & a=b,\\
        \{(i,j): z_i = a, z_j = b\}, & a\neq b.
    \end{cases}
\end{align*}
Denote the number of vertices in the $a$th group by $h_a = h + r \I(a=k)$ and  the number of possible edges in the $(a,b)$-th bin by $h_{ab}^2$, where
\begin{align*}
h_{ab}^2 := |R_{ab}(\bz)|= \begin{cases}
    \binom{h_a}{2} & a=b,\\
    h_a h_b & a\neq b.
    \end{cases}
\end{align*}
Then, the optimal group labels $\widehat{\bz} = (\widehat{z}_1,\ldots, \widehat{z}_n)$ are determined by maximizing the sum of layer-wise profile likelihoods \eqref{eq:nethist} under a common labelling across all layers. It is given by 
\begin{equation}
\label{eq:optim_pl_multinethist}
\begin{aligned}
\widehat{\bz} := \argmax_{\bz\in \mathcal{Z}_{k}}\sum_{\ell=1}^L\sum_{i<j}\Big\{&A_{ij}^{(\ell)}\log \bar{A}_{z_i z_j}^{(\ell)}
+ (1-A_{ij}^{(\ell)})\log (1-\bar{A}_{z_i z_j}^{(\ell)})\Big\},
\end{aligned}
\end{equation}
where the edge density of the bin containing vertices $i$ and $j$ in $\ell$th layer given $\bz$ is 
\begin{align*}
    \bar{A}_{z_i z_j}^{(\ell)}:=\bar{A}_{ab}^{(\ell)} (\bz)
    = \frac{1}{h_{ab}^2}\sum_{(i,j)\in R_{ab}(\bz)} A_{ij}^{(\ell)}.
\end{align*}
Once $\widehat{\bz}$ is obtained from \eqref{eq:optim_pl_multinethist}, 
the multi-network histogram is defined as
\begin{equation}
\label{eq:multi_nethist}
\widehat{f}_{\hat{z}_i \hat{z}_j} ^{(\ell)} := (\widehat{\rho}_n^{(\ell)})^{-}\bar{A}_{\hat{z}_i \hat{z}_j}^{(\ell)},
\ \ i,j\in[n], \ \ell\in[L],
\end{equation}
where $\widehat{\rho}_n^{(\ell)}= \binom{n}{2}^{-1}\sum_{i<j}A_{ij}^{(\ell)}$.
Each $\widehat{f}_{\hat{z}_i \hat{z}_j}^{(\ell)}$ provides an estimate of $f^{(\ell)}(\xi_i,\xi_j)$ for vertices $i$ and $j$ in $\ell$th layer.

Using the multi-network histogram for multiplex networks benefits particularly in estimating sparser layers compared to fitting a single-layer network histogram discussed in \Cref{subsec:nethist} to each layer.
Sparse layers contain few interactions, which makes it difficult to create sufficiently fine group partitions when applying the single-layer method to such layer. However, the multi-network histogram uses the common group labels obtained by maximizing the total profile likelihood \eqref{eq:optim_pl_multinethist}. This allows vertices to be partitioned more finely and improves blockmodel approximations by leveraging interaction information across all layers. 
As a result, the overall estimation error of the proposed joint method decreases, which is further explored in \Cref{sec:sim_study}.

\begin{remark}
To the best of our knowledge, this work is among the first to propose a nonparametric method for jointly estimating multilayer networks while allowing heterogeneity in both sparsity levels and network structures. 
\cite{chandna2020Nonparametric} assumes a common sparsity parameter across all layers, whereas our framework allows layer-specific sparsity parameters through the scaled set of graphons.  Recent neighbourhood smoothing approaches \cite{he2026Joint,guo2026Connection} identify similar layers using adjacency matrices $\mathbf{A}^{(\ell)}$ or preliminary connection probability estimates and borrow information across such layers, but do not model sparsity levels that decay at different rates.  In contrast, our multi-network histogram method explicitly incorporates layer-specific sparsity through the scaled graphon representation and accommodates heterogeneous graphon structures across layers.
\end{remark}

\begin{remark}
\label{remark:BM_approx}
The blockmodel approximation of graphons for multilayer networks differs from {multilayer stochastic blockmodels} \cite{han2015Consistent, paul2016Consistent} in several ways. Our approach targets nonparametric estimation under the common latent variable assumption, rather than vertex community detection \cite{paul2016Consistent, paul2020Spectral}, layer class partitioning \cite{donnat2018Tracking}, or parametric stochastic blockmodel estimations \cite{han2015Consistent}. Vertices are grouped only to estimate histogram bin heights, providing nonparametric estimates of H\"{o}lder-$\alpha$ smooth graphons, not to identify vertex communities or estimate blockmodel probabilities. This distinguishes our work from existing stochastic blockmodel-based methods for community detection in single-layer networks \cite{abbe2018Community, holland1983Stochastic, lee2019review} and multiplex networks \cite{chen2022global, fan2022alma, fu2023Profile, paul2016Consistent, paul2020Spectral, qing2024Estimating, zhang2024Consistent}.
\end{remark}

\begin{remark}
The use of a common labelling $\mathbf{z}$ across all $L$ layers in Algorithm~1 is not an artificial constraint. In a graph limit, a measure-preserving transformation is the analogue of a permutation of the vertices, so the same transformation must apply to every
layer \cite[Remark 3.3]{ganguly2025Multiplexons}.
The scaled graphon model in Definition~\ref{def:scaled_multigraphon}, with shared latent variables $\xi_i$, follows this framework. 
A common transformation is required across layers, but it does not have to be unique, and multiple labellings $\mathbf{z}$ are equally valid.
\end{remark}

Lastly, we introduce the {homogeneous multi-network histogram} for layers sharing a common function $f$, while sparsity levels $\rho_n^{(\ell)}$ may differ. 
It involves computing the weighted average of layer-specific estimates $\widehat{f}^{(\ell)}$, using weights proportional to the sparsity level of each layer.
Specifically, the weights are defined as
\begin{align}
\label{eq:sol_co}
    \tau_{\ell} = \frac{\widehat{\rho}_n^{(\ell)}}{\sum_{l=1}^L \widehat{\rho}_n^{(l)}},
\end{align}
and the weighted average of ${\widehat{f}_{z_i z_j}^{(\ell)}= { (\widehat{\rho}_n^{(\ell)}})^{-}\bar{A}_{z_i z_j}^{(\ell)}}$ for group labels $\bz$ is 
\begin{align*}
\widehat{f}_{z_i z_j} 
= \sum_{\ell=1}^L 
\tau_{\ell}\widehat{f}_{z_i z_j}^{(\ell)}
= \frac{\sum_{\ell=1}^L \bar{A}_{z_i z_j}^{(\ell)}}{\sum_{\ell=1}^L \widehat{\rho}_n^{(\ell)}}.
\end{align*}
Since $\widehat{\rho}_n^{(\ell)}\widehat{f}_{z_i z_j}$ serves as an estimate of  $\rho_n^{(\ell)}f(\xi_i,\xi_j)$ in \eqref{model:common_f}, 
the profile likelihood maximization problem \eqref{eq:optim_pl_multinethist} becomes
\begin{equation}
\label{eq:ll_homogeneous_multinethist}
\begin{aligned}
\widehat{\bz} = \argmax_{\bz \in \mathcal{Z}_k}&\sum_{\ell=1}^L\sum_{i<j}\Big\{
A_{ij}^{(\ell)}\log\left(\widehat{\rho}_n^{(\ell)}\widehat{f}_{z_i z_j}\right) + (1-A_{ij}^{(\ell)})\log\left(1-\widehat{\rho}_n^{(\ell)}\widehat{f}_{z_i z_j}\right)
\Big\},
\end{aligned}
\end{equation}
which determines the optimal group labels.
Given these labels, the homogeneous multi-network histogram is defined as
\begin{equation}
\label{eq:multinethist_homogenous}
\widehat{f}_{\hat{z}_i \hat{z}_j}  = \frac{\sum_{\ell=1}^L \bar{A}_{\hat{z}_i \hat{z}_j}^{(\ell)}}{\sum_{\ell=1}^L \widehat{\rho}_n^{(\ell)}}.
\end{equation}
In this homogeneous setting, using \eqref{eq:multinethist_homogenous} achieves lower variance than \eqref{eq:multi_nethist}. This is because taking an average typically makes the variance smaller, and the chosen weights \eqref{eq:sol_co} reduce the influence of the noisier estimates from the sparser layers.
For further discussions on weight selection, we refer to Section~\ref{sec_append:more_algorithm} in Supplementary Material.

\subsection{Algorithm: Multi-network histograms}
\label{subsec:algo_multinethist}

\begin{algorithm}[t]
\caption{Multi-network histogram in \texttt{nethist} \R{} package}
\label{alg:multi_nethist_combined_short}
\begin{algorithmic}[1]
    \STATE Input $L$ symmetric adjacency matrices of identical size $\mathcal{A} = (\bA^{(1)}, \ldots, \bA^{(L)})$, the number of groups $k$, a bandwidth $h$, an initial group label vector $\bz$, and a flag \texttt{homogeneous} (TRUE/FALSE).
    \STATE Compute estimates of layer-wise sparsity parameters, $\widehat{\rho}_n^{(\ell)} = \sum_{i<j} A_{ij}^{(\ell)} / \binom{n}{2}$ for $\ell\in [L]$.
    \STATE Update the group label vector using a greedy search algorithm.
    
    i) Randomly select two vertices and swap their group labels to form $\bz_{new}$ from $\bz$. 
    
    ii) Update $\bz \leftarrow \bz_{new}$ if $L(\mathcal{A},\bz_{new})> L(\mathcal{A},\bz)$ for $L(\mathcal{A},\bz)$ is the log-profile likelihood defined in \eqref{eq:optim_pl_multinethist} when \texttt{homogeneous} = FALSE, and \eqref{eq:ll_homogeneous_multinethist} when \texttt{homogeneous} = TRUE, for given $\bz$.
    \STATE Repeat Step 3 until either no changes occur during the pre-specified number of consecutive iterations or the maximum number of iterations is reached. Then set $\widehat{\bz} \leftarrow {\bz}$.
    \STATE Given $\widehat{\bz}$, compute the multi-network histogram estimates using \eqref{eq:multi_nethist} when \texttt{homogeneous} = FALSE, 
    and \eqref{eq:multinethist_homogenous} when \texttt{homogeneous} = TRUE.
\end{algorithmic}
\end{algorithm}

We provide an algorithm for both the multi-network histogram and the homogeneous multi-network histogram. Algorithm~\ref{alg:multi_nethist_combined_short} illustrates the procedure for identifying the optimal group labels and computing the corresponding block estimates for a given bandwidth.
Finding the group label vector that maximizes \eqref{eq:optim_pl_multinethist} or \eqref{eq:ll_homogeneous_multinethist} is a combinatorial optimization problem, making it infeasible to guarantee a global optimum within reasonable time. The multi-network histogram in \texttt{nethist} \R{} package is implemented using a greedy search with a data-driven initialization \cite{Rnethist}. Specifically, the group labels are initialized by performing spectral ordering of the nodes based on the node similarity matrix of the densest layer. The greedy search then iteratively permutes the labels of randomly selected vertex pairs, evaluating the profile likelihood to decide whether to update the group assignments.
This step is repeated until some stopping criterion is met, after which the resulting group labels $\widehat{\bz}$ are used to compute the multi-network histogram of the observed network.

Since Algorithm~\ref{alg:multi_nethist_combined_short} requires a bandwidth as input, we provide a data-driven bandwidth selection procedure in Algorithm~\ref{alg:bw_sel}. 
The bandwidth is chosen to minimize the upper bound of weighted mean integrated squared errors for multi-network histogram, following Theorems~\ref{thm:WMISE} and \ref{thm:MISE_homogeneous}.
Assuming the H\"{o}lder exponent $\alpha=1$, the optimal bandwidth depends on the sparsity $\rho_n^{(\ell)}$ and the layer-wise gradient magnitude $M_{\ell}$. 
The sparsity is estimated by $\widehat{\rho}_n^{(\ell)}$, and the gradients $M_{\ell}$ are estimated using layer-wise rank-1 graphon approximations based on degrees, following \cite{olhede2014network}. The details are outlined in Steps 2–4 of Algorithm~\ref{alg:bw_sel}.
Once these layer-wise estimates are obtained, the estimated bandwidth $\widehat{h}$ is a plug-in estimate of the theoretical optimal bandwidth.

\begin{algorithm}[!t]
\caption{Data-driven bandwidth selection for {multinethist}}
\label{alg:bw_sel}
\begin{algorithmic}[1]
    \STATE Input $L$ symmetric adjacency matrices of identical size $(\bA^{(1)}, \ldots, \bA^{(L)})$ and a flag \texttt{homogeneous} (TRUE/FALSE). 
    \STATE Compute degrees of $\bA^{(\ell)}$ for each $\ell$, then sort their entries to get an ordered degree vector $\bfd^{(\ell)}=(d_{1}^{(\ell)}, \ldots, d_{n}^{(\ell)})$. 
    \STATE Estimate the slope of $\bfd^{(\ell)}$ over indices $\lfloor n/2 \rfloor \pm \lfloor c\sqrt{n} \rfloor$ for some $c$, such as $\min(4,\sqrt{n}/8)$, by the least square method:
    $$
    (\widehat{m}^{(\ell)},\widehat{b}^{(\ell)}):=\argmin_{m,b} \sum_{j=\lfloor -c\sqrt{n}\rfloor}^{\lfloor c\sqrt{n}\rfloor}\left\{d_{\lfloor n/2 \rfloor + j}^{(\ell)} - (j m + b)\right\}^2.
    $$
    \STATE Estimate ${M}_{\ell}$ by a plug-in estimator,
    $$
    \widehat{M}_{\ell}:= \sqrt{2}n\left\{(\|\bfd^{(\ell)}\|_2^2)^{-}\right\}^2(\widehat{\rho}_n^{(\ell)})^{-}\{(\bfd^{(\ell)})^{\T}\bA^{(\ell)}\bfd^{(\ell)}\}\widehat{m}^{(\ell)}\widehat{b}^{(\ell)}.
    $$
    \STATE Use plug-in estimates of $h^*$ from \eqref{eq:oracle_bandwidth_WMISE} when \texttt{homogeneous}= FALSE, and from \eqref{eq:oracle_bandwidth_MISE_homogeneous} when \texttt{homogeneous}= TRUE. That is,
    $$
    \widehat{h} = \begin{cases}
    \sqrt{n}\left\{2L^{-1}\sum_{\ell=1}^L (\widehat{M}^{(\ell)})^2\widehat{\rho}_n^{(\ell)}\right\}^{-1/4}, & \text{\texttt{homogeneous}= FALSE},\\
    \sqrt{n}\left[2\{L^{-1}\sum_{\ell=1}^L (\widehat{M}^{(\ell)})^2\}\sum_{\ell=1}^L \widehat{\rho}_n^{(\ell)}\right]^{-1/4},& \text{\texttt{homogeneous}= TRUE}.
    \end{cases}
    $$
    Use its nearest integer greater than 1 as the bandwidth for multi-network histograms.
\end{algorithmic}
\end{algorithm}

\section{Theoretical study}
\label{sec:theory}

This section presents the theoretical properties of the multi-network histogram estimators proposed in \Cref{subsec:multinethist}. 
We establish an upper bound of the weighted mean integrated squared error (WMISE) for the multi-network histograms using the {oracle group labels}. The optimal bandwidth is then derived by minimizing the upper bound, and we show that it is largely influenced by denser layers. 
This section also discusses alternative bandwidth selection criteria based on the mean integrated squared error (MISE) of multi-network histograms and explains why they are less suitable than our proposed approach. 

\subsection{Main Results}

We begin by defining the oracle group labels, following \cite{olhede2014network}. For vertex latent variables $\xi_1,\ldots, \xi_n$ and bandwidth $h$, denote the oracle group label vector for a multi-network histogram with $k$ groups by $\widetilde{\bz}= (\widetilde{z}_1, \ldots, \widetilde{z}_n)^{\T}\in \mathcal{Z}_k$, where
\begin{equation}
    \label{def:oracle_z}
    \widetilde{z}_i = \min\left(\left\lceil\frac{\text{rank}(\xi_i)}{h}\right\rceil , k\right).
\end{equation}
The oracle labels defined in \eqref{def:oracle_z} represent the optimal grouping of vertices for multi-network histograms when the latent variables are known. Under the H\"{o}lder smoothness assumption, scaled graphon values at nearby latent positions do not differ much.
Consequently, the bin heights of multi-network histograms closely approximate the true function values, so these labels produce a high-quality grouping that achieves the best possible approximation. 

Using the oracle labels \eqref{def:oracle_z}, the $(a,b)$-th block probability matrix for the $\ell$th layer can be defined as
\begin{equation}
    \label{def:oracle_Theta}
    \begin{aligned}
        \left(\bar{A}_{ab}^{(\ell)}\right)^*(\widetilde{\bz}) 
&=       \frac{1}{h_{ab}^2}\sum_{(i,j)\in R_{ab}(\widetilde{\bz})} A_{ij}^{(\ell)}.
\end{aligned}
\end{equation}
Accordingly, the oracle estimator is defined as 
\begin{equation}
\label{eq:oracle_f}
\begin{aligned}(\widehat{f}^{(\ell)})^{*}(x,y;h)
=& ({\rho}_n^{(\ell)})^{-1}
\left(\bar{A}_{\min(\lceil \frac{nx}{h}\rceil,k)\min(\lceil \frac{ny}{h}\rceil,k)}^{(\ell)}\right)^*(\widetilde{\bz}), 
\end{aligned}
\end{equation}
where ${\rho}_n^{(\ell)}\in (0,1)$ for $\ell\in[L]$ and $x,y\in[0,1]$.

To evaluate the estimation performance of the oracle multi-network histogram in \eqref{eq:oracle_f}, we employ the WMISE, a standard criterion used in nonparametric statistics \cite{tsybakov2009Introducation}.
Denote the vector of layer-wise oracle graphon estimates by $\widehat{\bff}^* = \{(\widehat{f}^{(\ell)})^{*}\}_{\ell=1}^L$.
The WMISE of $\widehat{\bff}^*$ is defined as
\begin{equation}
\begin{aligned}
\label{eq:oracle_WMISE}
&\text{WMISE}(\widehat{\bff}^*) := 
\E \left(\inf_{\sigma} \iint_{[0,1]^2}\sum_{\ell=1}^L w_{\ell}^{*}\left|(\widehat{f}^{(\ell)})^*(\sigma(x),\sigma(y);h)-f^{(\ell)}(x,y)\right|^2 dxdy\right),
\end{aligned}
\end{equation}
where the weights $w_{\ell}^{*} = \rho_n^{(\ell)}/\sum_{l=1}^L \rho_n^{(l)}$ are proportional to the true layer-wise sparsity $\rho_n^{(\ell)}$ of scaled graphons, and the infimum is taken over all measure-preserving bijection maps $\sigma: [0,1] \to [0,1]$. These weights give more importance to denser layers, which contain more information about interaction patterns, and downweight sparser layers  in the total error.

We study the upper bound of WMISE under the following conditions.
\begin{condition}
\label{cond:multinethist}
For a scaled set of graphons $\{f^{(\ell)}(x,y)\}$ for $\ell \in [L]$ and its multi-network histogram, we assume:
\begin{enumerate}
\item[i)] the bandwidth $h\to\infty$ as $n\to\infty$, but slower than $n$,
\item[ii)] $f^{(\ell)}(x,y)$ is $\text{H\"{o}lder}^{\alpha}(M_{\ell})$ for $\alpha\in(0,1]$, meaning that for all $(x,y), (x',y')\in [0,1]^2$, it holds that $|f^{(\ell)}(x,y)- f^{(\ell)}(x',y')|\leq M_{\ell}\|(x,y)-(x',y')\|_2^{\alpha}$.
\end{enumerate}
\end{condition}
These conditions are widely used in network histogram-type estimators, see for example \cite{chandna2022Local,olhede2014network}.
The first condition ensures that the fitted multi-network histogram approximates the true function, while  
the second condition requires the true graphons to be smooth.

Theorem~\ref{thm:WMISE} provides an upper bound for the WMISE of the oracle multi-network histogram in \eqref{eq:oracle_f} and guides the choice of the optimal bandwidth. 
As a reference point, Theorem 1 in \cite{olhede2014network} is a special case of Theorem~\ref{thm:WMISE} with $L=1$ and $\alpha=1$.
All proofs of the theoretical results in this section are provided in Supplementary Material Section~\ref{sec_append:proofs_main}.

\begin{theorem}[WMISE of Multi-Network Histogram]
\label{thm:WMISE}
Assume that Condition~\ref{cond:multinethist} holds. Then, the WMISE defined in  \eqref{eq:oracle_WMISE} satisfies the following upper bound: 
\begin{align*}
\text{WMISE}(\widehat{\bff}^*) 
&\leq 
\frac{1}{\overline{\rho_n}}
\left\{
2^{\alpha}\left(\frac{h}{n}\right)^{2\alpha}\overline{M^2\rho_n}
+ 
\frac{2\overline{M^2\rho_n}}{(2n)^{\alpha}}
+ 
\frac{1}{h^2}
\right\}\{1+o(1)\},
\end{align*}
where $\overline{\rho_n} = L^{-1}\sum_{\ell=1}^{L} \rho_n^{(\ell)}$ and $\overline{M^2\rho_n} = L^{-1}\sum_{\ell=1}^{L} M_{\ell}^2\rho_n^{(\ell)}$. The right-hand side of the above inequality is minimized at 
\begin{equation}
\label{eq:oracle_bandwidth_WMISE}
h^* = n^{\alpha/(\alpha+1)}\left(\alpha2^{\alpha}\overline{M^2\rho_n}\right)^{-1/(2\alpha+2)}.
\end{equation}
Then, it holds that
\begin{equation}
\label{ineq:WMISE_oracle}
\begin{aligned}
\text{WMISE}(\widehat{\bff}^*)\Big|_{h=h^*}
&= O\left({\left\{\binom{n}{2}\overline{\rho_n}\right\}^{-\alpha/(\alpha+1)}}\right).
\end{aligned}
\end{equation}
\end{theorem}

The optimal bandwidth \eqref{eq:oracle_bandwidth_WMISE} is mainly determined by the sparsity level of the densest layer through the order of $\overline{M^2\rho_n}$ and is applied uniformly to all layers. 
This common bandwidth is used when a shared vertex partition is estimated jointly across all layers via profile log-likelihood \eqref{eq:optim_pl_multinethist}.
By leveraging information from other layers, the joint estimation improves the accuracy of label estimation and enables higher-resolution estimation even in sparser layers.

\begin{remark}
In practice, the bandwidth can be computed with $\alpha=1$, following \cite{olhede2014network}. Setting $\alpha=1$ corresponds to the Lipschitz condition, under which $M_\ell$ bounds the maximum gradient of the graphon. In this case, $M_\ell$ can be estimated in a relatively simple manner, as outlined in Algorithm~\ref{alg:bw_sel}.
\end{remark}

Next, we derive an upper bound of the homogeneous multi-network histograms \eqref{eq:multinethist_homogenous} with the oracle labels, given by
\begin{equation}
\label{eq:oracle_f_homo}
\begin{aligned}
\widehat{f}^{*}(x,y;h) := 
\frac{\sum_{\ell=1}^L \left(\bar{A}_{\min(\lceil nx/h\rceil,k)\min(\lceil ny/h\rceil,k)}^{(\ell)}\right)^*}{\sum_{\ell=1}^L {\rho}_n^{(\ell)}}. 
\end{aligned}
\end{equation}
Since all layers share the same estimate $\widehat{f}^*$ and $\sum_{\ell=1}^L w_{\ell}^*=1$, the WMISE of $\widehat{\bff}^*=(\widehat{f}^*,\ldots, \widehat{f}^*)$ simplifies to 
\begin{align*}
\text{WMISE}(\widehat{\bff}^*) &= \E \left(\inf_{\sigma} \iint_{[0,1]^2}\sum_{\ell=1}^L w_{\ell}^* \left|\widehat{f}^*(\sigma(x),\sigma(y);h)-f(x,y)\right|^2 dxdy\right)\\
&= \E \left(\inf_{\sigma} \iint_{[0,1]^2} \left|\widehat{f}^*(\sigma(x),\sigma(y);h)-f(x,y)\right|^2 dxdy\right).%\\
% &= \text{MISE}(\widehat{f}^*)
\end{align*}
Therefore, 
analyzing the WMISE of $\widehat{\bff}^*$ is equivalent to
finding an upper bound of the following MISE of $\widehat{f}^*$:
\begin{equation}
    \label{eq:oracle_MISE_homo}
    \begin{aligned}
    &\text{MISE}(\widehat{f}^*) := \E \left(\inf_{\sigma} \iint_{[0,1]^2}\left|\widehat{f}^*(\sigma(x),\sigma(y);h)-f(x,y)\right|^2 dxdy\right).
    \end{aligned}
\end{equation}
The following theorem presents an upper bound for the MISE and the optimal bandwidth when all layers share the function $f$, while sparsity levels differ across layers.

\begin{theorem}[MISE of Homogeneous Multi-Network Histogram]
\label{thm:MISE_homogeneous}
Assume that Condition~\ref{cond:multinethist} i) holds, and that the common function $f$ is $\text{H\"{o}lder}^{\alpha}(M)$ for $\alpha\in(0,1]$. Then, the MISE defined in \eqref{eq:oracle_MISE_homo} satisfies the following bound:
\begin{align*}
\text{MISE}(\widehat{f}^*) \leq 
M^2\left\{2^{\alpha}\left(\frac{h}{n}\right)^{2\alpha} + \frac{2}{(2n)^{\alpha}} + \frac{1}{M^2 L\overline{\rho_n}h^2}\right\}\left\{1+o(1)\right\},
\end{align*}
where $\overline{\rho_n} = L^{-1}\sum_{\ell=1}^L \rho_n^{(\ell)}$. The right-hand side of the above inequality is minimized at 
\begin{equation}
\label{eq:oracle_bandwidth_MISE_homogeneous}
h^* = n^{\alpha/(\alpha+1)}\left({\alpha 2^{\alpha}M^2 L\overline{\rho_n}}\right)^{-1/(2\alpha+2)}.
\end{equation}
Consequently, it holds that
\begin{equation}
\label{ineq:MISE_oracle_homogeneous}
\begin{aligned}
\text{MISE}(\widehat{f}^*)\Big|_{h=h^*} &= O\left({\left\{\binom{n}{2}L\overline{\rho_n}\right\}^{-\alpha/(\alpha+1)}}\right).
\end{aligned}
\end{equation}
\end{theorem}

Comparing the bandwidths in \eqref{eq:oracle_bandwidth_WMISE} and \eqref{eq:oracle_bandwidth_MISE_homogeneous}, the homogeneous multi-network histogram yields a narrower optimal bandwidth. 
When the homogeneous structure holds, it estimates a single function $f$ using $L$ layers, whereas the multi-network histogram produces $L$ distinct function estimates. 
Pooling information across layers effectively increases the amount of data available for estimating $f$, allowing a more precise estimation through a finer division of the domain.
Furthermore, as the number of layers increases, the upper bound in \eqref{ineq:MISE_oracle_homogeneous} decreases, indicating that the homogeneous approach can capture finer structural details. 

\subsection{Mean Integrated Squared Errors}
\label{subsec:MISE}

We consider an alternative bandwidth selection criterion based on the MISE of the oracle multi-network histogram in \eqref{eq:oracle_f}.
The MISE is defined as
\begin{equation}
\label{eq:oracle_MISE}
\begin{aligned}
&\text{MISE}(\widehat{\bff}^*) := 
\E \left(\inf_{\sigma} \iint_{[0,1]^2}\frac{1}{L}\sum_{\ell=1}^L 
\left|(\widehat{f}^{(\ell)})^*(\sigma(x),\sigma(y);h)-f^{(\ell)}(x,y)\right|^2 dxdy\right).
\end{aligned}
\end{equation}
Proposition~\ref{prop:MISE} establishes an upper bound for the MISE and the optimal bandwidth that minimizes this bound.
\begin{proposition}[MISE of Multi-Network Histogram]
\label{prop:MISE}
Assume that Condition~\ref{cond:multinethist} holds. Then, the MISE of the oracle multi-network histogram in \eqref{eq:oracle_f} satisfies the following upper bound:
\begin{equation*}
\label{ineq:MISE_oracle}
    \begin{aligned}
     &\text{MISE}(\widehat{\bff}^*) 
\leq  
\left\{
\overline{M^2}2^{\alpha}\left(\frac{h}{n}\right)^{2\alpha}
+ \frac{2\overline{M^2}}{(2n)^{\alpha}}
+ \frac{\overline{\rho_n^{-1}} }{h^2}
\right\}\{1+o(1)\},
\end{aligned}
\end{equation*}
where $\overline{M^2} = L^{-1}\sum_{\ell=1}^L M_{\ell}^2$ and $\overline{\rho_n^{-1}} = L^{-1}\sum_{\ell=1}^L(\rho_n^{(\ell)})^{-1}$.
Then, the right-hand side of the above inequality is minimized at 
\begin{equation}
\label{eq:oracle_bandwidth_MISE}
h_{\text{MISE}}^* = n^{\alpha/(\alpha+1)}\left(\frac{\overline{\rho_n^{-1}}}{\alpha 2^{\alpha} \overline{M^2}}\right)^{1/(2\alpha+2)}.
\end{equation}
Lastly, $\text{MISE}(\widehat{\bff}^*)|_{h=h_{\text{MISE}}^*}$ satisfies
\begin{equation}
\label{ineq:MISE_optimal}
\begin{aligned}
\text{MISE}(\widehat{\bff}^*)\Big|_{h=h_{\text{MISE}}^*}
&=  
O\left(\left\{\binom{n}{2}\frac{1}{\overline{\rho_n^{-1}}}\right\}^{-\alpha/(\alpha+1)}\right).
\end{aligned}
\end{equation}
\end{proposition}
The rate of the optimal bandwidth in \eqref{eq:oracle_bandwidth_MISE} depends on the harmonic mean of the layer-wise sparsity, so it is largely determined by the sparsity level of the sparsest layer. When the sparsity parameters vary, the selected bandwidth may be too wide for the denser layers, potentially obscuring important structural details. 
In contrast, the WMISE-based bandwidth is generally smaller than the MISE-based bandwidth because it downweights the influence of extremely sparse layers in bandwidth selection. Consequently, when sparse layers are present, the WMISE-based bandwidth retains a finer resolution for the denser layers.
In addition, as shown in Section~\ref{subsec:MISE_WMISE_bandwidth_comp}, the dense-layer estimation error under the WMISE-based bandwidth changes only mildly as sparse layers are added, whereas the MISE-based bandwidth produces a much larger and less predictable effect on dense-layer estimation across graphons.

\section{Simulation study}
\label{sec:sim_study}

We investigate the finite sample performance of the proposed multi-network histogram methods. Specifically, we examine the effect of the number of layers $L$ and the number of vertices $n$ on the estimation accuracy of the proposed methods, and compare them with existing single-layer graphon estimation methods.

\subsection{Setup}
\label{subsec:data_generation}

We begin by describing evaluation criterion and simulation settings. We then explain the procedure for generating synthetic multiplex networks, including the choice of functions, layer specification scenarios, and sparsity levels. 

To evaluate estimation performance, we use the weighted mean squared error (WMSE) of $\widehat{\bff} = \{\widehat{f}^{(1)},\ldots, \widehat{f}^{(L)}\}$, defined as
\begin{equation}
\label{eq:WMSE_global}
    \text{WMSE}(\widehat{\bff}) = \sum_{\ell=1}^L \frac{\rho_n^{(\ell)}}{\sum_{l=1}^L \rho_n^{(l)}} \binom{n}{2}^{-1}\sum_{i<j}\left\{\widehat{f}_{ij}^{(\ell)}-f^{(\ell)}(\xi_i,\xi_j)\right\}^2,
\end{equation}
where $\widehat{f}_{ij}^{(\ell)}$ is the estimated function value for the vertex pair $(i,j)$ in the $\ell$th layer, and $f^{(\ell)}(\xi_i,\xi_j)$ is the true function value for that pair. This error serves as the empirical counterpart of WMISE defined in \eqref{eq:oracle_WMISE}.

We investigate how the number of layers $L$ and the number of vertices $n$ affect the WMSE under the following settings. 
\begin{enumerate}
\item Effect of $L$: We change the number of layers $L\in\{5,7,10\}$ while keeping $n=400$ across all scenarios and sparsity levels. 
\item Effect of $n$: We vary the number of vertices $n\in\{200, 400,800\}$ while fixing $L=7$ across all scenarios and sparsity levels. 
\end{enumerate}
In the first setting, the proposed methods are expected to perform better with larger $L$, which demonstrates their advantage over single-layer graphon estimators. In the second setting, the WMSE is expected to decrease as the network size $n$ increases, as seen in the theoretical results in \Cref{sec:theory}.
\begin{table}[t]
\caption{The choice of function $f_k(x,y)$ for $(x,y)\in[0,1]^2$ and the ranges of sparsity parameters $\rho_n^{(\ell)}$ under three settings: i) layers can be either sparse or dense (mixed); ii) all layers are dense (all dense); and iii) all layers are sparse (all sparse). The function $sig(x) = 1/(1+e^{-x})$ is the sigmoid function.}
\label{tab:sparsity_level_setup}
\centering
\begin{tabular}{ll|lll}
& & \multicolumn{3}{c}{Ranges of $\rho_n^{(\ell)}$} \\
 ID ($k$) & $f_k(x,y)$   & mixed  & all dense      & all sparse            \\ \hline
1 & $0.75+2.25x^2 y^2$ & $[0.5n^{-1/2}, 0.275]$ & $[0.2, 0.275]$ & $[0.5n^{-1/2}, 2n^{-1/2}]$  \\
2 & $sig(x+y)/0.7238$ & $[0.5n^{-1/2}, 0.7]$   & $[0.4, 0.7]$   & $[0.5n^{-1/2}, 2n^{-1/2}]$  \\
3 & $\exp({-0.5|x-y|)/0.8522}$ & $[0.5n^{-1/2}, 0.7]$  & $[0.4, 0.7]$   & $[0.5n^{-1/2}, 2n^{-1/2}]$   \\
4 & $1 + 4(x-0.5)(y-0.5)$  & $[0.5n^{-1/2}, 0.45]$ & $[0.2, 0.45]$  & $[0.5n^{-1/2}, 2n^{-1/2}]$\\ \hline
\end{tabular}
\end{table}
\begin{figure}[t]
    \centering
     \begin{subfigure}[t]{0.225\linewidth}
         \centering
\includegraphics[width=\textwidth,trim={0.5cm 1cm 2cm 1cm},clip]{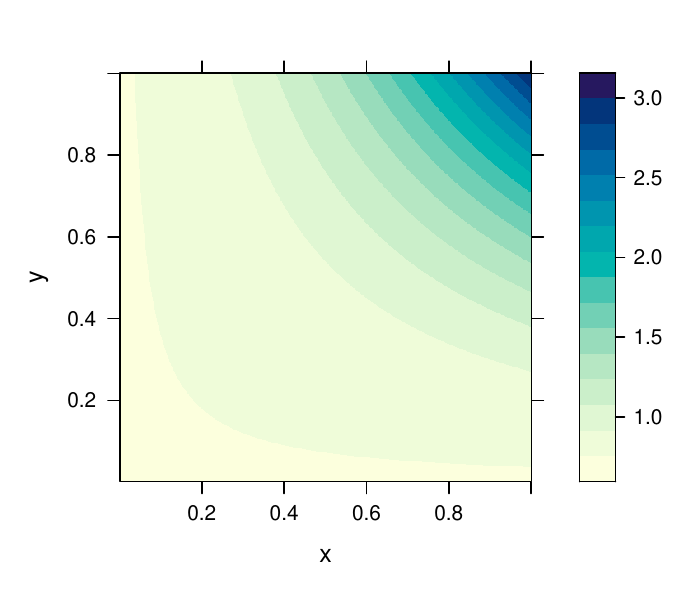}
    \caption{$f_1(x,y)$}
    \label{fig:graphon1}
    \end{subfigure}
    \hfill
     \begin{subfigure}[t]{0.225\linewidth}
         \centering
    \includegraphics[width=\textwidth, trim = {0.5cm 1cm 2cm 1cm},clip]{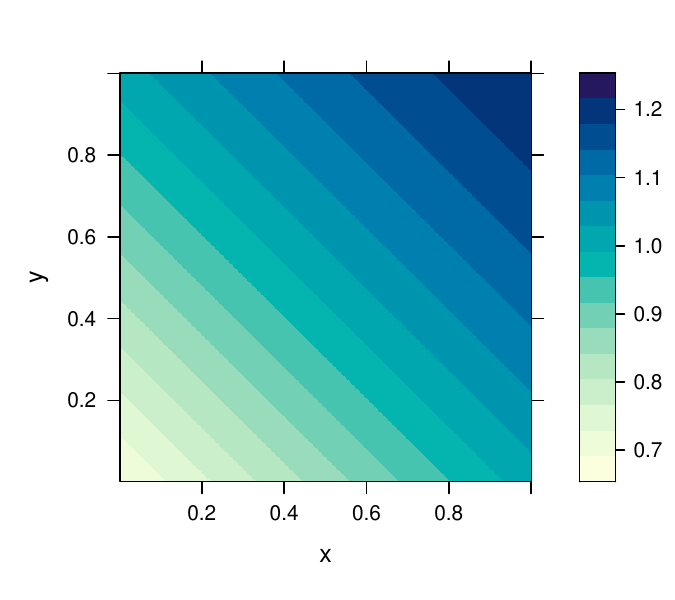}
   \caption{$f_2(x,y)$}  \label{fig:graphon2}
    \end{subfigure}
    % \newline
    \hfill
     \begin{subfigure}[t]{0.225\linewidth}
         \centering
   \includegraphics[width=\textwidth,trim={0.5cm 1cm 2cm 1cm},clip]{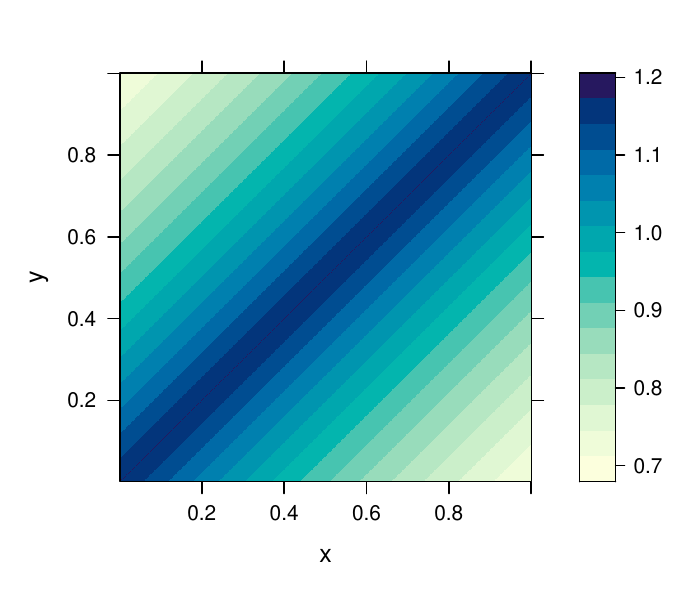}
    \caption{$f_3(x,y)$}
    \label{fig:graphon3}
    \end{subfigure}
    \hfill
    \begin{subfigure}[t]{0.225\linewidth}
         \centering
\includegraphics[width=\textwidth,trim={0.5cm 1cm 2cm 1cm},clip]{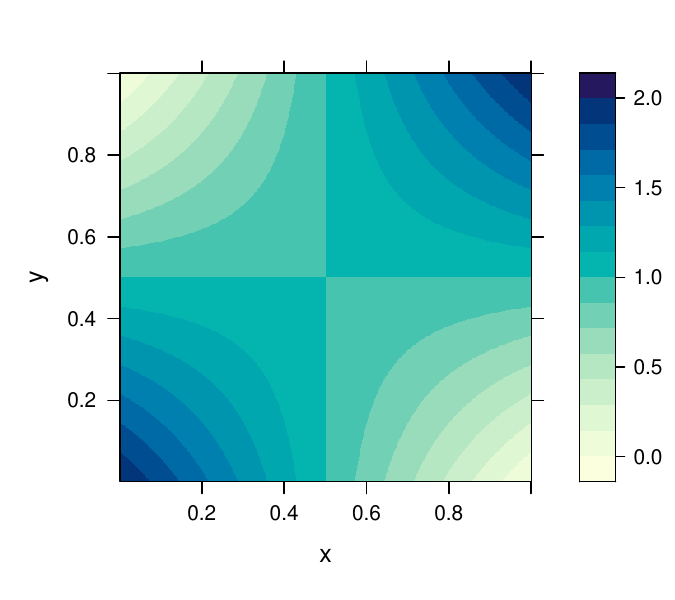}   
    \caption{$f_4(x,y)$ }
    \label{fig:graphon4}
    \end{subfigure}
    \caption{Base functions used to construct scaled graphons in simulation studies. Darker colour indicates larger function values. }
    \label{fig:graphons_sim}
\end{figure}

To generate synthetic multiplex networks with common latent variables, we follow the procedure described in Example~\ref{ex:scaled_multigraphon} under three  layer-structure scenarios: Homogeneous, Perturbation, and Heterogeneous. For each scenario, a scaled set of graphons is determined by the choice of $\rho_n^{(\ell)}$ and $f^{(\ell)}$. 
The function term $f^{(\ell)}$ is constructed from four base functions $\{f_k: \iint_{[0,1]^2} f_k(x,y)dxdy \approx 1\}_{k=1}^4$ as summarized in Table~\ref{tab:sparsity_level_setup} and Figure~\ref{fig:graphons_sim}. The layer-wise sparsity parameters $\rho_n^{(\ell)}$ are chosen under three sparsity settings as given in Table~\ref{tab:sparsity_level_setup}: mixed (some dense, some sparse), all layers dense, or all layers sparse. The base functions $f_1$, $f_2$, and $f_3$ are considered in \cite{chandna2022Local}. The function $f_4$ has identical degrees across all values of the latent variable, since $\int_{0}^1 f_4(x,y)dy=1$ for all $x\in[0,1]$. 
Similarly, $f_3$ has only minor variation in degree across $x$. With those functions and sparsity parameters, each scenario is constructed as follows:
\vspace{0.1cm}
\begin{enumerate}
    \item (Scenario 1: Homogeneous) All layers share a common function $f(x,y)$ selected from $\{f_k\}_{k=1}^4$. The sparsity parameters $\rho_n^{(1)}$ and $\rho_n^{(L)}$ are set to be the minimum and maximum values of the corresponding range in Table~\ref{tab:sparsity_level_setup} for the given sparsity setting, respectively. The sparsity parameters of intermediate layers are interpolated linearly, $\rho_n^{(\ell)} := \left(1-\frac{\ell-1}{L-1}\right) \rho_n^{(1)} + \frac{\ell-1}{L-1}\rho_n^{(L)}$, for $\ell=2,\ldots, L-1.$
    \item (Scenario 2: Perturbation) Two functions, $f_{k_1}$ and $f_{k_2}$, are chosen such that $(k_1,k_2)\in\{(1,2), (1,4), (3,4)\}$. The first and last layers use $f^{(1)} = f_{k_1}$ and $f^{(L)} = f_{k_2}$ with sparsity parameters at the minimum and maximum of the respective ranges: $\rho_n^{(1)}$ is the minimum value of the range for $f_{k_1}$, and $\rho_n^{(L)}$ is the maximum of the range for $f_{k_2}$, as given in Table~\ref{tab:sparsity_level_setup}. Intermediate layers are obtained by a linear interpolation,
    \begin{align*}
    \rho_n^{(\ell)}f^{(\ell)} = \left(1-\frac{\ell-1}{L-1}\right) \rho_n^{(1)}f^{(1)} + \frac{\ell-1}{L-1}\rho_n^{(L)}f^{(L)},
    \end{align*}
    for $\ell=2,\ldots, L-1.$ 
    \item (Scenario 3: Heterogeneous) Two functions, $f_{k_1}$ and $f_{k_2}$, are chosen such that $(k_1,k_2)\in\{(1,2), (1,4), (3,4)\}$. Layers are divided into two groups: $f^{(\ell)} = f_{k_1}$ for $\ell=1,\ldots, \lceil L/2\rceil$, and $f^{(\ell)} = f_{k_2}$ for $\ell = \lceil L/2\rceil +1,\ldots,L$.
    Each group uses its own sparsity range specified in Table~\ref{tab:sparsity_level_setup}. In the mixed setting, the first and second groups use the all sparse and all dense ranges, respectively. The sparsity parameters are obtained by linearly interpolating each group's range over $L$ layers and taking the values corresponding to the layers in that group.
    
\end{enumerate}
For each layer, the edge probability between vertices $i$ and $j$ is defined as $p_{ij}^{(\ell)} = \rho_n^{(\ell)}f^{(\ell)}(\xi_i,\xi_j) \in [0,1]$, where the latent variables are $\xi_1,\ldots, \xi_n \sim U(0,1)$. Then, the adjacency matrix entries are sampled from Bernoulli distributions,  $A_{ij}^{(\ell)}\sim \text{Bernoulli}(p_{ij}^{(\ell)})$ for $i<j$, and are symmetrized by setting $A_{ji}^{(\ell)} = A_{ij}^{(\ell)}$. 

For all scenarios, we compare WMSE in \eqref{eq:WMSE_global} of multi-network histogram (mnhist) to that of the following single-layer graphon estimation methods: network histogram (nethist) \cite{olhede2014network}, sort-and-smooth (SAS) \cite{chan2014Consistent}, universal singular value thresholding (USVT) \cite{chatterjee2015Matrix}, and neighbourhood smoothing (NBS) \cite{zhang2017Estimating}. The homogeneous multi-network histogram (h-mnhist) is additionally evaluated in the Homogeneous scenario.

\subsection{Simulation results}
\label{subsec:sim_res}

%Effect of L
\begin{table}[t]
\caption{Comparisons of WMSE ($\times 100$ with the standard deviation in parentheses) averaged over 100 replications across different number of layers $L\in\{5,7,10\}$ for multiplex networks with 400 vertices under Scenario 1 (Homogeneous) in the mixed sparsity case. Function IDs correspond to those in Table~\ref{tab:sparsity_level_setup}. Underlined methods indicate the proposed methods, and bolded values denote the smallest average WMSE in each case.}
\label{tab:homogeneous_mixed_by_L}
\centering
\begin{tabular}{ll|rrr}
\hline
ID & Method & 5 & 7 & 10\\
\hline
1 & \underline{mnhist} & 3.537 (0.227) & 2.888 (0.189) & 2.412 (0.159)\\
1 & \underline{h-mnhist} & \textbf{2.798 (0.150)} & \textbf{2.100 (0.109)} & \textbf{1.523 (0.094)}\\
1 & nethist & 11.439 (0.273) & 11.465 (0.279) & 11.523 (0.207)\\
1 & SAS & 3.197 (0.140) & 3.238 (0.123) & 3.273 (0.104)\\
1 & USVT & 10.398 (1.175) & 9.147 (1.032) & 8.573 (0.836)\\
1 & NBS & 17.988 (0.878) & 18.517 (1.085) & 19.246 (0.973)\\
\hline
2 & \underline{mnhist} & 0.830 (0.058) & 0.676 (0.045) & 0.558 (0.041)\\
2 & \underline{h-mnhist} & 0.702 (0.037) & \textbf{0.526 (0.028)} & \textbf{0.397 (0.024)}\\
2 & nethist & 3.043 (0.113) & 3.100 (0.080) & 3.143 (0.074)\\
2 & SAS & \textbf{0.553 (0.026)} & 0.579 (0.022) & 0.603 (0.021)\\
2 & USVT & 1.787 (0.399) & 1.502 (0.287) & 1.351 (0.227)\\
2 & NBS & 5.093 (0.267) & 4.473 (0.221) & 4.741 (0.168)\\
\hline
3 & \underline{mnhist} & \textbf{0.656 (0.223)} & 0.491 (0.156) & 0.370 (0.025)\\
3 & \underline{h-mnhist} & 0.669 (0.042) & \textbf{0.484 (0.029)} & \textbf{0.342 (0.022)}\\
3 & nethist & 3.246 (0.129) & 3.353 (0.104) & 3.402 (0.084)\\
3 & SAS & 1.544 (0.068) & 1.583 (0.062) & 1.589 (0.058)\\
3 & USVT & 2.666 (0.450) & 2.486 (0.299) & 2.313 (0.251)\\
3 & NBS & 5.628 (0.345) & 5.114 (0.208) & 5.385 (0.172)\\
\hline
4 & \underline{mnhist} & \textbf{0.881 (0.060)} & 0.754 (0.049) & 0.664 (0.042)\\
4 & \underline{h-mnhist} & 0.989 (0.087) & \textbf{0.706 (0.053)} & \textbf{0.502 (0.033)}\\
4 & nethist & 4.636 (0.199) & 4.818 (0.266) & 4.871 (0.180)\\
4 & SAS & 11.876 (0.960) & 11.859 (1.056) & 11.703 (1.023)\\
4 & USVT & 4.194 (0.580) & 3.875 (0.443) & 3.811 (0.311)\\
4 & NBS & 9.400 (0.551) & 9.611 (0.353) & 10.725 (0.423)\\
\hline
\end{tabular}
\end{table}

\begin{table}[!ht]
\caption{Comparisons of WMSE ($\times 100$ with the standard deviation in parentheses) averaged over 100 replications by the number of layers, $L\in \{5,7,10\}$, for multiplex networks with 400 vertices under Scenario 2 (Perturbation) in the mixed sparsity case. Function IDs correspond to those in Table~\ref{tab:sparsity_level_setup}. Underlined methods indicate the proposed methods, and bolded values denote the smallest average WMSE in each case.}
\label{tab:Perturbed_mixed_by_L}
\centering
\begin{tabular}{lll|rrr}
\hline
ID1 & ID2 & Method & 5 & 7 & 10\\
\hline
1 & 2 & \underline{mnhist} &0.872 (0.052) & 0.724 (0.041) & \textbf{0.589 (0.039)}\\
1 & 2 & nethist & 3.212 (0.099) & 3.246 (0.086) & 3.258 (0.070)\\
1 & 2 & SAS & \textbf{0.659 (0.025)} & \textbf{0.661 (0.022)} & 0.664 (0.020)\\
1 & 2 & USVT & 3.013 (0.517) & 2.466 (0.361) & 1.949 (0.265)\\
1 & 2 & NBS & 5.718 (0.336) & 4.982 (0.271) & 5.018 (0.184)\\
\hline
1 & 4 & \underline{mnhist} &\textbf{0.919 (0.076)} & \textbf{0.774 (0.054)} & \textbf{0.667 (0.038)}\\
1 & 4 & nethist & 4.725 (0.254) & 4.846 (0.209) & 4.821 (0.202)\\
1 & 4 & SAS & 11.389 (1.078) & 11.333 (1.061) & 11.451 (0.952)\\
1 & 4 & USVT & 5.997 (0.682) & 5.229 (0.526) & 4.776 (0.360)\\
1 & 4 & NBS & 10.330 (0.594) & 10.269 (0.416) & 11.003 (0.517)\\
\hline
3 & 4 & \underline{mnhist} &\textbf{0.876 (0.066)} & \textbf{0.743 (0.049)} & \textbf{0.649 (0.038)}\\
3 & 4 & nethist & 4.614 (0.312) & 4.790 (0.181) & 4.811 (0.204)\\
3 & 4 & SAS & 11.086 (1.060) & 11.031 (0.990) & 11.156 (1.020)\\
3 & 4 & USVT & 3.946 (0.643) & 3.803 (0.418) & 3.678 (0.323)\\
3 & 4 & NBS & 9.256 (0.535) & 9.545 (0.420) & 11.016 (0.390)\\
\hline
\end{tabular}
\end{table}

\begin{table}[!ht]
\caption{Comparisons of WMSE ($\times 100$ with the standard deviation in parentheses) averaged over 100 replications by the number of layers, $L\in \{5,7,10\}$, for multiplex networks with 400 vertices under Scenario 3 (Heterogeneous) in the mixed sparsity case. Function IDs correspond to those in Table~\ref{tab:sparsity_level_setup}. Underlined methods indicate the proposed methods, and bolded values denote the smallest average WMSE in each case.}
\label{tab:Heterogeneous_mixed_by_L}
\centering
\begin{tabular}{lll|rrr}
\hline
ID1 & ID2 & Method & 5 & 7 & 10\\
\hline
1 & 2 & \underline{mnhist} &1.241 (0.082) & \textbf{0.958 (0.060)} & \textbf{0.671 (0.039)}\\
1 & 2 & nethist & 4.444 (0.140) & 4.203 (0.129) & 3.596 (0.101)\\
1 & 2 & SAS & \textbf{1.012 (0.045)} & 0.968 (0.043) & 0.819 (0.033)\\
1 & 2 & USVT & 6.023 (0.962) & 5.391 (0.760) & 4.439 (0.556)\\
1 & 2 & NBS & 15.240 (0.952) & 14.745 (0.678) & 12.605 (0.530)\\
\hline
1 & 4 & \underline{mnhist} &\textbf{1.275 (0.110)} & \textbf{0.992 (0.070)} & \textbf{0.750 (0.039)}\\
1 & 4 & nethist & 6.549 (0.238) & 6.216 (0.183) & 5.391 (0.138)\\
1 & 4 & SAS & 11.041 (0.950) & 11.053 (0.879) & 11.109 (0.866)\\
1 & 4 & USVT & 9.644 (1.501) & 8.878 (1.178) & 7.751 (0.823)\\
1 & 4 & NBS & 24.038 (1.381) & 23.401 (1.164) & 20.316 (0.784)\\
\hline
3 & 4 & \underline{mnhist} &\textbf{1.085 (0.093)} & \textbf{0.846 (0.058)} & \textbf{0.674 (0.042)}\\
3 & 4 & nethist & 5.667 (0.196) & 5.390 (0.185) & 4.768 (0.129)\\
3 & 4 & SAS & 10.286 (0.856) & 10.668 (0.862) & 10.697 (0.888)\\
3 & 4 & USVT & 4.381 (0.926) & 4.030 (0.614) & 3.505 (0.437)\\
3 & 4 & NBS & 21.421 (1.162) & 23.602 (1.030) & 20.710 (0.883)\\
\hline
\end{tabular}
\end{table}

We present the averages and standard deviations of WMSE from 100 replications of two proposed methods (mnhist and h-mnhist) and four competing methods (nethist, SAS, USVT, and NBS). Tables~\ref{tab:homogeneous_mixed_by_L} to \ref{tab:Heterogeneous_mixed_by_L} display the effect of $L$, and Tables~\ref{tab:homogeneous_mixed_by_n} to~\ref{tab:Heterogeneous_mixed_by_n} report the effect of $n$ under the mixed sparsity case. For other sparsity cases, the results are summarized in Tables~\ref{tab:homogeneous_all_dense_by_L} to~\ref{tab:Heterogeneous_all_dense_by_n} for all dense case and Tables~\ref{tab:homogeneous_all_sparse_by_L} to~\ref{tab:Heterogeneous_all_sparse_by_n} for all sparse case in Supplementary Material. In addition, Tables~\ref{tab:layerwise_homogeneous}--\ref{tab:layerwise_heterogeneous} provide layer-wise MSE results obtained from Tables~\ref{tab:homogeneous_mixed_by_L} to \ref{tab:Heterogeneous_mixed_by_L}.

As shown in these tables, the proposed methods generally achieve lower WMSE values across all scenarios and sparsity levels compared with the competing methods.  This is observed regardless of the number of layers $L$, the number of vertices $n$, and the choice of base functions. 
In particular, even in settings involving $f_3$ or $f_4$, where the degree of $f_3$ does not vary much, and $f_4$ has the identical degrees, both mnhist and h-mnhist maintain low WMSE values consistently. This indicates that our profile likelihood approach leads to better estimation performance by allowing vertex grouping based on connection patterns, rather than relying solely on vertex degrees. 
Additionally, h-mnhist, which is designed for homogeneous layers, outperforms mnhist in Scenario 1 where the underlying layer-wise functions are identical across all layers, under all sparsity level settings. 
Lastly, the layer-wise results show that the proposed methods have smaller MSE in the sparsest layer relative to the competing methods. This suggests that sharing information across layers is particularly effective in graphon estimation when individual layers contain fewer edges.

The estimation errors of the proposed methods decrease as the number of layers $L$ increases, whereas the competing methods do not show this behaviour.
This is consistently observed under all scenarios in the mixed sparsity setting, as shown in Tables~\ref{tab:homogeneous_mixed_by_L} to \ref{tab:Heterogeneous_mixed_by_L}, in the all dense setting, as reported in Tables~\ref{tab:homogeneous_all_dense_by_L} to \ref{tab:Heterogeneous_all_dense_by_L}, and in the all sparse setting, as in Tables~\ref{tab:homogeneous_all_sparse_by_L} to \ref{tab:Heterogeneous_all_sparse_by_L}.
This is because, as $L$ increases, the proposed joint estimation methods with common latent variables improve group assignments for blockmodel approximations and consequently lead to more accurate estimates. 
However, the competing methods rely only on the marginal structure of each layer and do not use interaction information from other layers, which explains why their WMSEs do not decrease as $L$ increases.
% 서플먼터리 매터리얼은 레이어가 커짐에 따라 추정개선되는 소스가 레이블 추정인지 밴드위스 선택인지 분석하였다. 그 결과, 추정 오차의 감소는 bandwidth selection보다는 label estimation 개선의 영향이 큰 것으로 나타났다.
Supplementary Material Section~\ref{subsec:bandwidth_select} investigates the source of decrease in estimation error as the number of layers increases by separating the effects of label estimation and bandwidth selection. The results suggest that the reduction in estimation error is mainly due to improved label estimation.

We also observe that mnhist and h-mnhist exhibit smaller WMSE values as the number of vertices $n$ increases across all scenarios and sparsity cases. In particular, the decay rate of WMSE appears to depend on both $n$ and $\rho_n^{(\ell)}$, which is consistent with Theorems~\ref{thm:WMISE} and \ref{thm:MISE_homogeneous}. For example, in the mixed and all dense sparsity cases, the errors are expected to decrease at the rate of $n^{-1}$, as the average sparsity $\overline{\rho_n}$ is of order $O(1)$. This theoretical rate is observed in our results in Tables~\ref{tab:homogeneous_mixed_by_n} to \ref{tab:Heterogeneous_mixed_by_n} for the mixed sparsity case and Tables~\ref{tab:homogeneous_all_dense_by_n} to \ref{tab:Heterogeneous_all_dense_by_n} for the all dense sparsity case.
For the sparse case, the WMSE appears to decrease at a rate close to $n^{-1/2}$, as reported in Tables~\ref{tab:homogeneous_all_sparse_by_n} to \ref{tab:Heterogeneous_all_sparse_by_n}, which is slower than the theoretical decay rate $n^{-3/4}$ given that $\overline{\rho_n}$ is of order $O({n}^{-1/2})$ as seen in Table~\ref{tab:sparsity_level_setup}.

\section{Data Analysis: Indian village networks}
\label{sec:real_data}

We explore household networks from the Indian village data \cite{banerjee2013Diffusion} which consists of 12 social relationships across 75 villages. Our analysis focuses on Village 40, which was selected for its sufficient number of vertices for nonparametric estimation. For this analysis, we treat all layers as undirected simple networks (see Figure~\ref{fig:indian_village_vil_40}). After removing isolated households with zero degrees in all layers, the resulting multiplex network contains 231 households with layer-wise edge densities ranging from 0.0021 to 0.0198. For detailed descriptions of the layers and additional network summary statistics, see \cite{banerjee2013Diffusion} or Supplementary Material Section~\ref{sec_append:indian_vils}.
\begin{figure}[!ht]
    \centering
   \begin{subfigure}[b]{0.225\linewidth}
    \centering
    \includegraphics[width=\textwidth, trim = {1cm, 1.5cm, 1cm, 1cm}]{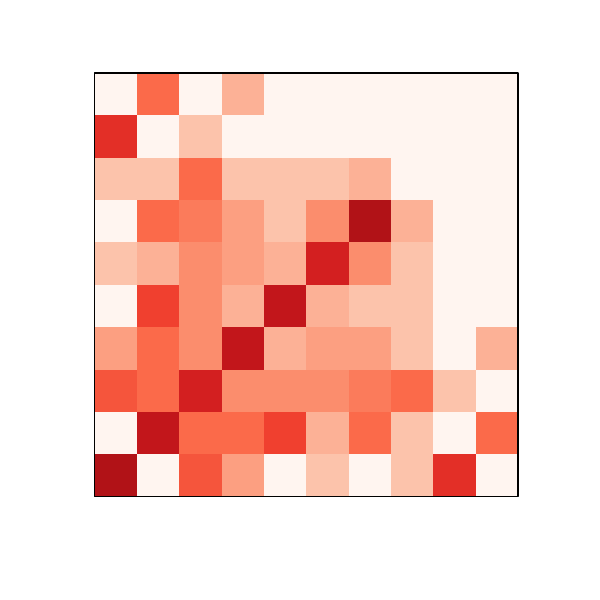}
    \caption{Borrow money}
\end{subfigure}
\hfill
\begin{subfigure}[b]{0.225\linewidth}
    \centering
    \includegraphics[width=\textwidth, trim = {1cm, 1.5cm, 1cm, 1cm}]{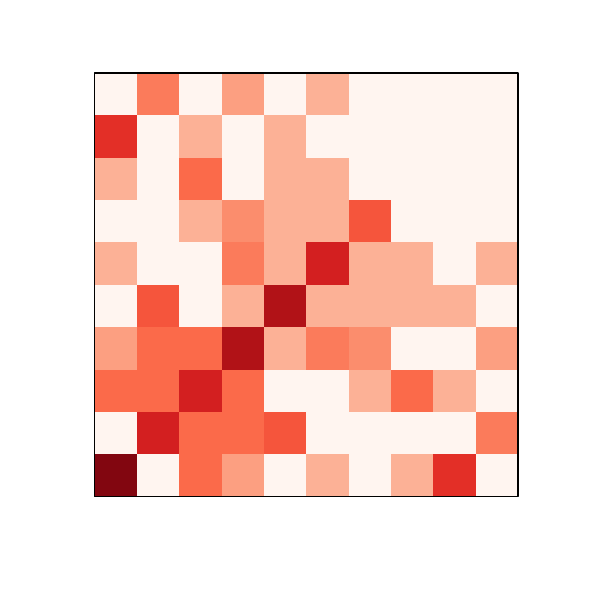}
    \caption{Give advice}
\end{subfigure}
\hfill
\begin{subfigure}[b]{0.225\linewidth}
    \centering
    \includegraphics[width=\textwidth, trim = {1cm, 1.5cm, 1cm, 1cm}]{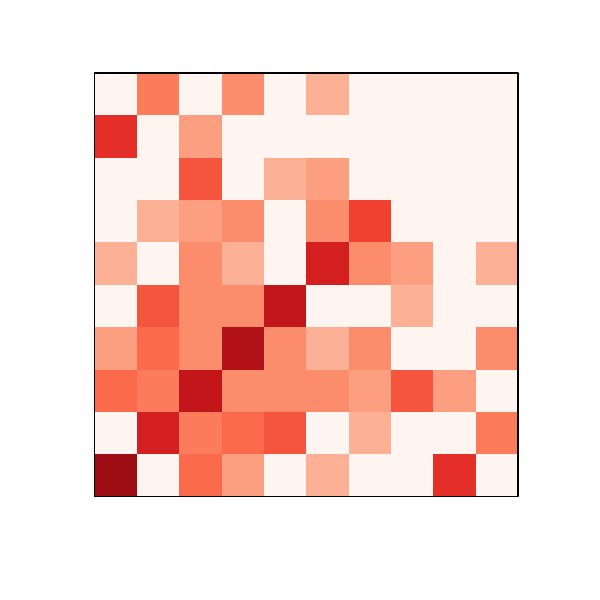}
    \caption{Help decision}
\end{subfigure}
\hfill
\begin{subfigure}[b]{0.225\linewidth}
    \centering
    \includegraphics[width=\textwidth, trim = {1cm, 1.5cm, 1cm, 1cm}]{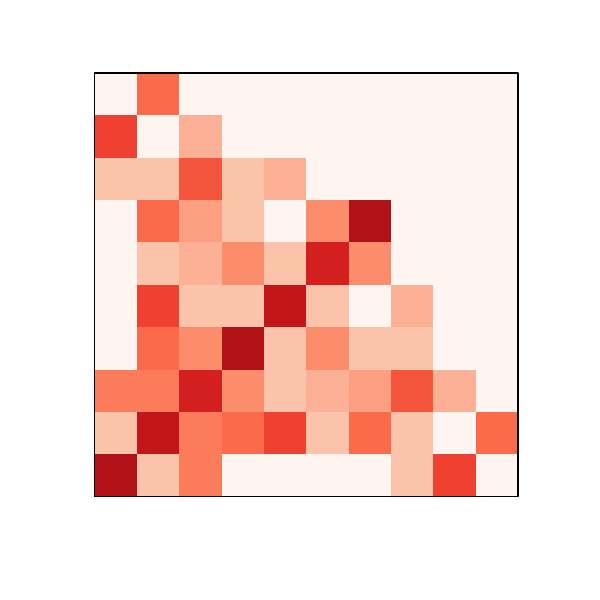}
    \caption{Kero rice come}
\end{subfigure}
\\
\begin{subfigure}[b]{0.225\linewidth}
    \centering
    \includegraphics[width=\textwidth, trim = {1cm, 1.5cm, 1cm, 1cm}]{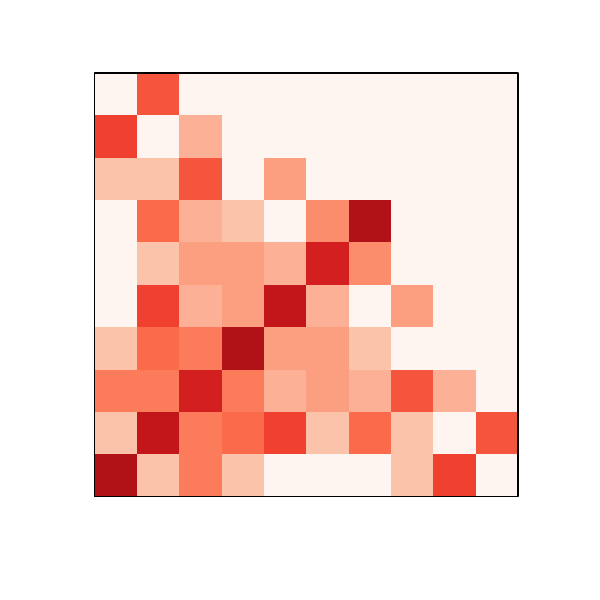}
    \caption{Kero rice go}
\end{subfigure}
\hfill
\begin{subfigure}[b]{0.225\linewidth}
    \centering
    \includegraphics[width=\textwidth, trim = {1cm, 1.5cm, 1cm, 1cm}]{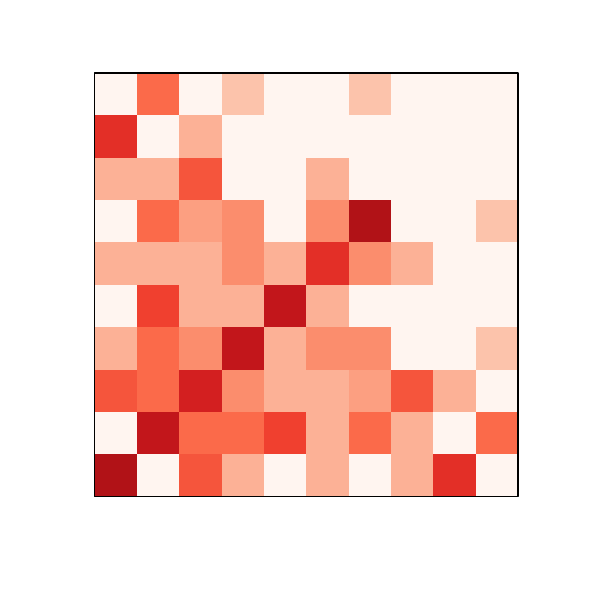}
    \caption{Lend money}
\end{subfigure}
\hfill
\begin{subfigure}[b]{0.225\linewidth}
    \centering
    \includegraphics[width=\textwidth, trim = {1cm, 1.5cm, 1cm, 1cm}]{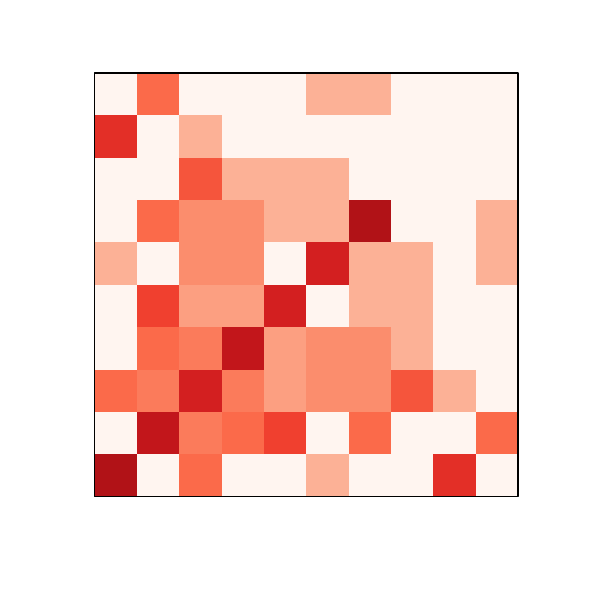}
    \caption{Medic}
\end{subfigure}
\hfill
\begin{subfigure}[b]{0.225\linewidth}
    \centering
    \includegraphics[width=\textwidth, trim = {1cm, 1.5cm, 1cm, 1cm}]{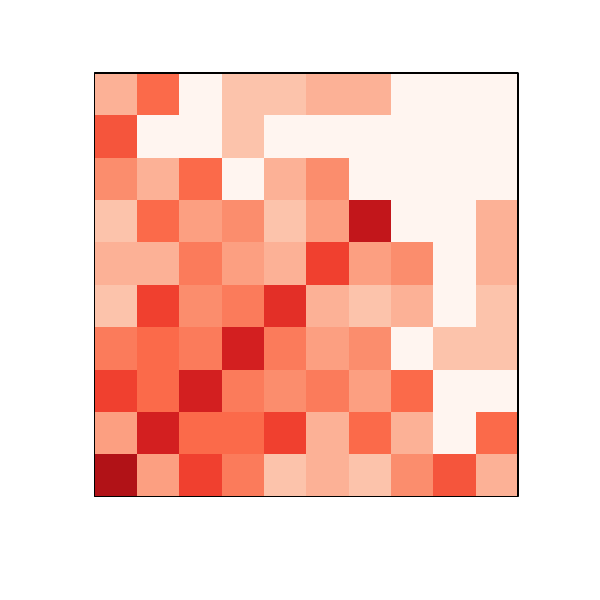}
    \caption{Nonrel}
\end{subfigure}
\\
\begin{subfigure}[b]{0.225\linewidth}
    \centering
    \includegraphics[width=\textwidth, trim = {1cm, 1.5cm, 1cm, 1cm}]{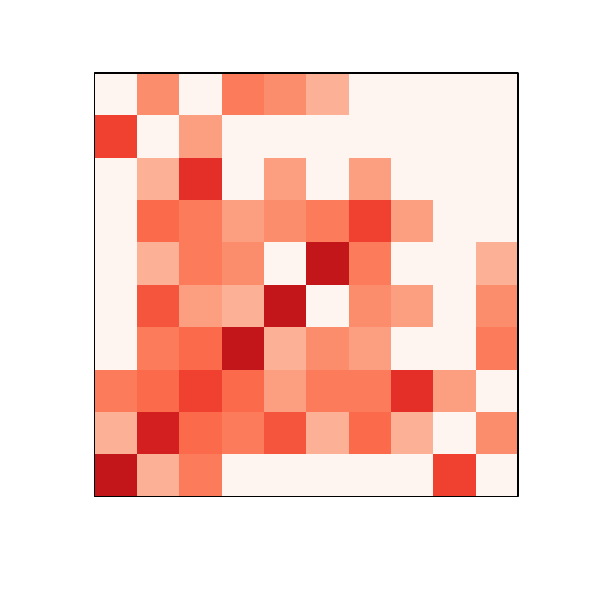}
    \caption{Rel}
\end{subfigure}
\hfill
\begin{subfigure}[b]{0.225\linewidth}
    \centering
    \includegraphics[width=\textwidth, trim = {1cm, 1.5cm, 1cm, 1cm}]{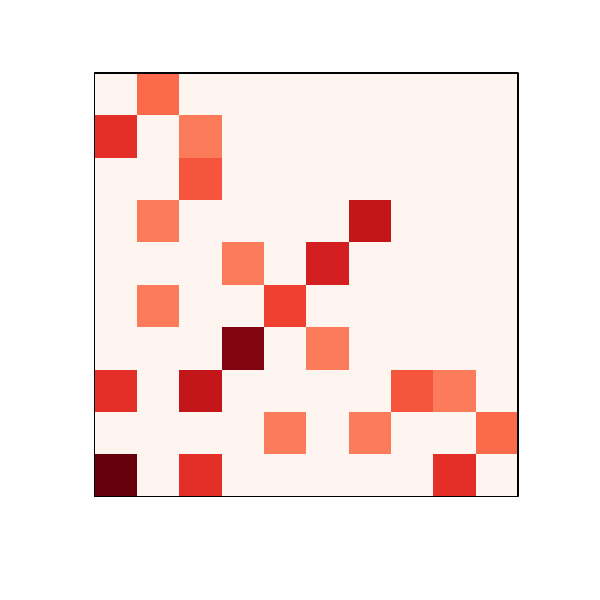}
    \caption{Temple company}
    \label{subfig:mnhist_temple}
\end{subfigure}
\hfill
\begin{subfigure}[b]{0.225\linewidth}
    \centering
    \includegraphics[width=\textwidth, trim = {1cm, 1.5cm, 1cm, 1cm}]{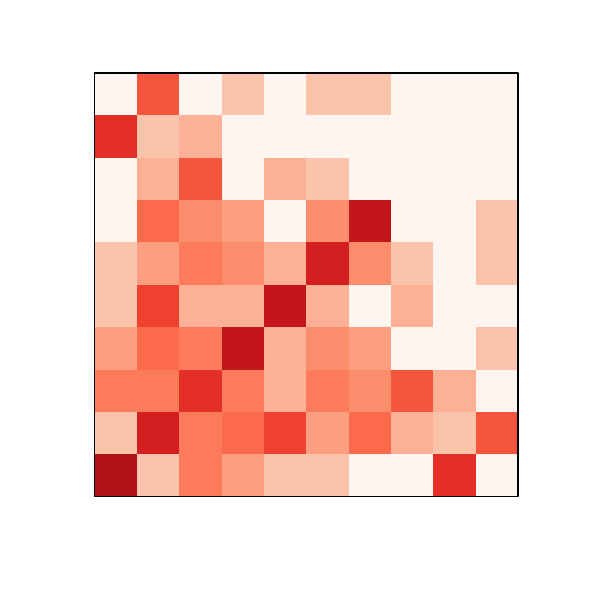}
    \caption{Visit come}
    \label{subfig:mnhist_visitcome}
\end{subfigure}
\hfill
\begin{subfigure}[b]{0.225\linewidth}
    \centering
    \includegraphics[width=\textwidth, trim = {1cm, 1.5cm, 1cm, 1cm}]{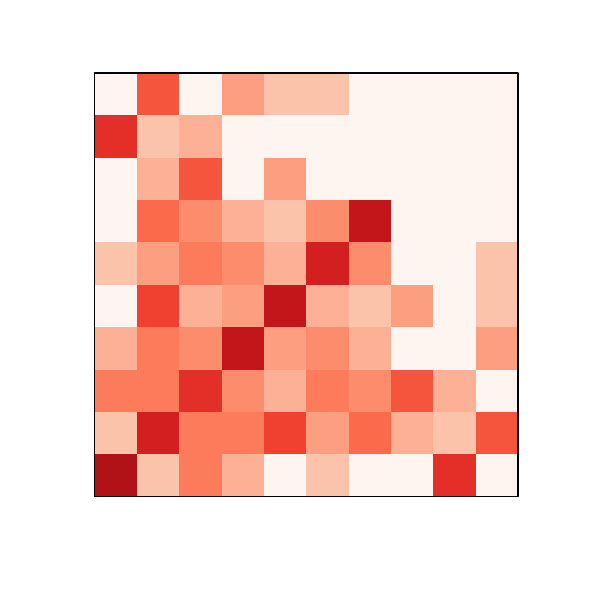}
    \caption{Visit go}
    \label{subfig:mnhist_visitgo}
\end{subfigure}
    \caption{Heatmaps of $\{\widehat{f}^{(\ell)}(x,y)\}^{1/4}$ estimated by multi-network histogram. In each heat map, the bottom left bin corresponds to $(1,1)$, and the top right one corresponds to $(10,10)$. From left to right, the first nine groups have size $23$, and the last group has size $24$.}
    \label{fig:indian_vil_40_multinethist}
\end{figure}
\begin{figure}[!ht]
        \centering
     \begin{subfigure}[b]{0.475\linewidth}
         \centering
    \includegraphics[width=\textwidth]{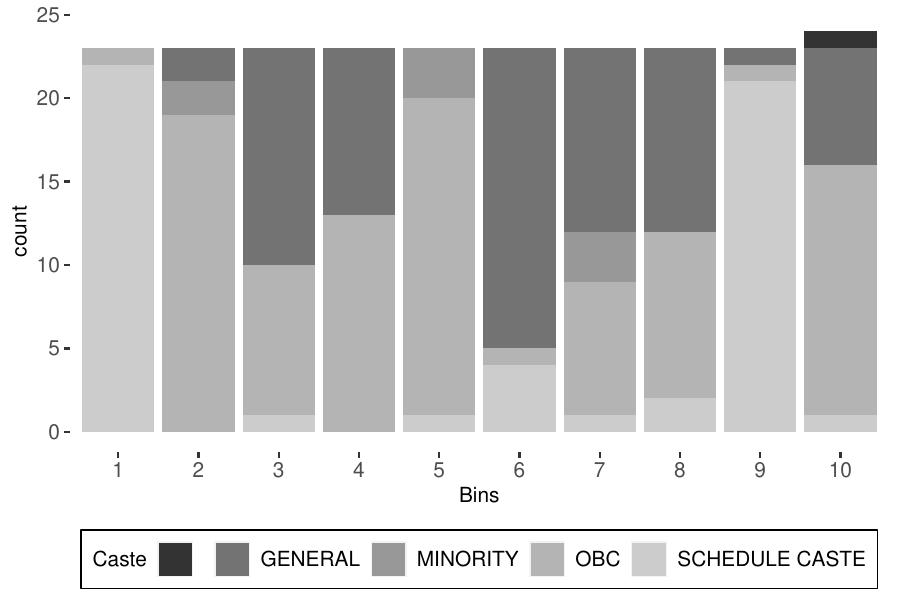}
    \caption{Caste}
    \label{fig:vil_40_caste}
    \end{subfigure}
    % \newline
    %    \centering
     \begin{subfigure}[b]{0.475\linewidth}
         \centering
    \includegraphics[width=\textwidth]{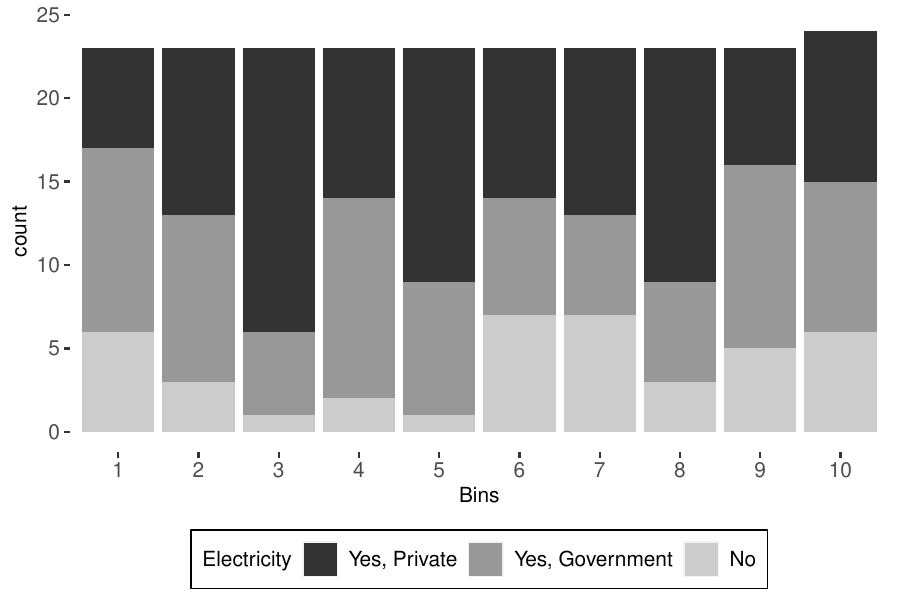}
    \caption{Electricity}
    \label{fig:vil_40_electricity}
    \end{subfigure}
\caption{Summary of covariates by bins (groups) from the multi-network histogram. Bin IDs correspond to their left-to-right positions in Figure~\ref{fig:indian_vil_40_multinethist}.}
\label{fig:indian_vil_40_covariate}
\end{figure}

We fit the proposed multi-network histograms to this dataset. Using Algorithm~\ref{alg:bw_sel}, we obtain the data-driven bandwidth, $\widehat{h}=23$, that is applied to all layers. Applying~Algorithm~\ref{alg:multi_nethist_combined_short} with this bandwidth splits the 231 households in the village into ten groups, including nine groups consisting of 23 households and one group of 24 households. With these groups, we estimate the functions $\{f^{(\ell)}\}_{\ell=1}^{12}$ through blockmodel approximation, as presented in Figure~\ref{fig:indian_vil_40_multinethist}. The position of each group in the heatmaps is identical across the layers, and each block corresponds to a pair of groups. From left to right and bottom to top, the blocks correspond to groups 1 through 10. For example, the bottom-left block maps to $(1,1)$, and the top-right block represents $(10,10)$. Darker blocks imply larger values of $\widehat{f}^{(\ell)}$.

Figure~\ref{fig:indian_vil_40_multinethist} illustrates both global and local structural features of the networks. First, the estimated functions $\{\widehat{f}^{(\ell)}\}_{\ell=1}^{12}$ exhibit distinct global structures from layer to layer. For instance, in relatively sparse layers such as Temple Company, few blocks have large function values, while most blocks are zero. However, in denser layers such as Visit Come, many blocks show positive values, indicating connections among multiple groups for this socioeconomic activity.
Next, blocks in the estimated graphons have different interaction patterns, highlighting variations in their local structures. For example, the bottom-left block $(1,1)$ shows darker colours across all layers, implying that the households in this group tend to participate together in various social activities. In contrast, the adjacent block $(2,1)$ has non-zero values only in a subset of layers such as Kero Rice Come/Go, Visit Come/Go, and Nonrel/Rel, while showing little or no connections in other layers. This suggests that their interactions primarily occur in activities with lower socioeconomic costs.

This data analysis further demonstrates that the proposed method can estimate graphons more accurately even in sparser layers. As shown in Theorem~\ref{thm:WMISE}, even when a layer has few edges, the bandwidth of the multi-network histogram can be reduced by borrowing vertex information from other layers  having a larger number of edges. In this example, the selected bandwidth is $23$, allowing the structure of the Temple Company layer to be visualized at a higher resolution. This cannot be achieved using a single-layer approach, which results in a single large block with a bandwidth of 231, as shown in Figure~\ref{subfig:nethist_templecompany}.

We also find that our method captures characteristics of each group to some extent. Figures~\ref{fig:vil_40_caste} and \ref{fig:vil_40_electricity} display counts by castes and electricity access across vertex groups, respectively, highlighting differences among groups. In Figure~\ref{fig:vil_40_caste}, groups 1 and 9 are predominantly scheduled caste, whereas group 6 contains mostly general caste. Figure~\ref{fig:vil_40_electricity} presents that the majority of groups 3, 5, and 8 rely on electricity from private companies, while groups 1 and 9 rely on it relatively less. It is worth noting that our method partitions the groups using only the network connections, without incorporating any additional covariates. 
Nevertheless, the estimated groups reflect household characteristics, which may be explained by the effect of these features on interaction patterns.
\begin{figure}[!ht]
    \centering
\includegraphics[width=0.5\linewidth 
%, trim={1cm, 1.5cm, 0cm, 1.5cm}
]{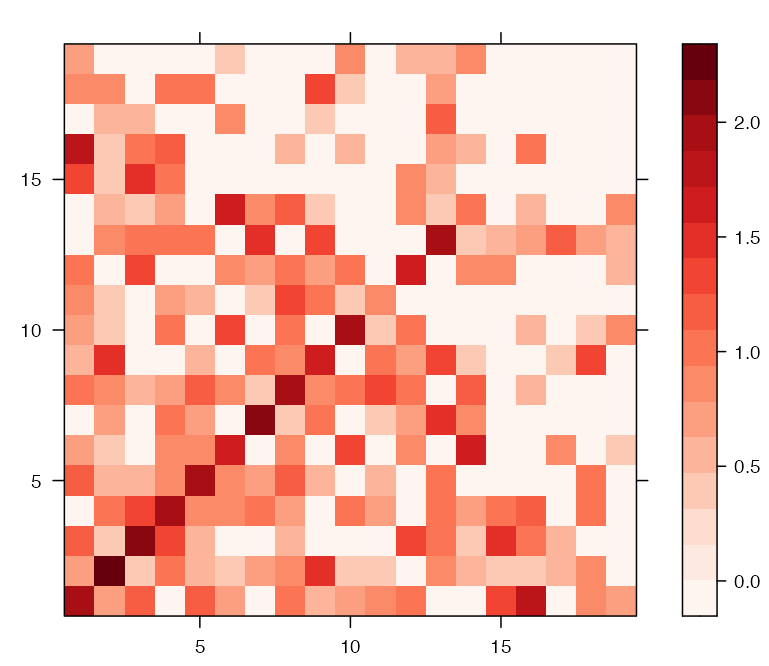}
\caption{Heatmap of $\{\widehat{f}(x,y)\}^{1/4}$ from the homogeneous multi-network histogram fitted to the nine selected layers: borrow/lend money, give advice, help decision, kero rice come/go, medic, visit come/go. }
\label{fig:indian_vil_40_homomultinethist}
\end{figure}

Lastly, we fit a homogeneous multi-network histogram to estimate a common graphon for layers with a similar structure. To identify these layers, we employ a two-sample network test proposed by \cite{shao2022Higherorder} with multiple testing. This procedure selects nine layers, including borrow/lend money, give advice, help decision, kero rice come/go, medic, visit come/go (see Supplementary Material Section~\ref{sec_append:indian_vils} for details). With the data-driven bandwidth of 12 obtained from \eqref{eq:oracle_bandwidth_MISE_homogeneous}, we employ Algorithm~\ref{alg:multi_nethist_combined_short} for the homogeneous case to fit a homogeneous multi-network histogram. The result is presented in Figure~\ref{fig:indian_vil_40_homomultinethist}. As outlined in Theorem~\ref{thm:MISE_homogeneous}, the homogeneous multi-network histogram produces a higher-resolution heatmap than the multi-network histogram in Figure~\ref{fig:indian_vil_40_multinethist}, providing a more precise approximation to the underlying graphon.

\section{Discussion}
\label{sec:discussion}

In this paper, we study the estimation of multiple graphons from a multiplex network realization. The birds-eye view of this problem considers 
overall heterogeneity and commonality of the vertex set across network layers. 
Overall heterogeneity accounts for layer-specific interaction structures and levels of edge sparsity. 
Considering different sparsity levels
is especially important when layers have varying sparsity orders, not just different rate constants, since the relative amount of information provided to each layer in the limit diverges as the number of vertices increases. 
Whether a network is considered sparse depends on the choice of generating mechanism~\cite{bianconi2018Multilayer,bianconi2014multiple,boccaletti2014Structure}. We account for varying sparsity across layers using the scaling operation defined in Definition~\ref{def:scaled_multigraphon}.
Even when each layer is represented by a distinct graphon function, each vertex shares a common latent variable across layers, which provides shared information. 
Under these settings, 
if each graphon is smooth, our blockmodel approximation approach can estimate each layer's success probabilities nonparametrically.

Network sparsity has generated considerable controversy in the field over the past few years, since sparse networks are empty in the limit.
Several generative mechanisms have been proposed to avoid empty limits and allow for power-law degrees and sparse realized networks, see for example \cite{caron2017Sparse}. However, it still remains important to study how networks behave under finite but large sampling schemes within the scaled graphon framework.
This point is illustrated in the Indian villages example in \Cref{sec:real_data}, where interactions in the temple company  layer are socially costly for a household, making it unlikely for them to connect with multiple households. Such sparsity corresponds to small parameters in that layer, highlighting the need to identify underlying structures when some layers provide limited information.

If the vertex sets across layers are identical, it is beneficial to group vertices for blockmodel approximation even when the layers share no other features. An area of future work might be to perturb that assumption and assess which methods continue to work. Multilayer networks in general also concern edges between vertices in different layers, so grouping vertices using the supra-adjacency matrix may capture these inter-layer relationships. 

The example of the household networks of Indian villages provides many types of insights. Firstly, it is clear that similar interaction patterns can be found across different types of interactions, but that each interaction comes with an associated cost. 
By combining information across layers to estimate a shared vertex partition, we obtain a higher-resolution estimator of a graph limit, which in turn allows a smaller bandwidth. 
For very sparse layers, such as the temple company layer, this makes a considerable difference.

There are many outstanding questions on the benefits of multiple layers; perhaps the most promising is the instance of privacy retracted layers, where using multiple layers allows for inference under limited information. 
There is also potential to further explore the vector nature of multilayer sparsity. Potential constraints on edges might be to couple degrees between layers. This means the sparsity could be coupled with an edge budget used between layers. As we have observed in the Indian household network data, the combined information across layers provides a higher-resolution understanding of interactions, which gives a more detailed view.

Our theoretical guarantees are established under the oracle label and do not extend to the label estimated by Algorithm~\ref{alg:multi_nethist_combined_short}. In the single-layer network
case, the minimax risk of graphon estimation separates into a nonparametric term
and a clustering term of order $\log k/n$ from assigning $n$ vertices to $k$
groups \cite{gao2015rate}. Within the graphon estimation framework, consistency of joint latent position estimation across multiple networks has only recently begun to be studied. \cite{sogan2026LowComplexity} establish consistency in a dense setting with a single common graphon and without a shared node set, and \cite{chandna2026ordinal} allow for a shared node set via ordinal embedding but likewise assume an aggregated graphon. Our setting is more challenging than theirs, as it allows layers to share a node set while having distinct graphons and sparsity levels that may decay at different orders. Establishing consistency guarantees for the estimated label under this more general setting is an important direction for future work.

\section*{Acknowledgments}

This work was supported by the European Research Council [CoG 2015-682172NETS], within the Seventh European Union Framework Program. 
We thank the anonymous reviewers for their careful reading of the manuscript and constructive comments that helped improve the presentation and clarity of this work.


\bibliographystyle{comnet}
\bibliography{ref}

\clearpage
\begin{center}
{\bfseries Supplementary Material: Joint Estimation of Sparse Multilayer Networks via Graph Limits}
\end{center}

\begin{appendix}
\numberwithin{equation}{section}
\numberwithin{lemma}{section}
\numberwithin{theorem}{section}
\numberwithin{corollary}{section}
\numberwithin{figure}{section}
\numberwithin{table}{section}

In the Supplementary Materials, we provide proofs for the theoretical results, technical lemmas, a justification for the weighting scheme in the homogeneous multi-network histogram, additional theory and simulation results, and further details on the real data analysis.

\section{Notation}
\label{sec_append:notation}

We first summarize the notation used throughout the Supplementary Material. Although some terms appear in the main article, they are restated here to ensure this document is self-contained.

A function $f(x,y): [0,1]^2 \to \mathbb{R}_0^+$ is H\"{o}lder-$\alpha$ continuous for $\alpha\in(0,1]$ with notation $f\in\text{H\"{o}lder}^{\alpha}(M)$ if there exists a constant $M$ such that
\begin{align*}
\sup_{(x,y)\neq (x',y')}\frac{|f(x,y)-f(x',y')|}{\|(x,y)-(x',y')\|_2^{\alpha}} \leq M < \infty.
\end{align*}
We assume $n = hk + r$ with the bandwidth $h$ and $0\leq r < h$. Denote an index set by $[n] = \{1,2,\ldots, n\}$. Let $\bz = (z_1,\ldots, z_n) \in \mathcal{Z}_k$ be a group label vector, where $\mathcal{Z}_k \subset [k]^n$ is the collection of possible network-histogram type group assignments.
We denote the set of index pairs in the $(a,b)$-th bin by
\begin{align*}
R_{ab}(\bz) = \begin{cases}
    \{(i,j): i<j, z_i = z_j = a\}, & a = b,\\
    \{(i,j): z_i = a, z_j = b\}, & a \neq b.
\end{cases}
\end{align*}
The group and bin sizes of the multi-network histogram satisfy $h_a = h + r\I(a=k)$ and 
\begin{align*}
h_{ab}^2 := |R_{ab}(\bz)|= \begin{cases}
    \binom{h_a}{2}& a = b\\
    h_a h_b & a \neq b
\end{cases},
\end{align*}
respectively. 

We also restate the definition of the oracle group label vector \eqref{def:oracle_z} in the main article. Let the latent variables $\xi_1,\ldots, \xi_n$ be an \iid sample from the uniform distribution $U(0,1)$ on $[0,1]$. The oracle group label vector $\widetilde{\bz} \in \mathcal{Z}_k$ is defined as
\begin{align*}
\widetilde{z}_i = \min\left\{\left\lceil\frac{\text{rank}(\xi_i)}{h}\right\rceil, k\right\},
\end{align*}
where $\text{rank}(\xi_i)$ is the rank of $\xi_i$ in increasing order, $\xi_{(1)}\leq \xi_{(2)}\leq \cdots \leq \xi_{(n)}$. 

We further introduce notation based on the oracle group labels.
We define a set of vertices pairs under the oracle group labels $\widetilde{\bz}$,
\begin{align}
\label{def:R_ab_oracle}
R_{ab}^* := R_{ab}(\widetilde{\bz}) = 
\begin{cases}
    \{(i,j): i<j, \widetilde{z}_i = \widetilde{z}_j = a\}, & a = b,\\
    \{(i,j): \widetilde{z}_i = a, \widetilde{z}_j = b\}, & a \neq b.
\end{cases}
\end{align}
for $a,b\in[k]$. With the oracle labels, multi-network histograms partition a vertex pair set $\{(i,j): i,j\in[n]\}$ into $k^2$ subsets, which corresponds to a partition of the domain $[0,1]^2$ into $k^2$ blocks. Let $\omega_{ab}$ denote the region in $[0,1]^2$ corresponding to the $(a,b)$-th bin:
\begin{align}
\label{def:omega_ab}
\omega_{ab}=\begin{cases}
\left[{(a-1)h}/{n}, {ah}/{n}\right]\times \left[{(b-1)h}/{n}, {bh}/{n}\right]  & a<k \text{ and } b<k,\\
\left[{(k-1)h}/{n},1\right]\times \left[{(b-1)h}/{n}, {bh}/{n}\right] & a=k \text{ and } b<k,\\
\left[{(a-1)h}/{n}, {ah}/{n}\right]\times \left[{(k-1)h}/{n}, 1\right] & a<k \text{ and } b=k,\\
\left[{(k-1)h}/{n}, 1\right]\times \left[{(k-1)h}/{n}, 1\right] & a=k \text{ and } b=k.
\end{cases}
\end{align}
Consequently, the area of the $(a,b)$-th region $\omega_{ab}$ is 
\begin{align}
\label{def:area_omega_ab}
|\omega_{ab}| = \frac{h_{a}h_{b}}{n^2},
\end{align}
and $\sum_{a,b=1}^k |\omega_{ab}| = 1$.
Within each $\omega_{ab}$, the averages of $f^{(\ell)}$ and $(f^{(\ell)})^2$ are denoted by
\begin{align}
\label{def:avg_f}
\Bar{f}_{ab}^{(\ell)}:=\frac{1}{|\omega_{ab}|}\iint_{\omega_{ab}}f^{(\ell)}(x,y)dxdy,
\ \ 
\overline{(f^{(\ell)})^2}_{ab}:=\frac{1}{|\omega_{ab}|}\iint_{\omega_{ab}}\{f^{(\ell)}(x,y)\}^2dxdy,
\end{align}
respectively.

\section{Proofs of the main results}
\label{sec_append:proofs_main}

In this section, we provide the proofs of Theorem~\ref{thm:WMISE}, and Theorem~\ref{thm:MISE_homogeneous}. The arguments for the former rely on the upper bound for the layer-wise mean integrated squared error (MISE) established in Lemmas~\ref{lem:moment_A_bar} and \ref{lem:layer_MISE}. The proof of the latter builds on Lemmas~\ref{lem:interlayer_covariance} and \ref{lem:moment_A_bar_homo}.

\subsection{Proof of Theorem~\ref{thm:WMISE}}

% \begin{proof}
From \eqref{eq:oracle_WMISE}, the oracle weighted mean integrated square error (WMISE) is naturally bounded by weighted average of layer-wise MISEs, 
\begin{align*}
\text{WMISE}(\widehat{\bff}^*) 
\leq \sum_{\ell=1}^L \frac{\rho_n^{(\ell)}}{\sum_{l=1}^L \rho_n^{(l)}}\E \left( \iint_{[0,1]^2}\left|(\widehat{f}^{(\ell)})^*(x,y;h)-f^{(\ell)}(x,y)\right|^2 dxdy\right).\\
\end{align*}
Applying Lemma~\ref{lem:layer_MISE} yields the upper bound,
\begin{equation}
\begin{aligned}
\label{ineq:WMISE_derivation}
\text{WMISE}(\widehat{\bff}^*) 
&\leq \frac{1}{\sum_{\ell=1}^L \rho_n^{(\ell)}}
\left[\sum_{\ell=1}^L M_{\ell}^2 
\left\{
2^{\alpha}\rho_n^{(\ell)}\left(\frac{h}{n}\right)^{2\alpha}
+ \frac{2\rho_n^{(\ell)}}{(2n)^{\alpha}}
+ \frac{1}{M_{\ell}^2 h^2 }
\right\}\right]\{1+o(1)\}\\
&=\frac{1}{\overline{\rho_n}}
\left\{
2^{\alpha}\left(\frac{h}{n}\right)^{2\alpha}
\overline{M^2 \rho_n}
+ \frac{2\overline{M^2\rho_n}}{(2n)^{\alpha}}
% + \frac{L}{h^2 }
+ \frac{1}{h^2 } %typo-fix
\right\}\{1+o(1)\},
\end{aligned}
\end{equation}
where $\overline{\rho_n} = L^{-1}\sum_{\ell=1}^L\rho_n^{(\ell)}$
and $\overline{M^2\rho_n} = L^{-1}\sum_{\ell=1}^L M_{\ell}^2\rho_n^{(\ell)}$.
The right-hand side of \eqref{ineq:WMISE_derivation} is minimized at
\begin{align*}
(h^*)^{2\alpha+2} &= \frac{n^{2\alpha} L}{\alpha2^{\alpha}\sum_{\ell=1}^L M_{\ell}^{2}\rho_n^{(\ell)}}
= \frac{n^{2\alpha}}{\alpha2^{\alpha}\overline{M^2\rho_n}},
\end{align*}
which corresponds to \eqref{eq:oracle_bandwidth_WMISE}.
By substituting \eqref{eq:oracle_bandwidth_WMISE} into the upper bound of $\text{WMISE}(\widehat{\bff}^*)$ \eqref{ineq:WMISE_derivation}, we can simplify the right-hand side of \eqref{eq:oracle_bandwidth_WMISE} as
\begin{align*}
\text{WMISE}(\widehat{\bff}^*)\Big|_{h=h^*}
&\leq
\frac{1}{\overline{\rho_n}}
\left\{
2^{\alpha}\left(\frac{h^*}{n}\right)^{2\alpha}
\overline{M^2 \rho_n}
+ \frac{2\overline{M^2\rho_n}}{(2n)^{\alpha}}
+ \frac{1}{(h^*)^2 }
\right\}\{1+o(1)\}\\
&= \frac{1}{\overline{\rho_n}}
\left\{\left(\alpha^{\frac{1}{\alpha+1}} + \alpha^{-\frac{\alpha}{\alpha+1}}\right)\left(\frac{2^\alpha \overline{M^2\rho_n}}{n^{2\alpha}}\right)^{\frac{1}{\alpha+1}} + \frac{2\overline{M^2\rho_n}}{(2n)^{\alpha}}\right\}\{1+o(1)\}\\
&= O\left(\left\{\binom{n}{2}\overline{\rho_n}\right\}^{-\frac{\alpha}{\alpha+1}}\right),
% \left[
% 2\left\{\frac{\sum_{\ell=1}^L w_{\ell}^{*}M_{\ell}^2}{\binom{n}{2}L^{-1}\sum_{\ell=1}^L \rho_n^{(\ell)}}\right\}^{1/2} + \frac{1}{n}\sum_{\ell=1}^L w_{\ell}^{*}M_{\ell}^2
% \right]\{1+o(1)\}\\
% =& O\left[\sqrt{\left\{\binom{n}{2}\frac{1}{L}\sum_{\ell=1}^L \rho_n^{(\ell)}\right\}^{-1}}\right].
\end{align*} 
which yields \eqref{ineq:WMISE_oracle}.
% \end{proof}

\subsection{Proof of Theorem~\ref{thm:MISE_homogeneous}}

% \begin{proof}
From \eqref{eq:oracle_MISE_homo}, the MISE of the oracle estimator is bounded by bias-variance decomposition, following the same argument of Lemma~\ref{lem:layer_MISE},
\begin{align*}
\text{MISE}(\widehat{\bff}^*) 
&\leq \E \left( \iint_{[0,1]^2}\left|\widehat{f}^*(x,y;h)-f(x,y)\right|^2 dxdy\right)\\
&=\sum_{a,b=1}^k \iint_{\omega_{ab}}\left[\left\{\text{bias}\left(\widehat{f}^{*}(x,y;h)\right)\right\}^2 + \var\left(\widehat{f}^{*}(x,y;h)\right)\right] dxdy.
\end{align*}
where $(\widehat{f}^{(\ell)})^{*}(x,y;h)$ denotes the oracle homogeneous multi-network histogram
defined in \eqref{eq:oracle_f_homo}. We focus on the terms for each block $\omega_{ab}$ in the last line of the inequality above. By Lemma~\ref{lem:moment_A_bar_homo}, we derive that
\begin{align*}
& \iint_{\omega_{ab}}\left[\left[\text{bias}\left\{\widehat{f}^{*}(x,y;h)\right\}\right]^2
+ \var\left\{\widehat{f}^{*}(x,y;h)\right\}\right]dxdy\\
\leq&
\iint_{\omega_{ab}}
\left\{
\left|\E\widehat{f}_{ab}^* - \Bar{f}_{ab}\right|^2 
+ 
\left|\Bar{f}_{ab}-f(x,y)\right|^2 \right\} dxdy
+
\iint_{\omega_{ab}}
\left\{
 \frac{1}{h_{ab}^2}\sum_{\ell=1}^L (\tau_{\ell}^*)^2\left(\frac{\Bar{f}_{ab}}{\rho_n^{(\ell)}} - \overline{f^2}_{ab}\right)\right\}dxdy\\
&+\iint_{\omega_{ab}}
\left[
\frac{1}{\left(\sum_{\ell=1}^L \rho_n^{(\ell)}\right)^2}
\sum_{\ell=1}^L \frac{M\rho_n^{(\ell)}}{h_{ab}^2(2n)^{\alpha/2}}\{1+o(1)\}
+ \frac{M^2}{(2n)^{\alpha}}
\right]
dxdy\\
%%%%%%%%%%%%%%
\leq& 
\iint_{\omega_{ab}}
\left|\Bar{f}_{ab}
-f(x,y)\right|^2  dxdy
+ |\omega_{ab}|
\left\{
\frac{M^2\{1+o(1)\}}{(2n)^{\alpha}}
+
 \frac{1}{h_{ab}^2}\sum_{\ell=1}^L(\tau_{\ell}^*)^2\left(\frac{\Bar{f}_{ab}}{\rho_n^{(\ell)}} - \overline{f^2}_{ab}\right)\right\}\\
&+ |\omega_{ab}|
\left[
\frac{1}{\left(\sum_{\ell=1}^L \rho_n^{(\ell)}\right)^2}
\sum_{\ell=1}^L \frac{M\rho_n^{(\ell)}}{h_{ab}^2(2n)^{\alpha/2}}\{1+o(1)\}
+ \frac{M^2}{(2n)^{\alpha}}
\right].
\end{align*}
Using similar arguments in the proof of Lemma~\ref{lem:layer_MISE}, we obtain
\begin{align*}
\frac{1}{|\omega_{ab}|}\iint_{\omega_{ab}}
\left|\Bar{f}_{ab}
-f(x,y)\right|^2  dxdy 
\leq     
2^{\alpha}M^2\left(\frac{h}{n}\right)^{2\alpha}\left\{1+ O\left(\frac{h}{n}\right)\right\}.
\end{align*}
Recall $\tau_{\ell}^* = \rho_n^{(\ell)}/\sum_{\ell=1}^L \rho_n^{(l)}$. Moreover, it holds that
\begin{align*}
\sum_{\ell=1}^L \sum_{a,b=1}^k\frac{|\omega_{ab}|}{h_{ab}^2}\frac{(\tau_{\ell}^*)^2\Bar{f}_{ab}}{\rho_n^{(\ell)}}
&=\sum_{\ell=1}^L\sum_{a,b=1}^k \frac{(\tau_{\ell}^*)^2}{\rho_n^{(\ell)}h_{ab}^2}\frac{|\omega_{ab}|}{|\omega_{ab}|}\iint_{\omega_{ab}}f(x,y)dxdy\\
&= \sum_{\ell=1}^L \frac{(\tau_{\ell}^*)^2}{\rho_n^{(\ell)}h^2}\{1+o(1)\}
= \frac{1}{(\sum_{\ell=1}^L \rho_n^{(\ell)})h^2}\{1+o(1)\} ,\\
\sum_{\ell=1}^L \sum_{a,b=1}^k|\omega_{ab}|\frac{(\tau_{\ell}^*)^2\overline{f^2}_{ab}}{h_{ab}^2}
&=\sum_{\ell=1}^L \sum_{a,b=1}^k \frac{(\tau_{\ell}^*)^2}{h_{ab}^2}\frac{|\omega_{ab}|}{|\omega_{ab}|}\iint_{\omega_{ab}}\{f(x,y)\}^2dxdy\\
&= \sum_{\ell=1}^L\frac{(\tau_{\ell}^*)^2}{h^2}\{O(1)+o(1)\}.
\end{align*}
These results follow from $|\omega_{ab}| = h_{ab,r}^2/n^2$ as in \eqref{def:area_omega_ab}, $\iint_{[0,1]^2} f(x,y)dxdy = 1$ by Condition~\ref{cond:multinethist}, and 
\begin{align*}
\sum_{\ell=1}^L \frac{(\tau_{\ell}^*)^2}{\rho_n^{(\ell)}}
= \sum_{\ell=1}^L \frac{1}{\rho_n^{(\ell)}}\frac{(\rho_n^{(\ell)})^2}{(\sum_{l=1}^L \rho_n^{(l)})^2}
= \frac{\sum_{\ell=1}^L \rho_n^{(\ell)}}{(\sum_{\ell=1}^L \rho_n^{(\ell)})^2}
= \frac{1}{\sum_{\ell=1}^L \rho_n^{(\ell)}}.
\end{align*}
Combining all results above with $\sum_{a,b=1}^k |\omega_{ab}|=1$, we have 
\begin{align*}
 &\E \left( \iint_{[0,1]^2}\left|
\widehat{f}^{*}(x,y;h)-f(x,y)\right|^2 dxdy\right)\\
\leq & 
2^{\alpha}M^2\left(\frac{h}{n}\right)^{2\alpha}\left\{1+ O\left(\frac{h}{n}\right)\right\}
+ \left\{\frac{2M^2}{(2n)^{\alpha}}
+ \frac{1}{(\sum_{\ell=1}^L \rho_n^{(\ell)})h^2}\right\}\{1+o(1)\}\\
&+ \sum_{\ell=1}^L\frac{(\tau_{\ell}^*)^2}{h^2}\{O(1)+o(1)\} \\
=& M^2\left\{2^{\alpha}\left(\frac{h}{n}\right)^{2\alpha} + \frac{2}{(2n)^{\alpha}} + \frac{1}{M^2 L \overline{\rho_n} h^2}\right\}\left\{1+o(1)\right\},
\end{align*}
where $\overline{\rho_n} = L^{-1}\sum_{\ell=1}^L\rho_n^{(\ell)}$,
which leads to the claimed upper bound of $\text{MISE}(\widehat{f}^*)$.

Next, we derive the optimal bandwidth $h^*$ by minimizing the upper bound above, 
\begin{align*}
(h^*)^{2\alpha+2} &= \frac{n^{2\alpha}}{\alpha2^{\alpha}M^2 L \overline{\rho_n}},
\end{align*}
which corresponds to \eqref{eq:oracle_bandwidth_MISE_homogeneous}.
The MISE of $\widehat{f}^*$ at $h^*$ is bounded by 
\begin{align*}
 \text{MISE}(\widehat{f}^*)|_{h=h^*} 
 \leq    &
 M^2\left\{2^{\alpha}\left(\frac{h^*}{n}\right)^{2\alpha} + \frac{2}{(2n)^{\alpha}} + \frac{1}{M^2 L \overline{\rho_n}(h^*)^2}\right\}\left\{1+o(1)\right\}\\
 \leq & M^2\left\{\left(\alpha^{\frac{1}{\alpha+1}} + \alpha^{-\frac{\alpha}{\alpha+1}}\right)\left(\frac{2}{n^{2}M^2 L \overline{\rho}_n}\right)^{\frac{\alpha}{\alpha+1}} + \frac{2}{(2n)^{\alpha}}\right\}\{1+o(1)\}\\
=& O\left[\left\{\binom{n}{2}L \overline{\rho_n}\right\}^{-\alpha/(\alpha+1)}\right],
\end{align*}
which is \eqref{ineq:MISE_oracle_homogeneous}.
This completes the proof.
% \end{proof}

\section{Lemmas for Proofs of Main Results}
\label{sec_appendix:aux_lemma}

In this section, we present technical lemmas used to prove the main results. Recall that we observe $L$ adjacency matrices $\bA^{(\ell)}:=\{A_{ij}^{(\ell)}\}$, modeled conditionally on \(\bxi=(\xi_1,\dots,\xi_n)\) by
\begin{equation}
\label{eq:cond_moments}
\begin{aligned}
\E\{A_{ij}^{(\ell)}\,|\,\bxi \}=& \rho_n^{(\ell)}f^{(\ell)}(\xi_i,\xi_j), \\
\var\{A_{ij}^{(\ell)}\,|\,\bxi \}=& \rho_n^{(\ell)}f^{(\ell)}(\xi_i,\xi_j)\{1-\rho_n^{(\ell)}f^{(\ell)}(\xi_i,\xi_j)\},
\end{aligned}
\end{equation}
where $\xi_1,\ldots, \xi_n$ are \iid sample from $U(0,1)$. These are used in the proof of the technical lemma.

\subsection{Technical Lemma for Theorem~\ref{thm:WMISE}}
\label{subsec_append:tech_lem}

We present lemmas to prove the theorems in the main article, providing the moment bounds and layer-wise MISEs, which are essential for theoretical studies for the multi-network histogram estimator.

We consider the multi-network histogram under the oracle label $\widetilde{\bz}$ in the theoretical studies. Restating the bin height of the $(a,b)$-th block in the $\ell$th layer, defined in \eqref{def:oracle_Theta}, using $R_{ab}^*$ in \eqref{def:R_ab_oracle}, we obtain
\begin{equation}
\label{eq:oracle_A_ab_bar}
\begin{aligned}
(\bar{A}_{ab}^{(\ell)})^* &= \frac{1}{h_{ab}^2}\sum_{(i,j)\in R_{ab}^*} A_{ij}^{(\ell)},
\end{aligned}
\end{equation}
for $a,b\in[k]$ and all $\ell \in [L].$
With this definition, Lemma~\ref{lem:moment_A_bar}, adapted from \cite{olhede2014network}, provides bounds for the moments of $(\Bar{A}_{ab}^{(\ell)})^*$. In the following, we restate the lemma and its proof below using the multi-layer notation. 

\begin{lemma}[Moments of $(\Bar{A}_{ab}^{(\ell)})^*$]
\label{lem:moment_A_bar}
% Let $f^{(\ell)}\in\text{H\"{o}lder}^{\alpha}(M_{\ell})$ be a symmetric function on $[0,1]^2 \to \mbR_0^+$.
Assuming that Condition~\ref{cond:multinethist} holds,
the following results apply to the oracle estimator in \eqref{eq:oracle_A_ab_bar}:
\begin{equation}
\label{eq:moment_A_ab_bar}
\begin{aligned}
\left|\E (\bar{A}_{ab}^{(\ell)})^* - \rho_n^{(\ell)}\Bar{f}_{ab}^{(\ell)}\right|
&\leq \frac{M_{\ell}\rho_n^{(\ell)}}{(2n)^{\alpha/2}}\{1+o(1)\},
\end{aligned}
\end{equation}
and
\begin{equation}
\label{eq:moment_A_ab_bar_var}
\begin{aligned}
\left|\var (\bar{A}_{ab}^{(\ell)})^* -\frac{\rho_n^{(\ell)}\Bar{f}_{ab}^{(\ell)} - (\rho_n^{(\ell)})^2\overline{(f^{(\ell)})^2}_{ab}}{h_{ab}^2}\right|
\leq& \frac{M_{\ell}\rho_n^{(\ell)}}{h_{ab}^2(2n)^{\alpha/2}}\{1+o(1)\}
+ \frac{(M_{\ell}\rho_n^{(\ell)})^2}{(2n)^{\alpha}},
\end{aligned}
\end{equation}
where $\Bar{f}_{ab}^{(\ell)}$ and $\overline{(f^{(\ell)})^2}_{ab}$ are defined as \eqref{def:avg_f}.
\end{lemma}
\begin{proof}
To derive layer-wise moment bounds, we apply the argument of Proposition 1 of \cite{olhede2014network} to each layer. First, we rewrite the oracle estimator \eqref{eq:oracle_A_ab_bar}. Following the notation from \cite{wolfe2013Nonparametric}, we denote by $A_{(i)(j)}^{(\ell)}$, the adjacency matrix entry corresponding to the order statistic of latent variables $\xi_{(i)}$ and $\xi_{(j)}$, where $(i)$ is the index of the $i$th smallest ordered statistic $\xi_{(i)}$ from $\bxi$. 
With this notation, \eqref{eq:oracle_A_ab_bar} can be rewritten as
\begin{align}
\label{eq:A_bar_oracle_ordered}
(\bar{A}_{ab}^{(\ell)})^* = \frac{1}{h_{ab}^2} \sum_{((i),(j))\in R_{ab}^*} A_{(i)(j)}^{(\ell)},
\end{align}
where $R_{ab}^*$ is defined as \eqref{def:R_ab_oracle}. Then, it holds that 
\begin{align*}
\E (\bar{A}_{ab}^{(\ell)})^* &= \frac{1}{h_{ab}^2} \sum_{(i,j)\in R_{ab}^*}\E A_{(i)(j)}^{(\ell)}.
\end{align*}
We also follow the notation in \cite{wolfe2013Nonparametric, olhede2014network} for the next step, defining
\begin{equation}
    \widetilde{f}_{ab}^{(\ell)} := \frac{1}{h_{ab}^2}\sum_{((i),(j))\in R_{ab}^*} f^{(\ell)}(i_n,j_n),
\end{equation}
where $i_n = i/(n+1)$ and $j_n = j/(n+1)$. This quantity will serve as a bridge for deriving bounds using the triangle inequality.

We first derive the bound for the expected value of $(\bar{A}_{ab}^{(\ell)})^*$. By the triangle inequality and Jensen's inequality, we obtain
\begin{align*}
\left|\E (\bar{A}_{ab}^{(\ell)})^* - \rho_n^{(\ell)}\Bar{f}_{ab}^{(\ell)}\right|
&\leq \left|\E (\bar{A}_{ab}^{(\ell)})^* - \rho_n^{(\ell)}\widetilde{f}_{ab}^{(\ell)}\right|
+ \left|\rho_n^{(\ell)}\widetilde {f}_{ab}^{(\ell)}- \rho_n^{(\ell)}\Bar{f}_{ab}^{(\ell)}\right|\\
&= \left|\E \{\rho_n^{(\ell)} f^{(\ell)}(\xi_{(i)},\xi_{(j)})\} - \rho_n^{(\ell)}\widetilde{f}_{ab}^{(\ell)}\right|
+ \left|\rho_n^{(\ell)}\widetilde {f}_{ab}^{(\ell)}- \rho_n^{(\ell)}\Bar{f}_{ab}^{(\ell)}\right|\\
&\leq \rho_n^{(\ell)}
\left(
\E\left|f^{(\ell)}(\xi_{(i)},\xi_{(j)}) - \widetilde{f}_{ab}^{(\ell)}\right|
+ \left|\widetilde {f}_{ab}^{(\ell)}- \Bar{f}_{ab}^{(\ell)}\right| 
\right),
\end{align*}
for each $\ell\in [L]$, where the second equality follows from $\E A_{(i)(j)}^{(\ell)} =\E\{\E(A_{(i)(j)}^{(\ell)}|\xi_{(i)},\xi_{(j)})\} = \E\{\rho_n^{(\ell)}f^{(\ell)}(\xi_{(i)},\xi_{(j)})\}$ by \eqref{eq:cond_moments}.
In the last line of the inequality above, the first and second terms are bounded by Lemma 3 and Lemma 4 of \cite{olhede2014network}, respectively. This establishes the bound in \eqref{eq:moment_A_ab_bar}.

Lastly, we derive a bound for the variance of $(\bar{A}_{ab}^{(\ell)})^*$. It holds that
\begin{align*}
\var (\Bar{A}_{ab}^{(\ell)})^* &= 
\frac{1}{h_{ab}^4}\sum_{((i),(j))\in R_{ab}^*}\sum_{((i'),(j'))\in R_{ab}^*}\cov\{A_{(i)(j)}^{(\ell)},A_{(i')(j')}^{(\ell)}\}.
\end{align*}
By Lemma 2 of \cite{olhede2014network} and triangle inequality, it holds that
\begin{align*}
&\left|\var (\bar{A}_{ab}^{(\ell)})^* -\frac{\rho_n^{(\ell)}\Bar{f}_{ab}^{(\ell)} - (\rho_n^{(\ell)})^2\overline{(f^{(\ell)})^2}_{ab}}{h_{ab}^2}\right|\\
\leq& \left|\var (\bar{A}_{ab}^{(\ell)})^* 
- \frac{\rho_n^{(\ell)}\widetilde{f}_{ab}^{(\ell)} - (\rho_n^{(\ell)})^2\widetilde{(f^{(\ell)})^2}_{ab}}{h_{ab}^2}\right| \\
&+ \left|\frac{\rho_n^{(\ell)}\widetilde{f}_{ab}^{(\ell)} - (\rho_n^{(\ell)})^2\widetilde{(f^{(\ell)})^2}_{ab}}{h_{ab}^2} -\frac{\rho_n^{(\ell)}\Bar{f}_{ab}^{(\ell)} - (\rho_n^{(\ell)})^2\overline{(f^{(\ell)})^2}_{ab}}{h_{ab}^2}\right|,
\end{align*}
where $\Bar{f}_{ab}^{(\ell)}$ and $\overline{(f^{(\ell)})^2}_{ab}$ are defined as \eqref{def:avg_f}, and
\begin{align*}
\widetilde{(f^{(\ell)})^2}_{ab} := \frac{1}{h_{ab}^2}\sum_{((i),(j))\in R_{ab}^*}\{f^{(\ell)}(i_n,j_n)\}^2.
\end{align*}
Then, the first term on the right-hand side is bounded by Lemma 2 of \cite{olhede2014network}. The second term can be further decomposed using the triangle inequality,
\begin{align*}
&\left|\frac{\rho_n^{(\ell)}\widetilde{f}_{ab}^{(\ell)} - (\rho_n^{(\ell)})^2\widetilde{(f^{(\ell)})^2}_{ab}}{h_{ab}^2} -\frac{\rho_n^{(\ell)}\Bar{f}_{ab}^{(\ell)} - (\rho_n^{(\ell)})^2\overline{(f^{(\ell)})^2}_{ab}}{h_{ab}^2}\right|\\
\leq& 
\frac{\rho_n^{(\ell)}}{h_{ab}^2}
\left|\widetilde{f}_{ab}^{(\ell)} -\Bar{f}_{ab}^{(\ell)}\right|
+
\frac{(\rho_n^{(\ell)})^2}{h_{ab}^2}
\left|\widetilde{(f^{(\ell)})^2}_{ab} - \overline{(f^{(\ell)})^2}_{ab}\right|.
\end{align*}
where the two terms on the right-hand side are bounded by Lemma 4 and 5 of \cite{olhede2014network}, respectively. This proves the desired result.
\end{proof}

Lemma~\ref{lem:layer_MISE}, a generalization of Theorem 1 of \cite{olhede2014network} under the H\"{o}lder-$\alpha$ smoothness condition, is the key tool for deriving the upper bound of the WMISE for $\widehat{\bff}^*$ in Theorem~\ref{thm:WMISE}.

\begin{lemma}
\label{lem:layer_MISE}
Assume that Condition~\ref{cond:multinethist} holds.  Then, \eqref{def:oracle_z} is bounded by 
\begin{align*}
&\E \left(\iint_{[0,1]^2} \left|(\widehat{f}^{(\ell)})^*(x,y;h)-f^{(\ell)}(x,y)\right|^2 dxdy\right)\\
&\leq 
M_{\ell}^2 
\left\{2^{\alpha}\left(
\frac{h}{n}\right)^{2\alpha}
+ \frac{2}{(2n)^{\alpha}}
+ \frac{1}{M_{\ell}^2\rho_n^{(\ell)}h^2}\right\}
\{1+o(1)\}.
\end{align*}
\end{lemma}

\begin{proof}
The proof uses the argument of Theorem 1 of \cite{olhede2014network} to each layer. For the $\ell$th layer, we obtain
\begin{align*}
&\ \E \left( \iint_{[0,1]^2}\left|(\widehat{f}^{(\ell)})^*(x,y;h)-f^{(\ell)}(x,y)\right|^2 dxdy\right)\\
=& \sum_{a,b=1}^k \iint_{\omega_{ab}}
\E\left| \left(\rho_n^{(\ell)}\right)^{-1}\left(\bar{A}_{ab}^{(\ell)}\right)^*-f^{(\ell)}(x,y)\right|^2 dxdy\\
=& \sum_{a,b=1}^k \iint_{\omega_{ab}}
\left\{\left|\left(\rho_n^{(\ell)}\right)^{-1}\E\left(\bar{A}_{ab}^{(\ell)}\right)^*-f^{(\ell)}(x,y)\right|^2 
+ \left(\rho_n^{(\ell)}\right)^{-2}\var \left(\bar{A}_{ab}^{(\ell)}\right)^*
\right\}dxdy\\
\end{align*}
from the bias-variance decomposition where $\omega_{ab}$ is the region of the $(a,b)$-th block defined in \eqref{def:omega_ab}. By Lemma~\ref{lem:moment_A_bar}, it holds that
\begin{align*}
&\iint_{\omega_{ab}}
\left\{\left|\left(\rho_n^{(\ell)}\right)^{-1}\E\left(\bar{A}_{ab}^{(\ell)}\right)^*-f^{(\ell)}(x,y)\right|^2 
+ \left(\rho_n^{(\ell)}\right)^{-2}\var \left(\bar{A}_{ab}^{(\ell)}\right)^*
\right\}dxdy\\
\leq&
\iint_{\omega_{ab}}
\left\{
\left|\left(\rho_n^{(\ell)}\right)^{-1}\E\left(\bar{A}_{ab}^{(\ell)}\right)^* - \Bar{f}_{ab}^{(\ell)}\right|^2 
+ 
\left|\Bar{f}_{ab}^{(\ell)}
-f^{(\ell)}(x,y)\right|^2 
\right\} dxdy\\
&+
\iint_{\omega_{ab}}
\left\{
\frac{\Bar{f}_{ab}^{(\ell)} - \rho_n^{(\ell)}\overline{(f^{(\ell)})^2}_{ab}}{\rho_n^{(\ell)}h_{ab}^2}
+ \frac{M_{\ell}\{1+o(1)\}}{\rho_n^{(\ell)}h_{ab}^2(2n)^{\alpha/2}}
+ \frac{M_{\ell}^2}{(2n)^{\alpha}}
\right\}dxdy\\
\leq& 
\iint_{\omega_{ab}}
\left|\Bar{f}_{ab}^{(\ell)}
-f^{(\ell)}(x,y)\right|^2  dxdy\\
&+ |\omega_{ab}|
\left\{
\frac{M_{\ell}^2\{1+o(1)\}}{(2n)^{\alpha}}
+
\frac{\Bar{f}_{ab}^{(\ell)} - \rho_n^{(\ell)}\overline{(f^{(\ell)})^2}_{ab}}{\rho_n^{(\ell)}h_{ab}^2}
+ \frac{M_{\ell}\{1+o(1)\}}{\rho_n^{(\ell)}h_{ab}^2(2n)^{\alpha/2}}
+ \frac{M_{\ell}^2}{(2n)^{\alpha}}
\right\},
\end{align*}
where $|\omega_{ab}|= h_a h_b/n^2$ as in \eqref{def:area_omega_ab}. By Lemma 1 of \cite{olhede2014network}, we have
\begin{align*}
\frac{1}{|\omega_{ab}|}\iint_{\omega_{ab}}
\left|\Bar{f}_{ab}^{(\ell)}
-f^{(\ell)}(x,y)\right|^2  dxdy 
\leq     
2^{\alpha}M_{\ell}^2\left(\frac{h}{n}\right)^{2\alpha}\left\{1+ O\left(\frac{h}{n}\right)\right\}.
\end{align*}
Moreover, it holds that
\begin{align}
&\sum_{a,b=1}^k|\omega_{ab}|\frac{\Bar{f}_{ab}^{(\ell)}}{\rho_n^{(\ell)}h_{ab}^2}
=\sum_{a,b=1}^k \frac{1}{\rho_n^{(\ell)}h_{ab}^2}\frac{|\omega_{ab}|}{|\omega_{ab}|}\iint_{\omega_{ab}}f^{(\ell)}(x,y)dxdy
= \frac{1}{\rho_n^{(\ell)}h^2}\{1+o(1)\},\label{eq:1st_lem3.2}\\
&\sum_{a,b=1}^k|\omega_{ab}|\frac{\rho_n^{(\ell)}\overline{(f^{(\ell)})^2}_{ab}}{\rho_n^{(\ell)}h_{ab}^2}
=\sum_{a,b=1}^k \frac{1}{h_{ab}^2}
% \frac{h_{ab,r}^2}{n^2}\frac{1}{|\omega_{ab}|}
\iint_{\omega_{ab}}\{f^{(\ell)}(x,y)\}^2dxdy
= \frac{1}{h^2}O(1), \label{eq:2nd_lem3.2}
\end{align}
since $\iint_{[0,1]^2} f^{(\ell)}(x,y)dxdy = 1$ for all $\ell\in [L]$. The last equality of the latter term is from the H\"{o}lder continuity of $f^{(\ell)}(x,y)$. 
Since $h= \omega(1)$, the term in \eqref{eq:2nd_lem3.2} is negligible compared to \eqref{eq:1st_lem3.2}.
Combining these results and applying $\sum_{a,b}^k |\omega_{ab}|=1$ gives the desired result.
\end{proof}

\subsection{Technical Lemma for Theorem~\ref{thm:MISE_homogeneous}}

Lemmas~\ref{lem:interlayer_covariance} and \ref{lem:moment_A_bar_homo} are used in the proof of Theorem~\ref{thm:MISE_homogeneous}, which is about homogeneous multi-network histograms. Lemma~\ref{lem:interlayer_covariance} shows that the covariance between entries from different layers becomes negligible as the number of vertices increases when all layers share the  identical $f(x,y)$.

\begin{lemma}[inter-layer covariances of $\{A_{(i)(j)}^{(\ell)}\}_{\ell=1}^L$ and $\{(\bar{A}_{ab}^{(\ell)})^*\}_{\ell=1}^L$ for homogeneous case]
\label{lem:interlayer_covariance}
Let $f \in \text{H\"{o}lder}^{\alpha}(M)$ be a symmetric function on $[0,1]^2\to \mbR_0^+$, such that $A_{ij}^{(\ell)}|\bxi \sim \textnormal{Bernoulli}\{\rho_n^{(\ell)}f(\xi_i,\xi_j)\}$ for $\ell\in [L].$ Assume that any two layers are conditionally independent given the latent variables $\bxi = (\xi_1,\ldots, \xi_n).$ Then, the inter-layer covariance between adjacency entries corresponding to possibly different vertex pairs satisfies
\begin{align}
\label{ineq:interlayer_cov_A}
\left|\cov\{A_{(i)(j)}^{(\ell)}, A_{(i')(j')}^{(\ell')}\}\right|
\leq 
\frac{\rho_n^{(\ell)}\rho_n^{(\ell')}M^2}{\{2(n+2)\}^{\alpha}}.
\end{align}
Moreover, the inter-layer covariance between histogram blocks, possibly corresponding to different blocks, satisfies
\begin{align}
\label{ineq:interlayer_cov_A_bar}
\left|\cov\{(\bar{A}_{ab}^{(\ell)})^{*}, (\bar{A}_{a'b'}^{(\ell')})^{*}\}\right|
\leq 
\frac{\rho_n^{(\ell)}\rho_n^{(\ell')}M^2}{\{2(n+2)\}^{\alpha}},
\end{align}
where $(\bar{A}_{ab}^{(\ell)})^*$ is defined in \eqref{eq:oracle_A_ab_bar}.
\end{lemma}
\begin{proof}
Using the notation in \eqref{eq:A_bar_oracle_ordered}, recall the conditional moments of $A_{(i)(i)}^{(\ell)}$ given $\bxi$ are
\begin{align*}
    \E\{A_{(i)(j)}^{(\ell)} |\bxi\} &= \rho_n^{(\ell)}f(\xi_{(i)},\xi_{(j)}),\\
    \var\{A_{(i)(j)}^{(\ell)} |\bxi\} &= \rho_n^{(\ell)}f(\xi_{(i)},\xi_{(j)})\left\{1-\rho_n^{(\ell)}f(\xi_{(i)},\xi_{(j)})\right\},
\end{align*}
due to $A_{(i)(j)}^{(\ell)}|\xi_{(i)},\xi_{(j)} \sim \textnormal{Bernoulli}\{\rho_n^{(\ell)}f(\xi_{(i)},\xi_{(j)})\}$. 
Moreover, for $\ell\neq \ell'$, it holds that
$\cov\{A_{(i)(j)}^{(\ell)}, A_{(i')(j')}^{(\ell')}|\bxi\} = 0$ by the conditional inter-layer independence assumption.
Using these conditional moments and the law of total covariance, we have 
\begin{equation}
\label{eq:cov_interlayer_same}
\begin{aligned}
\cov\{A_{(i)(j)}^{(\ell)}, A_{(i)(j)}^{(\ell')}\}   
&= \E\left[\cov\{A_{(i)(j)}^{(\ell)}, A_{(i)(j)}^{(\ell')}|\bxi\}\right]
+  \cov\left[\E\{A_{(i)(j)}^{(\ell)}|\bxi\}, \E\{A_{(i)(j)}^{(\ell')}|\bxi\}\right]\\
&= 0 + \rho_n^{(\ell)}\rho_n^{(\ell')}\var\{f(\xi_{(i)},\xi_{(j)})\},
\end{aligned}
\end{equation}
for the same pair $((i),(j))$
and
\begin{equation}
\label{eq:cov_interlayer_diff}
\begin{aligned}
\cov\{A_{(i)(j)}^{(\ell)}, A_{(i')(j')}^{(\ell')}\}   
&= \E\left[\cov\{A_{(i)(j)}^{(\ell)}, A_{(i')(j')}^{(\ell')}|\bxi\}\right]
+  \cov\left[\E\{A_{(i)(j)}^{(\ell)}|\bxi\}, \E\{A_{(i')(j')}^{(\ell')}|\bxi\}\right]\\
&= 0 + \rho_n^{(\ell)}\rho_n^{(\ell')}\cov\{f(\xi_{(i)},\xi_{(j)}),f(\xi_{(i')},\xi_{(j')}) \},
\end{aligned}
\end{equation}
for distinct pairs $((i),(j))$ and $((i'),(j')).$
Next, we follow the argument of Lemma 2 of \cite{olhede2014network} to derive upper bounds of $\var\{f(\xi_{(i)},\xi_{(j)})\}$ and $\cov\{f(\xi_{(i)},\xi_{(j)}),f(\xi_{(i')},\xi_{(j')}) \}$.
Let $i_n = i/(n+1), j_n = j/(n+1), i'_n = i'/(n+1),$ and $j'_n = j'/(n+1)$.
By the H\"{o}lder$^{\alpha}(M)$ condition on $f(x,y)$, Jensen's inequality, and Cauchy-Schwarz inequality, it holds that
\begin{align*}
\var\{f(\xi_{(i)},\xi_{(j)})\}
&\leq 
\E\{f(\xi_{(i)},\xi_{(j)}) - f(i_n, j_n)\}^2 \leq 
M^2\E\|(\xi_{(i)},\xi_{(j)}) - (i_n, j_n)\|_2^{2\alpha}
\leq \frac{M^2}{\{2(n+2)\}^{\alpha}},
\end{align*}
and
\begin{align*}
\left|\cov\{f(\xi_{(i)},\xi_{(j)}),f(\xi_{(i')},\xi_{(j')})\right|
&\leq 
\left|\E[\{f(\xi_{(i)},\xi_{(j)}) - f(i_n, j_n)\}\{f(\xi_{(i')},\xi_{(j')}) - f(i'_n, j'_n)\}]\right|\\
&\leq 
\E\left|\{f(\xi_{(i)},\xi_{(j)}) - f(i_n, j_n)\}\{f(\xi_{(i')},\xi_{(j')}) - f(i'_n, j'_n)\}\right|\\
&\leq 
M^2\E\|(\xi_{(i)},\xi_{(j)}) - (i_n, j_n)\|_2^{\alpha}
\|(\xi_{(i')},\xi_{(j')}) - (i'_n, j'_n)\|_2^{\alpha}\\
&\leq M^2 \sqrt{\E\|(\xi_{(i)},\xi_{(j)}) - (i_n, j_n)\|_2^{2\alpha}}
\sqrt{\E\|(\xi_{(i')},\xi_{(j')}) - (i'_n, j'_n)\|_2^{2\alpha}}\\
&\leq \frac{M^2}{\{2(n+2)\}^{\alpha}},
\end{align*}
which proves \eqref{ineq:interlayer_cov_A}.

It follows that 
\begin{align*}
\cov\{(\bar{A}_{ab}^{(\ell)})^{*}, (\bar{A}_{a'b'}^{(\ell')})^{*}\}
&= 
\cov\left\{
\frac{1}{h_{ab}^2} \sum_{((i),(j))\in R_{ab}^*} A_{(i)(j)}^{(\ell)},
\frac{1}{h_{a'b'}^2} \sum_{((i'),(j'))\in R_{a'b'}^*} A_{(i')(j')}^{(\ell')}
\right\}\\
&= \frac{1}{h_{ab}^2 h_{a'b'}^2} 
\sum_{((i),(j))\in R_{ab}^*}
\sum_{((i'),(j'))\in R_{a'b'}^*} 
\cov\left\{
A_{(i)(j)}^{(\ell)},
A_{(i')(j')}^{(\ell')}
\right\}.
\end{align*}
Combining this with \eqref{ineq:interlayer_cov_A} gives
\begin{align*}
\left|\cov\{(\bar{A}_{ab}^{(\ell)})^{*}, (\bar{A}_{a'b'}^{(\ell')})^{*}\}\right|
&\leq 
\frac{1}{h_{ab}^2 h_{a'b'}^2} 
\sum_{((i),(j))\in R_{ab}^*}
\sum_{((i'),(j'))\in R_{a'b'}^*} 
\left|\cov\left\{
A_{(i)(j)}^{(\ell)},
A_{(i')(j')}^{(\ell')}
\right\}\right|\\
&\leq 
\frac{\rho_n^{(\ell)}\rho_n^{(\ell')}M^2}{\{2(n+2)\}^{\alpha}},
\end{align*}
which proves \eqref{ineq:interlayer_cov_A_bar}.
\end{proof}

Lemma~\ref{lem:moment_A_bar_homo} provides bounds on the moments of homogeneous multi-network histogram using the ``oracle" weight $\tau_{\ell}^*$ derived from \Cref{sec_append:more_algorithm}.
\begin{lemma}[Moments of homogeneous multi-network histogram]
\label{lem:moment_A_bar_homo}
Assume the same condition as in Lemma~\ref{lem:interlayer_covariance}.
For the oracle homogeneous multi-network histogram $\widehat{f}_{ab}^*$ is defined by 
\begin{equation}
\label{eq:oracle_homo_multi_nethist_moments}
\begin{aligned}
\widehat{f}_{ab}^* 
&= \sum_{\ell=1}^L \tau_{\ell}^{*}\frac{(\bar{A}_{ab}^{(\ell)})^*}{\rho_n^{(\ell)}}, \quad a,b\in[k]
\end{aligned}
\end{equation}
with the oracle weights  
$
\tau_{\ell}^{*} = \frac{\rho_n^{(\ell)}}{\sum_{l=1}^L \rho_n^{(l)}},
$
and the block average $(\bar{A}_{ab}^{(\ell)})^*$ as defined in \eqref{eq:oracle_A_ab_bar}, the moments of $\widehat{f}_{ab}^*$ satisfy
\begin{equation}
\label{eq:mean_homo_multinethist}
\begin{aligned}
\left|\E \widehat{f}_{ab}^*  - \Bar{f}_{ab}\right|
&\leq \frac{M}{(2n)^{\alpha/2}}\{1+o(1)\},
\end{aligned}
\end{equation}
and
\begin{equation}
\label{eq:var_homo_multinethist}
\begin{aligned}
&\left|\var \widehat{f}_{ab}^*  
- \frac{1}{h_{ab}^2}\sum_{\ell=1}^L (\tau_{\ell}^{*})^2\left(\frac{\Bar{f}_{ab}}{\rho_n^{(\ell)}} - \overline{f^2}_{ab}\right)\right|\\
\leq&
\frac{1}{\left(\sum_{\ell=1}^L \rho_n^{(\ell)}\right)^2}
\sum_{\ell=1}^L \frac{M\rho_n^{(\ell)}}{h_{ab}^2(2n)^{\alpha/2}}\{1+o(1)\}
+ \frac{M^2}{(2n)^{\alpha}}
\end{aligned}
\end{equation}
where $\Bar{f}_{ab}$ and $\overline{f^2}_{ab}$ are defined in \eqref{def:avg_f}.
\end{lemma}
\begin{proof}
By substituting \eqref{eq:A_bar_oracle_ordered} into \eqref{eq:oracle_homo_multi_nethist_moments}, we can rewrite the oracle estimator as
\begin{equation}
\label{eq:f_star}
\widehat{f}_{ab}^* 
= \frac{1}{h_{ab}^2 \left(\sum_{\ell=1}^L\rho_n^{(\ell)}\right)} \sum_{((i),(j))\in R_{ab}^*} { \sum_{\ell=1}^L A_{(i)(j)}^{(\ell)}},
\end{equation}
where $(i)$ is the index of the $i$th smallest ordered statistic $\xi_{(i)}$ from $\bxi$.

We first consider the expected value of $\widehat{f}_{ab}^*$. It follows that 
\begin{align*}
\E \widehat{f}_{ab}^* &= \frac{1}{h_{ab}^2 \left(\sum_{\ell=1}^L\rho_n^{(\ell)}\right)} \sum_{((i),(j))\in R_{ab}^*} { \sum_{\ell=1}^L \E A_{(i)(j)}^{(\ell)}}
=\frac{1}{h_{ab}^2} \sum_{((i),(j))\in R_{ab}^*} \E{f(\xi_{(i)},\xi_{(j)})}.
\end{align*}
We also define
\begin{equation*}
    \widetilde{f}_{ab} := \frac{1}{h_{ab}^2}\sum_{((i),(j))\in R_{ab}^*}f(i_n,j_n),
\end{equation*}
where $i_n = i/(n+1)$ and $j_n = j/(n+1)$. This term serves as a bridge for deriving the upper bound via the triangle inequality.
By the triangle inequality and Jensen's inequality, we have
\begin{align*}
\left|\E \widehat{f}_{ab}^* - \Bar{f}_{ab}\right|
&\leq 
\left|\E \widehat{f}_{ab}^* - \widetilde{f}_{ab}\right|
+ 
\left|\widetilde{f}_{ab} - \Bar{f}_{ab}\right|\\
&= \left| \frac{1}{h_{ab}^2}\sum_{((i),(j))\in R_{ab}^*}\left\{\E f(\xi_{(i)},\xi_{(j)})-f(i_n,j_n)\right\}\right|
+ \left|\widetilde {f}_{ab}- \Bar{f}_{ab}\right|\\
&\leq 
\E\left| \frac{1}{h_{ab}^2}\sum_{((i),(j))\in R_{ab}^*}f(\xi_{(i)},\xi_{(j)}) - \widetilde{f}_{ab}\right|
+ \left|\widetilde {f}_{ab}- \Bar{f}_{ab}\right|.
\end{align*}
Using a similar argument as in Lemma~\ref{lem:moment_A_bar}, in the final line of the inequality above, the first and second terms in the final inequality are bounded by Lemma 3 and Lemma 4 of \cite{olhede2014network}, respectively. This establishes the bound in \eqref{eq:mean_homo_multinethist}.

Next, we derive an upper bound of the variance of $\widehat{f}_{ab}^*$. It holds that 
\begin{align*}
\var \widehat{f}^*_{ab} &= \frac{1}{\left(\sum_{\ell=1}^L \rho_n^{(\ell)}\right)^2}\var\left\{\sum_{\ell=1}^L (\bar{A}_{ab}^{(\ell)})^*\right\}.
\end{align*}
and expanding the variance of the sum,
\begin{align*}
\var\left\{\sum_{\ell=1}^L (\bar{A}_{ab}^{(\ell)})^*\right\}
&= \sum_{\ell=1}^L\var (\bar{A}_{ab}^{(\ell)})^*
+ 2\sum_{\ell < \ell'}^L
\cov\{ (\bar{A}_{ab}^{(\ell)})^*,(\bar{A}_{ab}^{(\ell')})^*\}.
\end{align*}
Using Definition of \eqref{eq:f_star}, it holds that
\begin{align*}
&\left|\var \widehat{f}_{ab}^*  
- \frac{1}{h_{ab}^2}\sum_{\ell=1}^L (\tau_{\ell}^{*})^2\left(\frac{\Bar{f}_{ab}}{\rho_n^{(\ell)}} - \overline{f^2}_{ab}\right)\right|\\
 =&\left|\frac{1}{\left(\sum_{\ell=1}^L \rho_n^{(\ell)}\right)^2}\left[\sum_{\ell=1}^L\left\{\var (\bar{A}_{ab}^{(\ell)})^*
 - \frac{\rho_n^{(\ell)}\Bar{f}_{ab} - (\rho_n^{(\ell)})^2\overline{f^2}_{ab}}{h_{ab}^2}\right\}
 + 2\sum_{\ell < \ell'}^L
\cov\{ (\bar{A}_{ab}^{(\ell)})^*,(\bar{A}_{ab}^{(\ell')})^*\}
\right]
 \right|\\
\leq& \frac{1}{\left(\sum_{\ell=1}^L \rho_n^{(\ell)}\right)^2}
\left[\sum_{\ell=1}^L\left\{\frac{M\rho_n^{(\ell)}}{h_{ab}^2(2n)^{\alpha/2}}\{1+o(1)\}+\frac{(M\rho_n^{(\ell)})^2}{(2n)^{\alpha}}\right\}
+ 2\sum_{\ell < \ell'}\frac{M^2\rho_n^{(\ell)}\rho_n^{(\ell')}}{\{2(n+2)\}^{\alpha}}
\right]\\
\leq& \frac{1}{\left(\sum_{\ell=1}^L \rho_n^{(\ell)}\right)^2}
\sum_{\ell=1}^L \frac{M\rho_n^{(\ell)}}{h_{ab}^2(2n)^{\alpha/2}}\{1+o(1)\}
+ \frac{M^2}{(2n)^{\alpha}},
\end{align*}
by \eqref{eq:moment_A_ab_bar_var} in Lemma~\ref{lem:moment_A_bar} and \eqref{ineq:interlayer_cov_A_bar} in Lemma~\ref{lem:interlayer_covariance}. 
This completes the proof.
\end{proof}

\section{Supplementary Discussion on Algorithm and Theory}
\label{sec_append:more_algorithm}

In this section, we provide the proof of Proposition~\ref{prop:MISE} and justify the weight selection strategy employed in the main text.

\subsection{Proof of Proposition~\ref{prop:MISE}}

The MISE in \eqref{eq:oracle_MISE} is bounded by the arithmetic average of layer-wise MISEs, which is established in Lemma~\ref{lem:layer_MISE},
\begin{align*}
\text{MISE}(\widehat{\bff}^*) 
&\leq \ \frac{1}{L}\sum_{\ell=1}^L\E \left( \iint_{[0,1]^2}\left|(\widehat{f}^{(\ell)})^*(x,y;h)-f^{(\ell)}(x,y)\right|^2 dxdy\right)\\
&\leq \frac{1}{L}\sum_{\ell=1}^L M_{\ell}^2 
\left\{
2^{\alpha}\left(\frac{h}{n}\right)^{2\alpha}
+ \frac{2}{(2n)^{\alpha}}
+ \frac{1}{M_{\ell}^2 \rho_n^{(\ell)}h^2}
\right\}\{1+o(1)\}\\
&= \left\{
\overline{M^2}2^{\alpha}\left(\frac{h}{n}\right)^{2\alpha}
+ \frac{2\overline{M^2}}{(2n)^{\alpha}}
+ \frac{\overline{\rho_n^{-1}} }{h^2}
\right\}\{1+o(1)\}.
\end{align*}
This establishes the upper bound in \eqref{ineq:MISE_oracle}.
Moreover, we can show that the right-hand side of the inequality above is minimized at
\begin{align*}
% 0 &= \sum_{\ell=1}^{L} \left(\frac{4M_{\ell}^2h^*}{n^2} -\frac{2}{\rho_n^{(\ell)}(h^*)^3}\right)\\
(h_{\text{MISE}}^*)^{2\alpha+2} &= \frac{n^{2\alpha}\overline{\rho_n^{-1}}}{\alpha 2^{\alpha} \overline{M^2}},
\end{align*}
which corresponds to \eqref{eq:oracle_bandwidth_MISE}. The upper bound at $h=h_{\text{MISE}}^*$ satisfies
\begin{align*}
\text{MISE}(\widehat{\bff}^*)\Big|_{h=h_{\text{MISE}}^*}
&\leq \left\{\overline{M^2}2^{\alpha}\left(\frac{h_{\text{MISE}}^*}{n}\right)^{2\alpha}
+ \frac{2\overline{M^2}}{(2n)^{\alpha}}
+ \frac{\overline{\rho_n^{-1}}}{(h_{\text{MISE}}^*)^2}
\right\}\{1+o(1)\}\\
&= \left\{\frac{\overline{M^2}2^{\alpha}}{n^{2\alpha}}n^{\frac{2\alpha^2}{\alpha+1}}\left(\frac{\overline{\rho_n^{-1}}}{\alpha 2^{\alpha} \overline{M^2}}\right)^{\frac{\alpha}{\alpha+1}}
+ \frac{2\overline{M^2}}{(2n)^{\alpha}}
+ \frac{\overline{\rho_n^{-1}}\left({\alpha 2^{\alpha} \overline{M^2}}\right)^{\frac{1}{\alpha+1}}}{n^{\frac{2\alpha}{\alpha+1}}\left({\overline{\rho_n^{-1}}}\right)^{\frac{1}{\alpha+1}}}
\right\}\{1+o(1)\}\\
&= 
\left\{\left(\alpha^{-\frac{\alpha}{\alpha+1}}+ \alpha^{\frac{1}{\alpha+1}}\right)(\overline{M^2})^{\frac{1}{\alpha+1}}
\left(\frac{2\overline{\rho_n^{-1}}}{ n^2}\right)^{\frac{\alpha}{\alpha+1}}
+ \frac{2\overline{M^2}}{(2n)^{\alpha}}
\right\}\{1+o(1)\}\\
&= O\left(\left\{\binom{n}{2}\frac{1}{\overline{\rho_n^{-1}}}\right\}^{-\alpha/(\alpha+1)}\right),
\end{align*}
which is \eqref{ineq:MISE_optimal}.

\subsection{Weights of homogeneous multi-network histograms}

This section provides the rationale for the weight selection in \eqref{eq:sol_co} for homogeneous multi-network histograms, which minimizes a proxy for the bin-wise variance. Specifically, we seek to determine weights that minimize the integrated point-wise variance of the weighted average of the layer-wise oracle estimators:
\begin{equation}
\label{eq:weights_MG}
\begin{aligned}
    {\btau}&:= \argmin_{\btau} 
    \iint_{[0.1]^2} \var\left[\sum_{\ell=1}^L \tau_{\ell}\left\{(\widehat{f}^{(\ell)})^*(x,y;h)\right\}\right]dx dy,\\
    &\text{subject to } \btau^{\mathrm{T}}\bone = 1, \tau_{\ell}\geq 0.
\end{aligned}    
\end{equation}
We decompose the right-hand side of \eqref{eq:weights_MG},
\begin{align*}
\iint_{[0.1]^2} \var\left[\sum_{\ell=1}^L \tau_{\ell}\left\{(\widehat{f}^{(\ell)})^*(x,y;h)\right\}\right]dx dy
&= \sum_{a,b=1}^k \iint_{\omega_{ab}} \var\left[\sum_{\ell=1}^L \tau_{\ell}\left\{(\widehat{f}^{(\ell)})^*(x,y;h)\right\}\right]dx dy.
\end{align*}
When $(x,y)\in \omega_{ab}$, $(\widehat{f}^{(\ell)})^*(x,y;h)= (\widehat{f}_{ab}^{(\ell)})^*= (\bar{A}_{ab}^{(\ell)})^*/\rho_n^{(\ell)}$,  where
$(\bar{A}_{ab}^{(\ell)})^*$ is defined in \eqref{eq:oracle_A_ab_bar}.
By \eqref{eq:moment_A_ab_bar_var} and \eqref{ineq:interlayer_cov_A_bar} with $\alpha = 1$, we have the variance 
\begin{align*}
\var\left(\sum_{\ell=1}^L \tau_{\ell}(\widehat{f}_{ab}^{(\ell)})^*\right)
=&
\sum_{\ell=1}^L  \tau_{\ell}^2\var \{(\widehat{f}_{ab}^{(\ell)})^*\}
+ 
\sum_{\ell\neq \ell'}^L  \tau_{\ell}\tau_{\ell'}\cov \{(\widehat{f}_{ab}^{(\ell)})^*,(\widehat{f}_{ab}^{(\ell')})^*\}
\\
=&
\sum_{\ell=1}^L  \frac{\tau_{\ell}^2}{(\rho_n^{(\ell)})^2}\var \{(\bar{A}_{ab}^{(\ell)})^*\}
+ 
\sum_{\ell\neq \ell'}^L  \frac{\tau_{\ell}\tau_{\ell'}}{\rho_n^{(\ell)}\rho_n^{(\ell')}}\cov \{(\bar{A}_{ab}^{(\ell)})^*,(\bar{A}_{ab}^{(\ell')})^*\}
\\
\leq& \frac{1}{h_{ab}^2}\sum_{\ell=1}^L \tau_{\ell}^2\left\{\frac{\Bar{f}_{ab}^{(\ell)}}{\rho_n^{(\ell)}} - \overline{(f^{(\ell)})^2}_{ab}\right\}
+ \sum_{\ell=1}^L \tau_{\ell}^2\left\{\frac{M}{h_{ab}^2(2n)^{1/2}\rho_n^{(\ell)}}\{1+o(1)\}+ \frac{M^2}{n}\right\}\\
&+ \sum_{\ell\neq \ell'} \tau_{\ell} \tau_{\ell'} \frac{M^2}{2(n+2)}\\
\leq& \frac{1}{h_{ab}^2}\sum_{\ell=1}^L \tau_{\ell}^2\frac{\Bar{f}_{ab}^{(\ell)}}{\rho_n^{(\ell)}}
+ O\left(n^{-1}\right),
\end{align*}
as long as $\rho_n^{(\ell)} = \omega(n^{-1/2})$.
Given that many real-world networks exhibit sparsity, we focus on the sparse regime where $\rho_n^{(\ell)} = o(1)$ and $\rho_n^{(\ell)} = \omega(n^{-1})$. 
Then, we could assume that $h_{ab}^2=o\{n(\rho_{n}^{(\ell)})^{-1}\}$ since $h_{a} \approx \sqrt{n}$.
The term ${\Bar{f}_{ab}^{(\ell)}}/{\rho_n^{(\ell)}}$ dominates the remaining terms within the first inequality. Hence, we simplify the objective by considering only the leading term in the second inequality. Then, we obtain an upper bound for the expression,
\begin{align*}
   \sum_{a,b=1}^k \frac{1}{h_{ab}^2}\sum_{\ell=1}^L \tau_{\ell}^2 {\frac{\Bar{f}_{ab}^{(\ell)}}{\rho_n^{(\ell)}}} 
&\leq 
\sum_{\ell=1}^L \frac{\tau_{\ell}^2}{\rho_n^{(\ell)}(\min_{a,b}h_{ab}^2)}
\sum_{a,b=1}^k   
{\Bar{f}_{ab}^{(\ell)}}\\
&\leq 
\sum_{\ell=1}^L \frac{\tau_{\ell}^2}{\rho_n^{(\ell)}(\min_{a,b}h_{ab}^2)(\min_{a,b}|\omega_{ab}|)}
\sum_{a,b=1}^k   
\iint_{\omega_{ab}}{{f}^{(\ell)}}(x,y)dx dy\\
&= O\left(\sum_{\ell=1}^L \frac{\tau_{\ell}^2}{\rho_n^{(\ell)}}\right).
\end{align*}
The second inequality above holds by the definition of $\bar{f}_{ab}^{(\ell)}$ given in Lemma~\ref{lem:interlayer_covariance}, and the last equality holds by $\iint_{[0,1]^2} f^{(\ell)}(x,y) dx dy = 1$.
Combining this result with \eqref{eq:weights_MG} and omitting constant terms independent of the layer index, we obtain the following convex optimization problem:
\begin{equation}
\label{eq:weights_approx_MG}
\begin{aligned}
    {\btau}^*&:= \argmin_{\btau} \sum_{\ell=1}^L  \frac{\tau_{\ell}^2}{\rho_n^{(\ell)}} ,\quad
    \text{subject to } \btau^{\mathrm{T}}\bone = 1, \tau_{\ell}\geq 0.
\end{aligned}    
\end{equation}
Then, \eqref{eq:weights_approx_MG} has the solution
$\tau_{\ell}^* = {\rho_n^{(\ell)}}/{(\sum_{l=1}^L \rho_n^{(l)})}.$
Motivated by this, we use the following plug-in {version} of $\tau_{\ell}^{*}$, 
\begin{align}
{\tau}_{\ell} = \frac{\widehat{\rho}_n^{(\ell)}}{\sum_{l=1}^L \widehat{\rho}_n^{(l)}},
\end{align}
for the weights of homogeneous multi-network histograms as in \eqref{eq:sol_co}.

\section{Additional Simulation Results}
\label{sec_append:sim_study}

\subsection{Additional Simulation Settings}

We present additional simulation results  corresponding to Section~\ref{sec:sim_study} of the main article. The comparisons include the multi-network histogram (mnhist) and homogeneous multi-network histogram (h-mnhist) against applying the single-layer graphon estimation methods, network histogram (nethist), sort-and-smoothing (SAS), universal singular value thresholding (USVT), and neighborhood smoothing (NBS), each applied layer-wise.

Tables~\ref{tab:homogeneous_mixed_by_n} to \ref{tab:Heterogeneous_all_sparse_by_n} present the averages and the standard deviations of the weighted mean squared errors (WMSE), defined in \eqref{eq:WMSE_global}, over 100 replications. Base functions $\{f_{k}\}_{k=1}^4$ are defined in Table~\ref{tab:sparsity_level_setup} and Figure~\ref{fig:graphons_sim}. 
For the mixed sparsity cases, Tables~\ref{tab:homogeneous_mixed_by_n} to \ref{tab:Heterogeneous_mixed_by_n} show the effect of $n$ on WMSE across three scenarios---Homogeneous, Perturbation, and Heterogeneous---defined as in Section~\ref{subsec:data_generation}. Tables~\ref{tab:homogeneous_all_dense_by_L} to \ref{tab:Heterogeneous_all_dense_by_n} display the WMSE comparisons for dense sparsity cases.  Tables~\ref{tab:homogeneous_all_sparse_by_L} to \ref{tab:Heterogeneous_all_sparse_by_n} show the corresponding results for sparse settings.
We observe a pattern consistent with that reported in the main article. The proposed methods generally outperform competing methods. 
Across all sparsity scenarios, WMSE values decrease as the number of layers $L$ or the number of vertices $n$ increases.

%%%%%%
% Tables
%%%
% Scenario 1, mixed, effect of n
\begin{table}[!h]
\caption{Comparisons of weighted mean squared errors ($\times 100$ with the standard deviation in parentheses) averaged over 100 replications across different number of vertices, $n\in\{200,400,800\}$, for multiplex networks with seven layers under Scenario 1 (Homogeneous) in the mixed sparsity case. 
Function IDs correspond to those in Table~\ref{tab:sparsity_level_setup}. Underlined methods indicate the proposed methods, and bolded values denote the smallest average WMSE in each case.}
\label{tab:homogeneous_mixed_by_n}
\centering

\begin{tabular}{ll|rrr}
\hline
ID & Method & 200 & 400 & 800\\
\hline
1 & mnhist & 5.680 (0.437) & 2.888 (0.189) & 1.524 (0.068)\\
1 & h-mnhist & \textbf{3.954 (0.271)} & \textbf{2.100 (0.109)} & \textbf{1.086 (0.056)}\\
1 & nethist & 20.131 (0.658) & 11.465 (0.279) & 6.321 (0.114)\\
1 & SAS & 5.565 (0.297) & 3.238 (0.123) & 1.836 (0.053)\\
1 & USVT & 12.253 (1.868) & 9.147 (1.032) & 6.883 (0.588)\\
1 & NBS & 31.024 (2.022) & 18.517 (1.085) & 13.073 (0.465)\\
\hline
2 & mnhist & 1.371 (0.125) & 0.676 (0.045) & 0.347 (0.015)\\
2 & h-mnhist & \textbf{1.043 (0.085)} & \textbf{0.526 (0.028)} & \textbf{0.260 (0.011)}\\
2 & nethist & 5.496 (0.221) & 3.100 (0.080) & 1.712 (0.026)\\
2 & SAS & 1.071 (0.057) & 0.579 (0.022) & 0.313 (0.008)\\
2 & USVT & 2.581 (0.555) & 1.502 (0.287) & 0.938 (0.145)\\
2 & NBS & 7.903 (0.572) & 4.473 (0.221) & 3.814 (0.186)\\
\hline
3 & mnhist & 1.312 (0.360) & 0.491 (0.156) & \textbf{0.200 (0.011)}\\
3 & h-mnhist & \textbf{0.999 (0.084)} & \textbf{0.484 (0.029)} & 0.222 (0.009)\\
3 & nethist & 5.794 (0.227) & 3.353 (0.104) & 1.831 (0.058)\\
3 & SAS & 2.016 (0.114) & 1.583 (0.062) & 1.320 (0.049)\\
3 & USVT & 3.555 (0.605) & 2.486 (0.299) & 1.285 (0.161)\\
3 & NBS & 8.699 (0.485) & 5.114 (0.208) & 4.263 (0.189)\\
\hline
4 & mnhist & 1.730 (0.177) & 0.754 (0.049) & 0.380 (0.030)\\
4 & h-mnhist & \textbf{1.481 (0.143)} & \textbf{0.706 (0.053)} & \textbf{0.318 (0.018)}\\
4 & nethist & 9.722 (0.459) & 4.818 (0.266) & 2.216 (0.064)\\
4 & SAS & 12.473 (1.301) & 11.859 (1.056) & 11.420 (0.662)\\
4 & USVT & 7.118 (0.808) & 3.875 (0.443) & 2.130 (0.198)\\
4 & NBS & 16.500 (0.942) & 9.611 (0.353) & 7.001 (0.291)\\
\hline
\end{tabular}
\end{table}
% mixed + Scenario 2
\begin{table}[!h]
\caption{Comparisons of weighted mean squared errors ($\times 100$ with the standard deviation in parentheses) averaged over 100 replications across different number of vertices, $n\in\{200,400,800\}$, for multiplex networks with seven layers under Scenario 2 (Perturbation) in the mixed sparsity case. 
Function IDs correspond to those in Table~\ref{tab:sparsity_level_setup}. Underlined methods indicate the proposed methods, and bolded values denote the smallest average WMSE in each case.}
\label{tab:Perturbed_mixed_by_n}
\centering
\begin{tabular}{lll|rrr}
\hline
ID1 & ID2 & Method & 200 & 400 & 800\\
\hline
1 & 2 & \underline{mnhist} & 1.460 (0.121) & 0.724 (0.041) & \textbf{0.362 (0.016)}\\
1 & 2 & nethist & 5.780 (0.233) & 3.246 (0.086) & 1.786 (0.031)\\
1 & 2 & SAS & \textbf{1.200 (0.065)} & \textbf{0.661 (0.022)} & 0.362 (0.010)\\
1 & 2 & USVT & 3.423 (0.635) & 2.466 (0.361) & 1.918 (0.182)\\
1 & 2 & NBS & 7.985 (0.450) & 4.982 (0.271) & 3.863 (0.157)\\
\hline
1 & 4 & \underline{mnhist} & \textbf{1.819 (0.188)} & \textbf{0.774 (0.054)} & \textbf{0.379 (0.022)}\\
1 & 4 & nethist & 9.822 (0.517) & 4.846 (0.209) & 2.264 (0.067)\\
1 & 4 & SAS & 11.883 (1.337) & 11.333 (1.061) & 11.143 (0.606)\\
1 & 4 & USVT & 8.165 (0.922) & 5.229 (0.526) & 3.616 (0.302)\\
1 & 4 & NBS & 16.370 (1.050) & 10.269 (0.416) & 6.925 (0.241)\\
\hline
3 & 4 & \underline{mnhist} & \textbf{1.748 (0.190)} & \textbf{0.743 (0.049)} & \textbf{0.373 (0.029)}\\
3 & 4 & nethist & 9.649 (0.527) & 4.790 (0.181) & 2.222 (0.094)\\
3 & 4 & SAS & 11.466 (1.163) & 11.031 (0.990) & 10.922 (0.664)\\
3 & 4 & USVT & 6.851 (0.753) & 3.803 (0.418) & 2.035 (0.268)\\
3 & 4 & NBS & 16.413 (0.842) & 9.545 (0.420) & 6.792 (0.271)\\
\hline
\end{tabular}
% \end{table}
% %%%
% % mixed + Scenario 3
% \begin{table}[!ht]
\caption{Comparisons of weighted mean squared errors ($\times 100$ with the standard deviation in parentheses) averaged over 100 replications by the number of vertices, $n\in \{200,400,800\}$, for multiplex networks with seven layers under Scenario 3 (Heterogeneous) in the mixed sparsity case. Function IDs correspond to those in Table~\ref{tab:sparsity_level_setup}. Underlined methods indicate the proposed methods, and bolded values denote the smallest average WMSE in each case.}
\label{tab:Heterogeneous_mixed_by_n}
\centering
\begin{tabular}{lll|rrr}
\hline
ID1 & ID2 & Method & 200 & 400 & 800\\
\hline
1 & 2 & \underline{mnhist} & 1.935 (0.139) & \textbf{0.958 (0.060)} & \textbf{0.480 (0.023)}\\
1 & 2 & nethist & 7.329 (0.333) & 4.203 (0.129) & 2.338 (0.046)\\
1 & 2 & SAS & \textbf{1.678 (0.101)} & 0.968 (0.043) & 0.546 (0.015)\\
1 & 2 & USVT & 5.822 (1.040) & 5.391 (0.760) & 5.368 (0.492)\\
1 & 2 & NBS & 19.182 (1.343) & 14.745 (0.678) & 12.530 (0.380)\\
\hline
1 & 4 & \underline{mnhist} & \textbf{2.259 (0.215)} & \textbf{0.992 (0.070)} & \textbf{0.448 (0.021)}\\
1 & 4 & nethist & 11.229 (0.478) & 6.216 (0.183) & 3.398 (0.064)\\
1 & 4 & SAS & 11.675 (1.199) & 11.053 (0.879) & 10.819 (0.725)\\
1 & 4 & USVT & 10.239 (1.644) & 8.878 (1.178) & 8.641 (0.660)\\
1 & 4 & NBS & 30.667 (1.650) & 23.401 (1.164) & 19.826 (0.600)\\
\hline
3 & 4 & \underline{mnhist} & \textbf{2.052 (0.236)} & \textbf{0.846 (0.058)} & \textbf{0.403 (0.022)}\\
3 & 4 & nethist & 9.974 (0.448) & 5.390 (0.185) & 2.937 (0.066)\\
3 & 4 & SAS & 10.545 (1.156) & 10.668 (0.862) & 10.468 (0.630)\\
3 & 4 & USVT & 6.799 (1.098) & 4.030 (0.614) & 2.326 (0.343)\\
3 & 4 & NBS & 32.199 (1.582) & 23.602 (1.030) & 21.041 (0.789)\\
\hline
\end{tabular}
\end{table}

%All dense + Scenario 1

\begin{table}[!h]
\vspace{10pt}
\caption{Comparisons of weighted mean squared errors ($\times 100$ with the standard deviation in parentheses) averaged over 100 replications by the number of layers, $L\in\{5,7,10\}$, for multiplex networks with 400 vertices under Scenario 1 (Homogeneous) in the dense layer case. Function IDs correspond to those in Table~\ref{tab:sparsity_level_setup}. Underlined methods indicate the proposed methods, and bolded values denote the smallest average WMSE in each case.}
\label{tab:homogeneous_all_dense_by_L}
\centering
\begin{tabular}{ll|rrr}
\hline
ID & Method & 5 & 7 & 10\\
\hline
1 & mnhist & 2.210 (0.133) & 1.843 (0.131) & 1.532 (0.126)\\
1 & h-mnhist & \textbf{1.708 (0.110)} & \textbf{1.276 (0.099)} & \textbf{0.954 (0.099)}\\
1 & nethist & 6.945 (0.199) & 6.971 (0.157) & 6.936 (0.170)\\
1 & SAS & 2.212 (0.092) & 2.223 (0.103) & 2.211 (0.092)\\
1 & USVT & 3.697 (0.170) & 3.667 (0.189) & 3.694 (0.179)\\
1 & NBS & 5.176 (0.168) & 5.175 (0.135) & 5.206 (0.187)\\
\hline
2 & mnhist & 0.435 (0.024) & 0.353 (0.022) & 0.288 (0.016)\\
2 & h-mnhist & 0.365 (0.020) & \textbf{0.277 (0.016)} & \textbf{0.200 (0.012)}\\
2 & nethist & 1.592 (0.048) & 1.590 (0.043) & 1.590 (0.040)\\
2 & SAS & \textbf{0.345 (0.013)} & 0.349 (0.013) & 0.348 (0.012)\\
2 & USVT & 0.425 (0.014) & 0.428 (0.014) & 0.428 (0.012)\\
2 & NBS & 1.264 (0.021) & 1.268 (0.024) & 1.267 (0.019)\\
\hline
3 & mnhist & 0.286 (0.020) & 0.219 (0.016) & 0.175 (0.011)\\
3 & h-mnhist & \textbf{0.263 (0.016)} & \textbf{0.191 (0.015)} & \textbf{0.135 (0.009)}\\
3 & nethist & 1.772 (0.099) & 1.775 (0.075) & 1.781 (0.073)\\
3 & SAS & 1.358 (0.064) & 1.369 (0.070) & 1.356 (0.066)\\
3 & USVT & 1.449 (0.066) & 1.460 (0.072) & 1.447 (0.069)\\
3 & NBS & 1.827 (0.027) & 1.832 (0.022) & 1.837 (0.018)\\
\hline
4 & mnhist & 0.656 (0.050) & 0.592 (0.059) & 0.535 (0.062)\\
4 & h-mnhist & \textbf{0.479 (0.044)} & \textbf{0.356 (0.036)} & \textbf{0.264 (0.021)}\\
4 & nethist & 2.531 (0.145) & 2.482 (0.151) & 2.502 (0.106)\\
4 & SAS & 11.513 (0.957) & 11.736 (1.072) & 11.662 (0.987)\\
4 & USVT & 1.957 (0.053) & 1.958 (0.042) & 1.954 (0.042)\\
4 & NBS & 3.856 (0.120) & 3.776 (0.112) & 3.746 (0.090)\\
\hline
\end{tabular}
\end{table}

%All dense + Scenario 2
\begin{table}[!h]
\caption{Comparisons of weighted mean squared errors ($\times 100$ with the standard deviation in parentheses) averaged over 100 replications by the number of layers, $L\in \{5,7,10\}$, for multiplex networks with 400 vertices under Scenario 2 (Perturbation) in the dense layer case. Function IDs correspond to those in Table~\ref{tab:sparsity_level_setup}. Underlined methods indicate the proposed methods, and bolded values denote the smallest average WMSE in each case.}
\label{tab:Perturbed_all_dense_by_L}
\centering
\begin{tabular}{lll|rrr}
\hline
ID1 & ID2 & Method & 5 & 7 & 10\\
\hline
1 & 2 & \underline{mnhist} & 0.765 (0.043) & \textbf{0.630 (0.032)} & \textbf{0.509 (0.031)}\\
1 & 2 & nethist & 2.530 (0.063) & 2.537 (0.055) & 2.519 (0.046)\\
1 & 2 & SAS & \textbf{0.659 (0.025)} & 0.651 (0.026) & 0.643 (0.029)\\
1 & 2 & USVT & 0.880 (0.031) & 0.859 (0.032) & 0.835 (0.031)\\
1 & 2 & NBS & 1.942 (0.046) & 1.925 (0.042) & 1.911 (0.039)\\
\hline
1 & 4 & \underline{mnhist} & \textbf{0.835 (0.063)} & \textbf{0.669 (0.057)} & \textbf{0.563 (0.036)}\\
1 & 4 & nethist & 3.322 (0.120) & 3.299 (0.111) & 3.296 (0.097)\\
1 & 4 & SAS & 7.806 (0.618) & 7.909 (0.682) & 7.730 (0.759)\\
1 & 4 & USVT & 2.024 (0.050) & 2.042 (0.045) & 2.049 (0.041)\\
1 & 4 & NBS & 3.533 (0.085) & 3.549 (0.101) & 3.540 (0.082)\\
\hline
3 & 4 & \underline{mnhist} & \textbf{0.427 (0.035)} & \textbf{0.365 (0.030)} & \textbf{0.328 (0.036)}\\
3 & 4 & nethist & 2.037 (0.080) & 2.039 (0.100) & 2.042 (0.091)\\
3 & 4 & SAS & 6.043 (0.458) & 5.923 (0.480) & 6.015 (0.449)\\
3 & 4 & USVT & 1.372 (0.026) & 1.405 (0.060) & 1.430 (0.028)\\
3 & 4 & NBS & 2.548 (0.045) & 2.556 (0.044) & 2.551 (0.035)\\
\hline
\end{tabular}

%All dense + Scenario 3
% \begin{table}[!ht]
\vspace{10pt}
\caption{Comparisons of weighted mean squared errors ($\times 100$ with the standard deviation in parentheses) averaged over 100 replications by the number of layers, $L\in \{5,7,10\}$, for multiplex networks with 400 vertices under Scenario 3 (Heterogeneous) in the dense layer case. Function IDs correspond to those in Table~\ref{tab:sparsity_level_setup}. Underlined methods indicate the proposed methods, and bolded values denote the smallest average WMSE in each case.}
\label{tab:Heterogeneous_all_dense_by_L}
\centering
\begin{tabular}{lll|rrr}
\hline
ID1 & ID2 & Method & 5 & 7 & 10\\
\hline
1 & 2 & \underline{mnhist} & 1.035 (0.087) & \textbf{0.809 (0.061)} & \textbf{0.601 (0.038)}\\
1 & 2 & nethist & 3.219 (0.103) & 3.113 (0.087) & 2.814 (0.071)\\
1 & 2 & SAS & \textbf{0.977 (0.064)} & 0.920 (0.048) & 0.812 (0.028)\\
1 & 2 & USVT & 1.485 (0.072) & 1.412 (0.068) & 1.235 (0.053)\\
1 & 2 & NBS & 2.505 (0.099) & 2.432 (0.106) & 2.210 (0.062)\\
\hline
1 & 4 & \underline{mnhist} & \textbf{1.164 (0.099)} & \textbf{0.915 (0.075)} & \textbf{0.718 (0.063)}\\
1 & 4 & nethist & 4.198 (0.146) & 4.082 (0.121) & 3.826 (0.105)\\
1 & 4 & SAS & 7.467 (0.554) & 7.604 (0.566) & 8.161 (0.593)\\
1 & 4 & USVT & 2.415 (0.085) & 2.382 (0.079) & 2.275 (0.067)\\
1 & 4 & NBS & 3.876 (0.116) & 3.856 (0.104) & 3.781 (0.091)\\
\hline
3 & 4 & \underline{mnhist} & \textbf{0.382 (0.025)} & \textbf{0.335 (0.025)} & \textbf{0.318 (0.028)}\\
3 & 4 & nethist & 2.164 (0.091) & 2.150 (0.079) & 2.216 (0.060)\\
3 & 4 & SAS & 5.163 (0.450) & 5.300 (0.484) & 6.012 (0.493)\\
3 & 4 & USVT & 1.467 (0.047) & 1.477 (0.040) & 1.511 (0.037)\\
3 & 4 & NBS & 2.373 (0.028) & 2.413 (0.027) & 2.522 (0.026)\\
\hline
\end{tabular}
\end{table}

%All dense + Scenario 1
%Effect of n
\begin{table}[!h]
\caption{Comparisons of weighted mean squared errors ($\times 100$ with the standard deviation in parentheses) averaged over 100 replications by the number of vertices, $n\in \{200,400,800\}$, for multiplex networks with seven layers under Scenario 1 (Homogeneous) in the dense layer case. Function IDs correspond to those in Table~\ref{tab:sparsity_level_setup}. Underlined methods indicate the proposed methods, and bolded values denote the smallest average WMSE in each case.}
\label{tab:homogeneous_all_dense_by_n}
\centering
\begin{tabular}{ll|rrr}
\hline
ID & Method & 200 & 400 & 800\\
\hline
1 & mnhist & 3.637 (0.411) & 1.843 (0.131) & 0.945 (0.067)\\
1 & h-mnhist & \textbf{2.488 (0.304)} & \textbf{1.276 (0.099)} & \textbf{0.668 (0.047)}\\
1 & nethist & 12.844 (0.509) & 6.971 (0.157) & 3.641 (0.082)\\
1 & SAS & 4.081 (0.255) & 2.223 (0.103) & 1.183 (0.044)\\
1 & USVT & 5.429 (0.289) & 3.667 (0.189) & 1.814 (0.053)\\
1 & NBS & 9.241 (0.628) & 5.175 (0.135) & 3.047 (0.059)\\
\hline
2 & mnhist & 0.713 (0.051) & 0.353 (0.022) & 0.184 (0.010)\\
2 & h-mnhist & \textbf{0.522 (0.032)} & \textbf{0.277 (0.016)} & \textbf{0.140 (0.007)}\\
2 & nethist & 2.875 (0.106) & 1.590 (0.043) & 0.856 (0.020)\\
2 & SAS & 0.658 (0.032) & 0.349 (0.013) & 0.184 (0.007)\\
2 & USVT & 0.824 (0.035) & 0.428 (0.014) & 0.229 (0.007)\\
2 & NBS & 2.179 (0.048) & 1.268 (0.024) & 0.769 (0.010)\\
\hline
3 & mnhist & 0.536 (0.156) & 0.219 (0.016) & 0.101 (0.006)\\
3 & h-mnhist & \textbf{0.421 (0.055)} & \textbf{0.191 (0.015)} & \textbf{0.089 (0.006})\\
3 & nethist & 3.182 (0.108) & 1.775 (0.075) & 0.809 (0.056)\\
3 & SAS & 1.623 (0.099) & 1.369 (0.070) & 1.199 (0.042)\\
3 & USVT & 1.824 (0.105) & 1.460 (0.072) & 0.550 (0.022)\\
3 & NBS & 3.009 (0.074) & 1.832 (0.022) & 1.088 (0.011)\\
\hline
4 & mnhist & 1.194 (0.116) & 0.592 (0.059) & 0.350 (0.036)\\
4 & h-mnhist & \textbf{0.825 (0.095)} & \textbf{0.356 (0.036)} & \textbf{0.169 (0.009)}\\
4 & nethist & 5.907 (0.543) & 2.482 (0.151) & 1.091 (0.045)\\
4 & SAS & 12.226 (1.624) & 11.736 (1.072) & 11.391 (0.752)\\
4 & USVT & 4.006 (0.170) & 1.958 (0.042) & 0.968 (0.014)\\
4 & NBS & 6.783 (0.296) & 3.776 (0.112) & 2.129 (0.026)\\
\hline
\end{tabular}
\end{table}

\begin{table}[!h]
\caption{Comparisons of weighted mean squared errors ($\times 100$ with the standard deviation in parentheses) averaged over 100 replications by the number of vertices, $n\in \{200,400,800\}$, for multiplex networks with seven layers under Scenario 2 (Perturbation) in the dense layer case. Function IDs correspond to those in Table~\ref{tab:sparsity_level_setup}. Underlined methods indicate the proposed methods, and bolded values denote the smallest average WMSE in each case.}
\label{tab:Perturbed_all_dense_by_n}
\centering
\begin{tabular}{lll|rrr}
\hline
ID1 & ID2 & Method & 200 & 400 & 800\\
\hline
1 & 2 & \underline{mnhist} & 1.251 (0.098) & \textbf{0.630 (0.032)} & \textbf{0.322 (0.015)}\\
1 & 2 & nethist & 4.672 (0.169) & 2.537 (0.055) & 1.342 (0.023)\\
1 & 2 & SAS & \textbf{1.210 (0.067)} & 0.651 (0.026) & 0.346 (0.011)\\
1 & 2 & USVT & 1.484 (0.077) & 0.859 (0.032) & 0.503 (0.012)\\
1 & 2 & NBS & 3.370 (0.167) & 1.925 (0.042) & 1.152 (0.017)\\
\hline
1 & 4 & \underline{mnhist} & \textbf{1.462 (0.148)} & \textbf{0.669 (0.057)} & \textbf{0.333 (0.026)}\\
1 & 4 & nethist & 6.639 (0.258) & 3.299 (0.111) & 1.662 (0.054)\\
1 & 4 & SAS & 9.134 (0.812) & 7.909 (0.682) & 6.722 (0.606)\\
1 & 4 & USVT & 4.008 (0.128) & 2.042 (0.045) & 1.031 (0.018)\\
1 & 4 & NBS & 6.277 (0.243) & 3.549 (0.101) & 2.067 (0.031)\\
\hline
3 & 4 & \underline{mnhist} & \textbf{0.764 (0.070)} & \textbf{0.365 (0.030)} & \textbf{0.197 (0.021)}\\
3 & 4 & nethist & 4.298 (0.258) & 2.039 (0.100) & 0.952 (0.050)\\
3 & 4 & SAS & 6.458 (0.644) & 5.923 (0.480) & 5.744 (0.329)\\
3 & 4 & USVT & 2.840 (0.126) & 1.405 (0.060) & 0.707 (0.028)\\
3 & 4 & NBS & 4.459 (0.109) & 2.556 (0.044) & 1.504 (0.017)\\
\hline
\end{tabular}

\vspace{10pt}

\caption{Comparisons of weighted mean squared errors ($\times 100$ with the standard deviation in parentheses) averaged over 100 replications by the number of vertices, $n\in \{200,400,800\}$, for multiplex networks with seven layers under Scenario 3 (Heterogeneous) in the dense layer case. Function IDs correspond to those in Table~\ref{tab:sparsity_level_setup}. Underlined methods indicate the proposed methods, and bolded values denote the smallest average WMSE in each case.}
\label{tab:Heterogeneous_all_dense_by_n}
\centering
\begin{tabular}{lll|rrr}
\hline
ID1 & ID2 & Method & 200 & 400 & 800\\
\hline
1 & 2 & \underline{mnhist} & \textbf{1.595 (0.110)} & \textbf{0.809 (0.061)} & \textbf{0.419 (0.022)}\\
1 & 2 & nethist & 5.666 (0.221) & 3.113 (0.087) & 1.641 (0.037)\\
1 & 2 & SAS & 1.666 (0.095) & 0.920 (0.048) & 0.493 (0.024)\\
1 & 2 & USVT & 2.157 (0.108) & 1.412 (0.068) & 0.752 (0.030)\\
1 & 2 & NBS & 4.203 (0.171) & 2.432 (0.106) & 1.408 (0.034)\\
\hline
1 & 4 & \underline{mnhist} & \textbf{1.898 (0.175)} & \textbf{0.915 (0.075)} & \textbf{0.466 (0.032)}\\
1 & 4 & nethist & 7.916 (0.309) & 4.082 (0.121) & 2.083 (0.052)\\
1 & 4 & SAS & 8.609 (0.838) & 7.604 (0.566) & 7.075 (0.482)\\
1 & 4 & USVT & 3.927 (0.151) & 2.382 (0.079) & 1.240 (0.034)\\
1 & 4 & NBS & 6.903 (0.256) & 3.856 (0.104) & 2.237 (0.040)\\
\hline
3 & 4 & \underline{mnhist} & \textbf{0.719 (0.070)} & \textbf{0.335 (0.025)} & \textbf{0.172 (0.010)}\\
3 & 4 & nethist & 3.935 (0.132) & 2.150 (0.079) & 0.981 (0.062)\\
3 & 4 & SAS & 5.785 (0.620) & 5.300 (0.484) & 5.140 (0.312)\\
3 & 4 & USVT & 2.325 (0.079) & 1.477 (0.040) & 0.693 (0.031)\\
3 & 4 & NBS & 4.021 (0.070) & 2.413 (0.027) & 1.466 (0.016)\\
\hline
\end{tabular}
\end{table}

%All sparse + Scenario 1
%Effect of L
\begin{table}[!h]
\vspace{10pt}
    \caption{Comparisons of weighted mean squared errors ($\times 100$ with the standard deviation in parentheses) averaged over 100 replications by the number of layers, $L\in \{5,7,10\}$, for multiplex networks with 400 vertices under Scenario 1 (Homogeneous) in the sparse layer case. Function IDs correspond to those in Table~\ref{tab:sparsity_level_setup}. Underlined methods indicate the proposed methods, and bolded values denote the smallest average WMSE in each case.}
    \label{tab:homogeneous_all_sparse_by_L}
    \centering
\begin{tabular}{ll|rrr}
\hline
ID & Method & 5 & 7 & 10\\
\hline
1 & mnhist & 8.94 (0.446) & 7.19 (0.334) & 5.81 (0.301)\\
1 & h-mnhist & 7.53 (0.364) & \textbf{5.60 (0.299)} & \textbf{4.07 (0.210)}\\
1 & nethist & 28.25 (0.863) & 28.29 (0.750) & 28.39 (0.597)\\
1 & SAS & \textbf{7.04 (0.349)} & 7.12 (0.342) & 7.08 (0.277)\\
1 & USVT & 31.23 (4.696) & 30.69 (3.713) & 30.20 (3.857)\\
1 & NBS & 86.91 (5.413) & 90.87 (5.316) & 91.97 (3.831)\\
\hline
2 & mnhist & 7.54 (0.278) & 6.17 (0.226) & 5.05 (0.160)\\
2 & h-mnhist & 6.83 (0.275) & 5.17 (0.182) & 3.85 (0.129)\\
2 & nethist & 23.48 (0.738) & 23.38 (0.650) & 23.48 (0.566)\\
2 & SAS & \textbf{3.31 (0.150)} & \textbf{3.34 (0.131)} & \textbf{3.36 (0.125)}\\
2 & USVT & 13.65 (2.646) & 12.80 (2.197) & 12.57 (1.982)\\
2 & NBS & 85.92 (3.600) & 100.77 (3.699) & 102.31 (3.268)\\
\hline
3 & mnhist & 7.96 (0.350) & 6.65 (0.304) & 5.61 (0.253)\\
3 & h-mnhist & 7.41 (0.261) & 5.82 (0.176) & 4.56 (0.148)\\
3 & nethist & 23.65 (0.676) & 23.71 (0.641) & 23.58 (0.502)\\
3 & SAS & \textbf{4.03 (0.153)} & \textbf{4.06 (0.153)} & \textbf{4.04 (0.118)}\\
3 & USVT & 14.28 (2.489) & 13.22 (2.271) & 12.45 (1.663)\\
3 & NBS & 85.49 (3.377) & 103.38 (3.539) & 104.31 (3.617)\\
\hline
4 & mnhist & 6.88 (1.839) & 5.67 (2.610) & 4.05 (1.134)\\
4 & h-mnhist & \textbf{5.88 (0.379)} & \textbf{4.16 (0.259)} & \textbf{3.01 (0.162)}\\
4 & nethist & 30.45 (1.309) & 30.64 (1.013) & 30.60 (0.936)\\
4 & SAS & 13.94 (1.095) & 13.84 (0.966) & 13.75 (1.116)\\
4 & USVT & 22.11 (2.294) & 21.88 (1.767) & 21.09 (1.674)\\
4 & NBS & 88.38 (3.539) & 105.67 (4.068) & 107.22 (3.124)\\
\hline
\end{tabular}
\end{table}

%All sparse + Scenario 2
\begin{table}[!h]
\caption{Comparisons of weighted mean squared errors ($\times 100$ with the standard deviation in parentheses) averaged over 100 replications by the number of layers, $L\in \{5,7,10\}$, for multiplex networks with 400 vertices under Scenario 2 (Perturbation) in the sparse layer case. Function IDs correspond to those in Table~\ref{tab:sparsity_level_setup}. Underlined methods indicate the proposed methods, and bolded values denote the smallest average WMSE in each case.}
\label{tab:Perturbed_all_sparse_by_L}
\centering
\begin{tabular}{lll|rrr}
\hline
ID1 & ID2 & Method & 5 & 7 & 10\\
\hline
1 & 2 & \underline{mnhist} & 8.23 (0.325) & 6.55 (0.233) & 5.23 (0.184)\\
1 & 2 & nethist & 24.85 (0.750) & 24.67 (0.612) & 24.60 (0.537)\\
1 & 2 & SAS & \textbf{4.21 (0.189)} & \textbf{4.14 (0.140)} & \textbf{4.14 (0.133)}\\
1 & 2 & USVT & 23.21 (3.394) & 20.27 (2.991) & 19.43 (2.318)\\
1 & 2 & NBS & 91.86 (3.388) & 101.81 (3.473) & 102.78 (3.293)\\
\hline
1 & 4 & \underline{mnhist} & \textbf{7.22 (0.458)} & \textbf{5.52 (0.335)} & \textbf{4.34 (0.262)}\\
1 & 4 & nethist & 29.82 (1.177) & 29.91 (1.020) & 29.89 (0.920)\\
1 & 4 & SAS & 12.58 (0.889) & 12.65 (0.871) & 12.71 (0.886)\\
1 & 4 & USVT & 28.86 (3.648) & 27.11 (2.496) & 25.74 (2.251)\\
1 & 4 & NBS & 93.54 (4.005) & 105.39 (4.268) & 106.32 (3.204)\\
\hline
3 & 4 & \underline{mnhist} & \textbf{7.67 (2.099)} & \textbf{5.42 (0.989)} & \textbf{4.21 (0.842)}\\
3 & 4 & nethist & 28.31 (1.117) & 28.65 (0.954) & 28.80 (0.692)\\
3 & 4 & SAS & 11.41 (0.835) & 11.36 (0.765) & 11.47 (0.727)\\
3 & 4 & USVT & 19.54 (2.284) & 19.24 (2.064) & 19.00 (1.917)\\
3 & 4 & NBS & 87.92 (3.233) & 104.95 (3.540) & 105.90 (3.445)\\
\hline
\end{tabular}

% \begin{table}[!ht]
\vspace{10pt}
\caption{Comparisons of weighted mean squared errors ($\times 100$ with the standard deviation in parentheses) averaged over 100 replications by the number of layers, $L\in \{5,7,10\}$, for multiplex networks with 400 vertices under Scenario 3 (Heterogeneous) in the sparse layer case. Function IDs correspond to those in Table~\ref{tab:sparsity_level_setup}. Underlined methods indicate the proposed methods, and bolded values denote the smallest average WMSE in each case.}
\label{tab:Heterogeneous_all_sparse_by_L}
\centering
\begin{tabular}{lll|rrr}
\hline
ID1 & ID2 & Method & 5 & 7 & 10\\
\hline
1 & 2 & \underline{mnhist} & 8.78 (0.425) & 6.94 (0.329) & 5.50 (0.276)\\
1 & 2 & nethist & 26.19 (0.632) & 26.11 (0.621) & 25.63 (0.554)\\
1 & 2 & SAS & \textbf{5.41 (0.253)} & \textbf{5.31 (0.225)} & \textbf{5.06 (0.202)}\\
1 & 2 & USVT & 30.10 (4.107) & 28.03 (3.711) & 26.57 (2.847)\\
1 & 2 & NBS & 92.41 (4.242) & 99.33 (3.966) & 101.70 (3.092)\\
\hline
1 & 4 & \underline{mnhist} & \textbf{7.75 (0.477)} & \textbf{5.97 (0.428)} & \textbf{4.50 (0.300)}\\
1 & 4 & nethist & 29.30 (1.324) & 29.76 (0.847) & 30.18 (0.842)\\
1 & 4 & SAS & 11.59 (0.724) & 11.63 (0.744) & 12.22 (0.746)\\
1 & 4 & USVT & 33.55 (4.259) & 31.66 (3.318) & 31.19 (2.953)\\
1 & 4 & NBS & 93.42 (4.788) & 101.46 (4.274) & 104.42 (3.455)\\
\hline
3 & 4 & \underline{mnhist} & \textbf{9.24 (2.680)} & \textbf{6.62 (2.480)} & \textbf{4.25 (0.821)}\\
3 & 4 & nethist & 26.64 (1.187) & 27.21 (1.080) & 27.64 (0.963)\\
3 & 4 & SAS & 9.75 (0.591) & 10.10 (0.652) & 10.70 (0.722)\\
3 & 4 & USVT & 17.15 (2.607) & 16.78 (1.907) & 16.50 (1.641)\\
3 & 4 & NBS & 86.38 (3.067) & 101.58 (4.027) & 104.07 (2.940)\\
\hline
\end{tabular}
\end{table}

\begin{table}[!h]
\caption{Comparisons of weighted mean squared errors ($\times 100$ with the standard deviation in parentheses) averaged over 100 replications by the number of vertices, $n\in \{200,400,800\}$, for multiplex networks with seven layers under Scenario 1 (Homogeneous) in the sparse layer case. Function IDs correspond to those in Table~\ref{tab:sparsity_level_setup}. Underlined methods indicate the proposed methods, and bolded values denote the smallest average WMSE in each case.}
\label{tab:homogeneous_all_sparse_by_n}
\centering
\begin{tabular}{ll|rrr}
\hline
ID & Method & 200 & 400 & 800\\
\hline
1 & \underline{mnhist} & 10.55 (0.690) & 7.19 (0.334) & 4.97 (0.153)\\
1 & \underline{h-mnhist} & \textbf{8.14 (0.682)} & \textbf{5.79 (0.341)} & \textbf{4.03 (0.155)}\\
1 & nethist & 36.01 (1.322) & 28.29 (0.750) & 21.97 (0.318)\\
1 & SAS & 8.96 (0.547) & 7.12 (0.342) & 5.57 (0.158)\\
1 & USVT & 25.50 (3.669) & 30.69 (3.713) & 43.57 (3.909)\\
1 & NBS & 83.63 (5.107) & 90.87 (5.316) & 109.35 (4.121)\\
\hline
2 & \underline{mnhist} & 8.61 (0.433) & 6.17 (0.226) & 4.25 (0.100)\\
2 & \underline{h-mnhist} & 6.79 (0.446) & 4.91 (0.262) & 3.46 (0.108)\\
2 & nethist & 29.79 (1.081) & 23.38 (0.650) & 18.19 (0.358)\\
2 & SAS & \textbf{4.61 (0.247)} & \textbf{3.34 (0.131)} & \textbf{2.42 (0.066)}\\
2 & USVT & 15.54 (2.843) & 12.80 (2.197) & 10.84 (1.847)\\
2 & NBS & 90.55 (4.461) & 100.77 (3.699) & 129.95 (3.450)\\
\hline
3 & \underline{mnhist} & 9.09 (0.386) & 6.65 (0.304) & 4.88 (0.139)\\
3 & \underline{h-mnhist} & 7.31 (0.412) & 5.49 (0.211) & 4.26 (0.118)\\
3 & nethist & 30.01 (1.247) & 23.71 (0.641) & 18.43 (0.360)\\
3 & SAS & \textbf{5.31 (0.258)} & \textbf{4.06 (0.153)} & \textbf{3.13 (0.071)}\\
3 & USVT & 17.31 (3.150) & 13.22 (2.271) & 10.23 (1.569)\\
3 & NBS & 93.56 (3.873) & 103.38 (3.539) & 132.25 (3.418)\\
\hline
4 & \underline{mnhist} & 8.24 (1.743) & 5.67 (2.610) & 3.06 (0.154)\\
4 & \underline{h-mnhist} & \textbf{5.73 (0.571)} & \textbf{3.76 (0.266)} & \textbf{2.42 (0.141)}\\
4 & nethist & 37.58 (1.424) & 30.64 (1.013) & 23.86 (0.954)\\
4 & SAS & 15.32 (1.279) & 13.84 (0.966) & 13.11 (0.752)\\
4 & USVT & 26.48 (2.816) & 21.88 (1.767) & 15.93 (1.377)\\
4 & NBS & 95.30 (4.621) & 105.67 (4.068) & 134.51 (3.690)\\
\hline
\end{tabular}
\end{table}

%All sparse + Scenario 2 and 3
%Effect of n
\begin{table}[!h]
\caption{Comparisons of weighted mean squared errors ($\times 100$ with the standard deviation in parentheses) averaged over 100 replications by the number of vertices, $n\in \{200,400,800\}$, for multiplex networks with seven layers under Scenario 2 (Perturbation) in the sparse layer case. Function IDs correspond to those in Table~\ref{tab:sparsity_level_setup}. Underlined methods indicate the proposed methods, and bolded values denote the smallest average WMSE in each case.}
\label{tab:Perturbed_all_sparse_by_n}
\centering
\begin{tabular}{lll|rrr}
\hline
ID1 & ID2 & Method & 200 & 400 & 800\\
\hline
1 & 2 & \underline{mnhist} & 9.46 (0.555) & 6.55 (0.233) & 4.55 (0.169)\\
1 & 2 & nethist & 31.17 (1.274) & 24.67 (0.612) & 19.33 (0.317)\\
1 & 2 & SAS & \textbf{5.47 (0.261)} & \textbf{4.14 (0.140)} & \textbf{3.18 (0.099)}\\
1 & 2 & USVT & 20.10 (3.511) & 20.27 (2.991) & 22.54 (2.115)\\
1 & 2 & NBS & 91.03 (4.580) & 101.81 (3.473) & 132.34 (3.518)\\
\hline
1 & 4 & \underline{mnhist} & \textbf{8.55 (1.014)} & \textbf{5.52 (0.335)} & \textbf{3.54 (0.155)}\\
1 & 4 & nethist & 37.06 (1.150) & 29.91 (1.020) & 22.75 (0.962)\\
1 & 4 & SAS & 14.05 (1.174) & 12.65 (0.871) & 11.78 (0.577)\\
1 & 4 & USVT & 27.93 (3.313) & 27.11 (2.496) & 26.47 (2.024)\\
1 & 4 & NBS & 94.61 (4.091) & 105.39 (4.268) & 136.70 (3.868)\\
\hline
3 & 4 & \underline{mnhist} & \textbf{9.53 (2.740)} & \textbf{5.42 (0.989)} & \textbf{3.38 (0.918)}\\
3 & 4 & nethist & 35.23 (1.314) & 28.65 (0.954) & 22.21 (0.673)\\
3 & 4 & SAS & 12.51 (1.166) & 11.36 (0.765) & 10.54 (0.480)\\
3 & 4 & USVT & 22.24 (2.152) & 19.24 (2.064) & 14.96 (1.402)\\
3 & 4 & NBS & 94.60 (3.546) & 104.95 (3.540) & 134.37 (3.880)\\
\hline
\end{tabular}
% \end{table}

% \begin{table}[!ht]
\vspace{10pt}
\caption{Comparisons of weighted mean squared errors ($\times 100$ with the standard deviation in parentheses) averaged over 100 replications by the number of vertices, $n\in \{200,400,800\}$, for multiplex networks with seven layers under Scenario 3 (Heterogeneous) in the sparse layer case. Function IDs correspond to those in Table~\ref{tab:sparsity_level_setup}. Underlined methods indicate the proposed methods, and bolded values denote the smallest average WMSE in each case.}
\label{tab:Heterogeneous_all_sparse_by_n}
\centering
\begin{tabular}{lll|rrr}
\hline
ID1 & ID2 & Method & 200 & 400 & 800\\
\hline
1 & 2 & \underline{mnhist} & 10.28 (0.697) & 6.94 (0.329) & 4.77 (0.176)\\
1 & 2 & nethist & 33.13 (1.287) & 26.11 (0.621) & 20.32 (0.327)\\
1 & 2 & SAS & \textbf{6.85 (0.411)} & \textbf{5.31 (0.225)} & \textbf{4.14 (0.110)}\\
1 & 2 & USVT & 23.87 (4.031) & 28.03 (3.711) & 37.28 (2.489)\\
1 & 2 & NBS & 86.60 (4.702) & 99.33 (3.966) & 125.72 (3.058)\\
\hline
1 & 4 & \underline{mnhist} & \textbf{8.97 (0.739)} & \textbf{5.97 (0.428)} & \textbf{3.90 (0.170)}\\
1 & 4 & nethist & 37.07 (1.387) & 29.76 (0.847) & 22.26 (1.000)\\
1 & 4 & SAS & 13.21 (1.084) & 11.63 (0.744) & 10.63 (0.468)\\
1 & 4 & USVT & 29.58 (4.035) & 31.66 (3.318) & 40.00 (2.853)\\
1 & 4 & NBS & 89.80 (5.098) & 101.46 (4.274) & 127.97 (3.354)\\
\hline
3 & 4 & \underline{mnhist} & \textbf{10.14 (2.436)} & \textbf{6.62 (2.480)} & \textbf{3.89 (1.952)}\\
3 & 4 & nethist & 33.77 (1.377) & 27.21 (1.080) & 20.24 (0.905)\\
3 & 4 & SAS & 11.22 (0.952) & 10.10 (0.652) & 9.19 (0.425)\\
3 & 4 & USVT & 21.41 (2.900) & 16.78 (1.907) & 12.61 (1.606)\\
3 & 4 & NBS & 93.61 (4.250) & 101.58 (4.027) & 132.95 (3.475)\\
\hline
\end{tabular}
\end{table}

\subsection{Layer-wise Estimation Error}

Tables~\ref{tab:layerwise_homogeneous}--\ref{tab:layerwise_heterogeneous} report the layer-wise MSEs obtained from the results in Tables~\ref{tab:homogeneous_mixed_by_L}--\ref{tab:Heterogeneous_mixed_by_L} of the main manuscript when $L=10$. Overall, the proposed methods show smaller layer-wise MSE in the sparsest layer relative to the competing methods. This suggests that sharing information across layers is effective in graphon estimation when the layers are sparse. 
In the densest layer, the proposed methods achieve estimation errors that are smaller than or comparable to those of the competing methods. 

\begin{table}
\caption{
Layer-wise MSEs ($\times 100$ with the standard deviation in parentheses) of the sparsest layer ($\ell = 1$), middle layer ($\ell = 5$), and densest layer ($\ell = 10$) in the mixed sparsity case under the Homogeneous scenario with $n=400$ and $L=10$.
All values are averaged over 100 replications.
}
\label{tab:layerwise_homogeneous}
\centering
\begin{tabular}{ll|ccc}
\hline
ID & Method & $\ell = 1$ & $\ell = 5$ & $\ell = 10$ \\
\hline
1 & \underline{mnhist} & 12.168 (1.842) & 2.604 (0.313) & \textbf{1.519 (0.138)} \\
1 & \underline{h-mnhist} & \textbf{1.523 (0.094)} & \textbf{1.523 (0.094)} & 1.523 (0.094) \\
1 & nethist & 64.052 (4.607) & 13.392 (0.673) & 5.700 (0.342) \\
1 & SAS & 12.096 (1.045) & 3.930 (0.305) & 1.847 (0.135) \\
1 & USVT & 172.582 (33.878) & 5.490 (0.330) & 3.332 (0.189) \\
1 & NBS & 260.249 (26.601) & 12.707 (0.605) & 3.995 (0.137) \\
\hline
2 & \underline{mnhist} & 9.419 (1.722) & 0.620 (0.088) & 0.245 (0.026) \\
2 & \underline{h-mnhist} & \textbf{0.397 (0.024)} & \textbf{0.397 (0.024)} & 0.397 (0.024) \\
2 & nethist & 54.502 (3.722) & 3.921 (0.197) & 0.806 (0.057) \\
2 & SAS & 4.853 (0.792) & 0.828 (0.065) & \textbf{0.197 (0.017)} \\
2 & USVT & 82.381 (32.408) & 1.072 (0.077) & 0.236 (0.017) \\
2 & NBS & 211.654 (23.075) & 3.434 (0.097) & 0.630 (0.027) \\
\hline
3 & \underline{mnhist} & 5.092 (0.991) & 0.408 (0.065) & \textbf{0.204 (0.017)} \\
3 & \underline{h-mnhist} & \textbf{0.342 (0.022)} & \textbf{0.342 (0.022)} & 0.342 (0.022) \\
3 & nethist & 55.621 (4.721) & 4.500 (0.211) & 0.680 (0.079) \\
3 & SAS & 5.271 (0.734) & 1.800 (0.081) & 1.207 (0.059) \\
3 & USVT & 75.187 (34.872) & 2.088 (0.091) & 1.260 (0.059) \\
3 & NBS & 205.451 (22.043) & 4.365 (0.123) & 0.942 (0.043) \\
\hline
4 & \underline{mnhist} & 3.491 (0.686) & 0.699 (0.073) & 0.512 (0.041) \\
4 & \underline{h-mnhist} & \textbf{0.502 (0.033)} & \textbf{0.502 (0.033)} & \textbf{0.502 (0.033)} \\
4 & nethist & 65.692 (4.556) & 5.339 (0.459) & 1.192 (0.127) \\
4 & SAS & 15.162 (1.234) & 12.044 (1.013) & 11.098 (1.083) \\
4 & USVT & 92.422 (28.221) & 3.692 (0.207) & 1.065 (0.059) \\
4 & NBS & 219.194 (20.883) & 7.597 (0.266) & 1.938 (0.069) \\
\hline
\end{tabular}
\end{table}

\begin{table}
\caption{Layer-wise MSEs ($\times 100$ with the standard deviation in parentheses) of the sparsest layer ($\ell = 1$), middle layer ($\ell = 5$), and densest layer ($\ell = 10$) in the mixed sparsity case under the Perturbed scenario with $n=400$ and $L=10$.
All values are averaged over 100 replications.}
\label{tab:layerwise_perturbation}
\centering
\begin{tabular}{lll|ccc}
\hline
ID1 & ID2 & Method & $\ell = 1$ & $\ell = 5$ & $\ell = 10$ \\
\hline
1 & 2 & \underline{mnhist} & \textbf{11.822 (1.865)} & \textbf{0.638 (0.086)} & 0.235 (0.023) \\
1 & 2 & nethist & 64.011 (4.429) & 3.928 (0.194) & 0.804 (0.044) \\
1 & 2 & SAS & 12.156 (1.168) & 0.821 (0.064) & \textbf{0.197 (0.016)} \\
1 & 2 & USVT & 170.323 (38.738) & 1.045 (0.077) & 0.236 (0.016) \\
1 & 2 & NBS & 261.472 (24.961) & 3.389 (0.105) & 0.627 (0.023) \\
\hline
1 & 4 & \underline{mnhist} & \textbf{3.562 (0.829)} & \textbf{0.699 (0.063)} & \textbf{0.522 (0.038)} \\
1 & 4 & nethist & 64.640 (4.026) & 5.323 (0.554) & 1.178 (0.128) \\
1 & 4 & SAS & 12.179 (1.075) & 11.570 (0.915) & 11.303 (1.079) \\
1 & 4 & USVT & 179.559 (33.793) & 3.682 (0.192) & 1.072 (0.059) \\
1 & 4 & NBS & 263.976 (25.473) & 7.647 (0.269) & 1.938 (0.075) \\
\hline
3 & 4 & \underline{mnhist} & \textbf{2.957 (0.634)} & \textbf{0.687 (0.070)} & \textbf{0.526 (0.039)} \\
3 & 4 & nethist & 55.073 (4.403) & 5.472 (0.490) & 1.156 (0.121) \\
3 & 4 & SAS & 5.333 (0.771) & 11.234 (0.986) & 11.297 (1.152) \\
3 & 4 & USVT & 76.066 (29.316) & 3.721 (0.205) & 1.071 (0.059) \\
3 & 4 & NBS & 211.297 (23.506) & 7.699 (0.292) & 1.937 (0.068) \\
\hline
\end{tabular}
% \end{table}
\vspace{10pt}
% \begin{table}
\caption{Layer-wise MSEs ($\times 100$ with the standard deviation in parentheses) of the sparsest layer ($\ell = 1$), middle layer ($\ell = 5$), and densest layer ($\ell = 10$) in the mixed sparsity case under the Heterogeneous scenario with $n=400$ and $L=10$.
All values are averaged over 100 replications.}
\label{tab:layerwise_heterogeneous}
\centering
\begin{tabular}{lll|ccc}
\hline
ID1 & ID2 & Method & $\ell = 1$ & $\ell = 5$ & $\ell = 10$ \\
\hline
1 & 2 & \underline{mnhist} & \textbf{10.924 (1.832)} & \textbf{5.440 (0.681)} & 0.211 (0.022) \\
1 & 2 & nethist & 64.893 (4.848) & 30.850 (1.755) & 0.807 (0.054) \\
1 & 2 & SAS & 12.108 (1.082) & 7.796 (0.602) & \textbf{0.196 (0.014)} \\
1 & 2 & USVT & 176.431 (35.376) & 25.827 (11.484) & 0.234 (0.017) \\
1 & 2 & NBS & 262.887 (24.748) & 127.427 (13.550) & 0.628 (0.025) \\
\hline
1 & 4 & \underline{mnhist} & \textbf{4.317 (0.903)} & \textbf{2.034 (0.302)} & \textbf{0.508 (0.032)} \\
1 & 4 & nethist & 64.178 (4.386) & 30.947 (1.612) & 1.184 (0.126) \\
1 & 4 & SAS & 12.395 (1.053) & 7.802 (0.661) & 11.122 (0.958) \\
1 & 4 & USVT & 175.905 (34.208) & 29.257 (10.121) & 1.066 (0.058) \\
1 & 4 & NBS & 257.268 (25.226) & 127.156 (16.785) & 1.939 (0.068) \\
\hline
3 & 4 & \underline{mnhist} & \textbf{3.215 (0.691)} & \textbf{1.298 (0.310)} & \textbf{0.522 (0.044)} \\
3 & 4 & nethist & 56.157 (4.226) & 25.691 (1.593) & 1.185 (0.130) \\
3 & 4 & SAS & 5.356 (0.700) & 4.341 (0.334) & 11.166 (1.014) \\
3 & 4 & USVT & 75.888 (31.850) & 10.274 (1.690) & 1.067 (0.056) \\
3 & 4 & NBS & 210.119 (21.953) & 134.442 (17.882) & 1.937 (0.070) \\
\hline
\end{tabular}
\end{table}

\subsection{Robustness of the bandwidth to sparse layers}
\label{subsec:MISE_WMISE_bandwidth_comp}

In Section~E.3.1, we investigate how the selected bandwidths are affected by the addition of sparse layers, and in Section~E.3.2, we report the estimation errors of the two-stage estimator.

We first describe the common simulation settings used in both subsections. We generate multiplex networks with five dense layers. Specifically, $f^{(\ell)}$ is determined based on the functions introduced in Table~1 and the three scenarios defined in Section~5.1. Under the Perturbed and Heterogeneous scenarios, we use the graphon ID pairs $(1,2)$, $(1,4)$, and $(3,4)$. The sparsity parameter $\rho_n^{(\ell)}$ is set within the same range as in the all-dense setting.

We add $m\in\{0,2,4,8\}$ near-empty layers with the same graphon structure, where $\rho_n^{(\ell)}=0.5/n$. Under the Perturbed and Heterogeneous scenarios, the near-empty layers are assigned to the first graphon in each pair and equally to the two graphons, respectively. All results are based on $n=400$ and averaged over 100 replications. Tables~\ref{tab:stress}, \ref{tab:stress_perturbed}, and \ref{tab:stress_heterogeneous} summarize the simulation results under the Homogeneous, Perturbed, and Heterogeneous scenarios, respectively.

\begin{table}
\caption{
(i) The selected bandwidths and dense-layer MSEs under the WMISE-based ($\widehat{h}$) and MISE-based ($\widehat{h}_{\mathrm{MISE}}$) bandwidth selections, and (ii) the WMSE $(\times 100)$ of the multi-network histogram (mnhist) and the two-stage estimator. Results are based on multiplex networks with five dense layers and $n=400$ vertices under the Homogeneous scenario, where $m\in\{0,2,4,8\}$ sparse layers are added. All values are averaged over 100 replications.
}
\label{tab:stress}
\centering
\begin{tabular}{l|cccc}
\hline
 & $m=0$ & $m=2$ & $m=4$ & $m=8$ \\
\hline
\multicolumn{5}{c}{Homogeneous $f_1$} \\
\hline
$\widehat{h}$ & 24.7 & 26.7 & 28.2 & 31.1 \\
$\widehat{h}_{MISE}$ & 24.8 & 63.2 & 64.9 & 65.9 \\
% \hline
Dense-layer MSE, $\widehat{h}$ & 2.24 & 2.02 & 1.88 & 1.68 \\
Dense-layer MSE, $\widehat{h}_{MISE}$ & 2.25 & 1.21 & 1.23 & 1.21 \\
\hline
WMSE, mnhist & 2.22 & 2.42 & 2.61 & 2.94 \\
WMSE, two-stage & 2.15 & 2.04 & 1.97 & 1.87 \\
\hline
\multicolumn{5}{c}{Homogeneous $f_2$} \\
\hline
$\widehat{h}$ & 31.5 & 34.1 & 35.8 & 38.5 \\
$\widehat{h}_{MISE}$ & 31.7 & 75.9 & 67.7 & 62.8 \\
% \hline
Dense-layer MSE, $\widehat{h}$ & 0.46 & 0.42 & 0.38 & 0.34 \\
Dense-layer MSE, $\widehat{h}_{MISE}$ & 0.46 & 0.22 & 0.24 & 0.22 \\
\hline
WMSE, mnhist & 0.43 & 0.54 & 0.62 & 0.76 \\
WMSE, two-stage & 0.42 & 0.41 & 0.41 & 0.43 \\
\hline
\multicolumn{5}{c}{Homogeneous $f_3$} \\
\hline
$\widehat{h}$ & 43.0 & 45.5 & 47.1 & 49.5 \\
$\widehat{h}_{MISE}$ & 42.8 & 88.6 & 72.6 & 64.5 \\
% \hline
Dense-layer MSE, $\widehat{h}$ & 0.30 & 0.27 & 0.29 & 0.29 \\
Dense-layer MSE, $\widehat{h}_{MISE}$ & 0.29 & 0.45 & 0.36 & 0.35 \\
\hline
WMSE, mnhist & 0.29 & 0.34 & 0.41 & 0.51 \\
WMSE, two-stage & 0.43 & 0.32 & 0.35 & 0.36 \\
\hline
\multicolumn{5}{c}{Homogeneous $f_4$} \\
\hline
$\widehat{h}$ & 43.8 & 46.1 & 47.4 & 50.1 \\
$\widehat{h}_{MISE}$ & 43.4 & 77.7 & 68.6 & 63.6 \\
% \hline
Dense-layer MSE, $\widehat{h}$ & 0.69 & 0.70 & 0.71 & 0.74 \\
Dense-layer MSE, $\widehat{h}_{MISE}$ & 0.68 & 2.20 & 1.32 & 1.11 \\
\hline
WMSE, mnhist & 0.67 & 0.77 & 0.84 & 1.00 \\
WMSE, two-stage & 0.91 & 0.83 & 0.80 & 0.89 \\
\hline
\end{tabular}
\end{table}

\begin{table}
\caption{
(i) The selected bandwidths and dense-layer MSEs under the WMISE-based ($\widehat{h}$) and MISE-based ($\widehat{h}_{\mathrm{MISE}}$) bandwidth selections, and (ii) the WMSE $(\times 100)$ of the multi-network histogram (mnhist) and the two-stage estimator. Results are based on multiplex networks with five dense layers and $n=400$ vertices under the Perturbed scenario, where $m\in\{0,2,4,8\}$ sparse layers are added. All values are averaged over 100 replications.}
\label{tab:stress_perturbed}
\centering
\begin{tabular}{l|cccc}
\hline
 & $m=0$ & $m=2$ & $m=4$ & $m=8$ \\
\hline
\multicolumn{5}{c}{{Perturbed (1, 2)}} \\
\hline
$\widehat{h}$ & 27.7 & 30.3 & 32.2 & 34.7 \\
$\widehat{h}_{MISE}$ & 27.0 & 71.0 & 70.6 & 67.8 \\
Dense-layer MSE, $\widehat{h}$ & 1.11 & 1.00 & 0.93 & 0.86 \\
Dense-layer MSE, $\widehat{h}_{MISE}$ & 1.13 & 0.73 & 0.78 & 0.70 \\
\hline
WMSE, mnhist & 0.77 & 0.88 & 0.98 & 1.16 \\
WMSE, two-stage & 0.75 & 0.72 & 0.71 & 0.71 \\
\hline
\multicolumn{5}{c}{{Perturbed (1, 4)}} \\
\hline
$\widehat{h}$ & 32.6 & 35.1 & 37.0 & 40.1 \\
$\widehat{h}_{MISE}$ & 30.8 & 73.5 & 70.2 & 70.9 \\
Dense-layer MSE, $\widehat{h}$ & 0.89 & 0.83 & 0.77 & 0.75 \\
Dense-layer MSE, $\widehat{h}_{MISE}$ & 0.95 & 1.21 & 1.18 & 1.17 \\
\hline
WMSE, mnhist & 0.84 & 0.95 & 1.03 & 1.20 \\
WMSE, two-stage & 1.07 & 0.99 & 0.96 & 0.95 \\
\hline
\multicolumn{5}{c}{{Perturbed (3, 4)}} \\
\hline
$\widehat{h}$ & 44.2 & 46.9 & 47.8 & 49.8 \\
$\widehat{h}_{MISE}$ & 44.2 & 92.8 & 71.0 & 63.0 \\
Dense-layer MSE, $\widehat{h}$ & 0.41 & 0.41 & 0.42 & 0.44 \\
Dense-layer MSE, $\widehat{h}_{MISE}$ & 0.42 & 1.52 & 0.77 & 0.57 \\
\hline
WMSE, mnhist & 0.42 & 0.49 & 0.55 & 0.67 \\
WMSE, two-stage & 1.27 & 0.80 & 0.76 & 0.67 \\
\hline
\end{tabular}
\end{table}

\begin{table}
\caption{
(i) The selected bandwidths and dense-layer MSEs under the WMISE-based ($\widehat{h}$) and MISE-based ($\widehat{h}_{\mathrm{MISE}}$) bandwidth selections, and (ii) the WMSE $(\times 100)$ of the multi-network histogram (mnhist) and the two-stage estimator. Results are based on multiplex networks with five dense layers and $n=400$ vertices under the Heterogenous scenario, where $m\in\{0,2,4,8\}$ sparse layers are added. All values are averaged over 100 replications.}
\label{tab:stress_heterogeneous}
\centering
\begin{tabular}{l|cccc}
\hline
 & $m=0$ & $m=2$ & $m=4$ & $m=8$ \\
\hline
\multicolumn{5}{c}{{Heterogeneous (1, 2)}} \\
\hline
$\widehat{h}$ & 26.6 & 29.0 & 30.5 & 33.1 \\
$\widehat{h}_{MISE}$ & 25.8 & 64.7 & 64.0 & 64.8 \\
Dense-layer MSE, $\widehat{h}$ & 1.59 & 1.43 & 1.34 & 1.27 \\
Dense-layer MSE, $\widehat{h}_{MISE}$ & 1.66 & 1.19 & 1.09 & 1.18 \\
\hline
WMSE, mnhist & 1.02 & 1.13 & 1.25 & 1.48 \\
WMSE, two-stage & 0.97 & 0.93 & 0.91 & 0.93 \\
\hline
\multicolumn{5}{c}{{Heterogeneous (1, 4)}} \\
\hline
$\widehat{h}$ & 28.1 & 30.5 & 32.3 & 35.0 \\
$\widehat{h}_{MISE}$ & 26.9 & 66.0 & 66.4 & 63.8 \\
Dense-layer MSE, $\widehat{h}$ & 1.25 & 1.09 & 1.02 & 0.95 \\
Dense-layer MSE, $\widehat{h}_{MISE}$ & 1.33 & 1.16 & 1.16 & 1.07 \\
\hline
WMSE, mnhist & 1.14 & 1.25 & 1.38 & 1.61 \\
WMSE, two-stage & 1.26 & 1.23 & 1.22 & 1.20 \\
\hline
\multicolumn{5}{c}{{Heterogeneous (3, 4)}} \\
\hline
$\widehat{h}$ & 43.5 & 46.0 & 47.7 & 49.6 \\
$\widehat{h}_{MISE}$ & 43.6 & 85.6 & 74.8 & 63.1 \\
Dense-layer MSE, $\widehat{h}$ & 0.40 & 0.41 & 0.41 & 0.42 \\
Dense-layer MSE, $\widehat{h}_{MISE}$ & 0.40 & 1.15 & 0.91 & 0.63 \\
\hline
WMSE, mnhist & 0.39 & 0.45 & 0.52 & 0.62 \\
WMSE, two-stage & 1.75 & 0.82 & 0.88 & 0.72 \\
\hline
\end{tabular}
\end{table}

\subsubsection{Effect of near-empty layers on the WMISE bandwidth}

To assess the sensitivity of the WMISE- and MISE-based bandwidth selection methods to the addition of sparse layers, we examine how the selected bandwidths and dense-layer MSEs change.
The WMISE-based bandwidth increases only mildly across all graphons and scenarios considered (13--26\% from $m=0$ to $m=8$). In contrast, the MISE-based bandwidth increases much more substantially, ranging from 43\% to 166\%, as it does not downweight the near-empty layers. When using the WMISE-based bandwidth, the impact of adding near-empty layers on the dense-layer MSE varies across graphons. Under the Homogeneous scenario, the MSE changes across graphons, ranging from a decrease of 27\% to an increase of 7\%. The MISE-based bandwidth selection has a larger impact on dense-layer estimation performance. Relative to the WMISE-based bandwidth at the same $m$, the dense-layer MSE under the MISE-based bandwidth ranges from a decrease of 46\% to an increase of 214\% across graphons. 
The same pattern is observed under the Perturbed and Heterogeneous scenarios.

\subsubsection{Two-stage layer-specific bandwidth variant}
\label{subsubsec:two-stage-variant}

We implement a two-stage method that estimates labels jointly across all layers while selecting the bandwidth separately for each layer. Stage 1 fits the multi-network histogram. Stage 2 merges groups based on layer-specific bandwidth values through the following three steps.

\begin{enumerate}
\item \textbf{Group ordering.} We calculate the weighted average of $\widehat{f}_{ab}^{(\ell)}$ with the weight
  $w_\ell = \widehat\rho^{(\ell)}/\sum_{\ell=1}^L \widehat\rho^{(\ell)}$, denoted
  $\bar f_{ab} = \sum_{\ell=1}^L w_\ell\, \widehat f_{ab}^{(\ell)}$. Then for the group size vector $\mathbf s=(s_1,\dots,s_k)$, we compute the group-wise expected degree $d_a =  \sum_{b=1}^k \bar f_{ab} s_b$ and order the $k$ groups by increasing $d_a$. The weighted average of the expected degrees across layers serves as a one-dimensional proxy for the unobserved ordering of the latent variable.
\item \textbf{Target number of groups per layer.} For each layer $\ell$, we compute the bandwidth $\widehat{h}_\ell = \sqrt n(2\widehat M_\ell^2\widehat\rho^{(\ell)})^{-1/4}\vee \widehat{h}$, and set the target number of groups to $k_\ell = \lfloor n/\widehat{h}_\ell \rfloor$.
\item \textbf{Merging.} We merge the ordered $k$ groups into $k_\ell$ consecutive blocks while preserving their order, with as equal a number of original groups as possible, giving a map $m_\ell : \{1,\dots,k\}\to\{1,\dots,k_\ell\}$. Layer $\ell$'s merged labels are $\widehat z^{(\ell)}_i = m_\ell(\widehat z_i)$, and the block probabilities are recomputed under this merged partition.
\end{enumerate}

This two-stage extension is a simple implementation of the layer-specific bandwidth idea. In this approach, both the group ordering and merging are determined by fixed rules rather than optimization. This method improves WMSE for networks consisting of structurally similar layers, but it may increase WMSE for graphon pairs with substantially different structures (Tables~\ref{tab:stress_perturbed} and \ref{tab:stress_heterogeneous}).
This behaviour arises because the ordering criterion used in this method assumes a monotone relationship between the expected degree and the latent position. This assumption is appropriate for graphons whose expected degree is approximately monotone in the latent position, but it does not hold for graphons with non-monotone or constant expected degrees. In such cases, degree-based ordering may fail to adequately approximate the ordering of the latent positions, which is related to the well-known identifiability issue in graphon models \cite{sogan2026Degree}. 
Therefore, in both scenarios, WMSE improves for graphon combinations where both graphons have approximately monotone expected-degree functions, whereas WMSE increases for combinations where neither graphon satisfies this condition. Intermediate cases, where only one graphon satisfies the condition, show mixed effects.

An approach to address this issue would be to determine the merging criterion based on profile likelihood or the similarity of estimated block heights. Such an approach is conceptually related to the block-merging method proposed by \cite{verdeyme2024Hybrid}, which merges blocks with similar estimated heights to form an irregular partition. 

\subsection{Bandwidth Selection}
\label{subsec:bandwidth_select}

We conduct additional simulations to further evaluate the performance of the bandwidth selection procedure based on Algorithm~\ref{alg:bw_sel}. We consider the following estimators:
\begin{itemize}
\item Baseline: Estimating all edge probabilities by the layer-wise overall density $\widehat{\rho}^{(\ell)}$. Due to the normalization $\iint f^{(\ell)}(x,y)dxdy=1$, this is equivalent to the constant estimator $\widehat{f}^{(\ell)}\equiv 1$ and serves as a benchmark.
\item mnhist: The multi-network histogram using the labels $\widehat{\mathbf{z}}$ obtained from Algorithm~\ref{alg:multi_nethist_combined_short}.
\item mnhist-oracle (\eqref{eq:oracle_f} of the manuscript): The multi-network histogram using the oracle labels $\widetilde{\mathbf{z}}^*$ defined by the ranks of the true latent variables (\eqref{def:oracle_z} of the manuscript), which allows us to assess the impact of label estimation errors.
\end{itemize}
For the bandwidth selection, we consider two bandwidths:
\begin{itemize}
\item $\widehat{h}$ (plug-in bandwidth): The bandwidth obtained from Algorithm~\ref{alg:bw_sel} based on \eqref{eq:oracle_bandwidth_WMISE}, which is derived from Theorem~\ref{thm:WMISE} by minimizing the WMISE upper bound of the oracle estimator.
\item $h_{oracle}$ (oracle bandwidth): The bandwidth selected for each replication by refitting mnhist or mnhist-oracle over a grid of values $\{0.5,0.75,1,1.5,2,3\}\times\widehat{h}$ around $\widehat{h}$ and choosing the bandwidth with the smallest realized WMSE under the true graphon. This serves as an empirical benchmark unavailable in real data applications.
\end{itemize}

To separate the effects of bandwidth selection and label estimation error, we compare the following five estimators: (i) mnhist, (ii) mnhist-oracle + $\widehat{h}$, (iii) mnhist + $h_{\mathrm{oracle}}$, (iv) mnhist-oracle + $h_{\mathrm{oracle}}$, and (v) the baseline estimator.
\begin{table}
\caption{WMSEs ($\times 100$) for five estimators under the mixed sparsity setting: (i) mnhist with plug-in bandwidth $\widehat{h}$, (ii) mnhist-oracle with $\widehat{h}$, (iii) mnhist with oracle bandwidth $h_{\mathrm{oracle}}$, (iv) mnhist-oracle with $h_{\mathrm{oracle}}$, and (v) baseline. Results are reported for multiplex networks with $n=400$ vertices across $L \in \{5,7,10\}$ layers under the Homogeneous, Perturbed, and Heterogeneous scenarios. Each function ID corresponds to a graphon combination specified in Table~1. Values are averaged over 100 replications. The reported bandwidths include $\widehat{h}$ and the oracle bandwidths separately selected for mnhist and mnhist-oracle.}
\label{tab:decomposition_by_graphon}
\centering
\centering
\resizebox{\textwidth}{!}{%
\begin{tabular}{c|c|ccc|cccc|c}
\hline
ID & $L$ & $\widehat{h}$ & \shortstack{$h_{oracle}$\\(mnhist)} & \shortstack{$h_{oracle}$\\(mnhist-oracle)} & \shortstack{mnhist\\$+\widehat{h}$} & \shortstack{mnhist-oracle\\$+\widehat{h}$} & \shortstack{mnhist\\$+h_{oracle}$} & \shortstack{mnhist-oracle\\$+h_{oracle}$} & baseline \\
\hline
\multicolumn{10}{l}{\textbf{Homogeneous}} \\
\hline
1 &  5 & 27.5 & 73.6 & 41.9 & 3.550 & 0.952 & 1.652 & 0.713 & 13.93 \\
1 &  7 & 27.2 & 67.0 & 43.7 & 2.930 & 0.939 & 1.419 & 0.686 & 14.24 \\
1 & 10 & 27.3 & 57.8 & 43.9 & 2.433 & 0.954 & 1.218 & 0.712 & 14.11 \\
\hline
2 &  5 & 33.6 & 117.7 & 61.1 & 0.827 & 0.168 & 0.246 & 0.087 & 1.21 \\
2 &  7 & 33.8 & 96.8 & 62.5 & 0.686 & 0.171 & 0.230 & 0.090 & 1.22 \\
2 & 10 & 33.7 & 86.7 & 62.4 & 0.558 & 0.168 & 0.190 & 0.087 & 1.22 \\
\hline
3 &  5 & 41.8 & 70.7 & 51.0 & 0.633 & 0.150 & 0.352 & 0.147 & 1.30 \\
3 &  7 & 41.9 & 65.3 & 51.7 & 0.494 & 0.150 & 0.285 & 0.147 & 1.30 \\
3 & 10 & 41.9 & 63.8 & 50.9 & 0.404 & 0.150 & 0.241 & 0.148 & 1.29 \\
\hline
4 &  5 & 42.4 & 46.6 & 37.3 & 0.899 & 0.479 & 0.886 & 0.466 & 10.99 \\
4 &  7 & 42.4 & 42.9 & 37.3 & 0.751 & 0.471 & 0.751 & 0.461 & 11.05 \\
4 & 10 & 42.5 & 42.5 & 38.1 & 0.660 & 0.472 & 0.660 & 0.463 & 11.14 \\
\hline
\multicolumn{10}{l}{\textbf{Perturbation}} \\
\hline
1, 2 &  5 & 32.9 & 96.9 & 60.8 & 0.885 & 0.181 & 0.301 & 0.100 & 1.48 \\
1, 2 &  7 & 33.0 & 92.2 & 61.9 & 0.723 & 0.182 & 0.246 & 0.098 & 1.45 \\
1, 2 & 10 & 33.2 & 87.7 & 60.3 & 0.581 & 0.178 & 0.205 & 0.097 & 1.43 \\
\hline
1, 4 &  5 & 41.2 & 47.9 & 37.0 & 0.919 & 0.474 & 0.902 & 0.463 & 10.79 \\
1, 4 &  7 & 41.6 & 42.2 & 35.4 & 0.759 & 0.477 & 0.757 & 0.465 & 10.91 \\
1, 4 & 10 & 41.7 & 41.7 & 36.3 & 0.674 & 0.473 & 0.674 & 0.465 & 10.84 \\
\hline
3, 4 &  5 & 42.5 & 49.1 & 38.1 & 0.883 & 0.458 & 0.864 & 0.446 & 10.47 \\
3, 4 &  7 & 42.5 & 43.6 & 40.0 & 0.741 & 0.453 & 0.740 & 0.449 & 10.37 \\
3, 4 & 10 & 42.5 & 43.1 & 39.5 & 0.649 & 0.456 & 0.648 & 0.451 & 10.52 \\
\hline
\multicolumn{10}{l}{\textbf{Heterogeneous}} \\
\hline
1, 2 &  5 & 32.5 & 95.1 & 58.4 & 1.249 & 0.263 & 0.461 & 0.162 & 2.37 \\
1, 2 &  7 & 32.4 & 94.1 & 59.2 & 0.955 & 0.247 & 0.370 & 0.155 & 2.24 \\
1, 2 & 10 & 32.4 & 90.3 & 60.5 & 0.674 & 0.217 & 0.275 & 0.135 & 1.98 \\
\hline
1, 4 &  5 & 37.1 & 55.8 & 37.3 & 1.313 & 0.536 & 0.979 & 0.536 & 11.58 \\
1, 4 &  7 & 37.4 & 53.8 & 37.3 & 0.997 & 0.532 & 0.856 & 0.532 & 11.58 \\
1, 4 & 10 & 38.0 & 46.6 & 38.0 & 0.747 & 0.488 & 0.718 & 0.488 & 11.50 \\
\hline
3, 4 &  5 & 41.9 & 59.8 & 41.3 & 1.076 & 0.500 & 0.972 & 0.499 & 9.95 \\
3, 4 &  7 & 42.1 & 50.2 & 42.0 & 0.851 & 0.487 & 0.827 & 0.487 & 9.90 \\
3, 4 & 10 & 42.2 & 42.6 & 40.5 & 0.672 & 0.474 & 0.672 & 0.472 & 10.22 \\
\hline
\end{tabular}%
}
\end{table}

Table~\ref{tab:decomposition_by_graphon} reports the WMSE of the estimators described above under the settings corresponding to Tables~\ref{tab:homogeneous_mixed_by_L}--\ref{tab:Heterogeneous_mixed_by_L} in the main manuscript. Compared with the baseline estimator $\widehat{f}=1$ (v), mnhist (i)--(iv) consistently achieves lower WMSE regardless of the bandwidth and label selection schemes. In addition, when mnhist is used, WMSE tends to decrease as $L$ increases regardless of the bandwidth selection scheme. The following comparisons evaluate whether this improvement comes from bandwidth selection or label estimation.

First, we assess the effect of bandwidth selection by comparing the WMSEs of estimators using the plug-in bandwidth and the oracle bandwidth under oracle labels ((ii) mnhist-oracle + $\widehat{h}$ vs. (iv) mnhist-oracle + $h_{\mathrm{oracle}}$).
The WMSE with the oracle bandwidth is mostly similar to that with the plug-in bandwidth, with a ratio of 1.0--1.4. For some function combinations in the Homogeneous and Perturbation scenarios, the ratio increases to 1.6--1.9. Thus, the plug-in bandwidth selected by Algorithm~\ref{alg:bw_sel} achieves similar performance to the oracle bandwidth in most cases.

To evaluate the effect of label estimation error, we now compare the WMSEs of the estimators using the estimated and oracle labels under the plug-in bandwidth ((i) mnhist vs. (ii) mnhist-oracle + $\widehat{h}$). The WMSE of mnhist is 1.4--5 times larger than that of mnhist-oracle. We also observe similar patterns when oracle bandwidths are separately selected for each label setting ((iii) mnhist + $h_{\mathrm{oracle}}$ vs. (iv) mnhist-oracle + $h_{\mathrm{oracle}}$). Combined with the previous results, this shows that graphon estimation error is more affected by label estimation error than by bandwidth selection.

\section{Data analysis: further on Indian village data}
\label{sec_append:indian_vils}

We first provide a brief description of the original dataset; see \cite{banerjee2013Diffusion} for the full description.
The dataset contains 75 villages, each with two types of networks, household and individual. Each network type is constructed from twelve kinds of socio-economical interactions, along with their union, and their intersection. 
Each individual and household has associated demographic information. Table~\ref{tab:covariate} lists the covariates selected for the data analysis. A full description of the survey questions and their possible responses is available in \cite{indian_vil_data}. In the main article, we focus on two covariates for the bin summary plots: caste and electricity.

For our analysis, we focus on the 12 household social interactions from village ID 40, treating each layer as undirected and simple graphs. Table~\ref{tab:indian_vil_layer_info} lists the interaction types, and Table~\ref{tab:indian_vil_40} summarizes their summary statistics of Village 40. The layer-wise edge densities (sparsity) vary from 0.0021 (temple company) to 0.0198 (visit come), implying a sparse pattern across layers. The average degree varies from 0.4762 to 4.5628. 
Degree scale parameters for most layers range from 4.3684 to 6.5971, except for the visit come layer ($\gamma= 12.56$), indicating that the degree distributions decay faster than those of typical scale-free networks.
Most of the layers show weak degree assortativity, suggesting a limited tendency of vertices for connecting with similar degrees. The global clustering coefficients range from 0.1111 to 0.3149, with the kero rice come/go and rel layers showing comparatively higher clustering.
\begin{table}[!th]
\caption{Covariates about household characteristics used in Figure~\ref{fig:indian_vil_40_covariate}}
\label{tab:covariate}
\begin{tabular}{|p{0.31\linewidth}|p{0.59\linewidth}|}
\hline
Variable        & Response               \\ \hline
Caste    & 1: OBC, 2: Schedule Caste, 3: General, 4: Schedule Tribe, 5: Other       \\
Have electricity?       & 1: Yes, Private, 2: Yes, Government, 3: No 
\\ \hline
\end{tabular}
\end{table}
\begin{table}[!h]
\caption{Description of layers in Indian village dataset \cite{banerjee2013Diffusion}.}
\label{tab:indian_vil_layer_info}
\centering
\begin{tabular}{|p{0.21\linewidth}|p{0.7\linewidth}|}
\hline
\textbf{Layers} & \textbf{Description}             \\ \hline
borrow money     & those from whom the respondent would borrow money   \\
give advice      & those to whom the respondent gives advice           \\
help decision    & those from whom the respondent gets advice          \\
kerorice come    & those from whom the respondent  would borrow material goods (kerosene, rice, etc.)
\\
kerorice go      & those to whom the respondent would lend material goods             \\
lend money       & those to whom the respondent would lend money       \\
medic            & those from whom the respondent receives medical advice             \\
nonrel           & nonrelatives with whom the respondent socializes    \\
rel              & kin in the village    \\
temple company   & those whom the respondent goes to pray with (at a temple, church, or mosque)      \\
visit come       & those who visit the respondent’s home               \\
visit go         & those whose homes the respondent visits             \\
And              & Intersection of all networks above   \\
Or               & Union of all networks above     
\\ \hline       
\end{tabular}
\end{table}
\begin{table}[!ht]
\vspace{10pt}
\caption{Network summary statistics for each layer of village ID 40. The numbers of vertices and edges are denoted by $V$ and $E$, respectively. The edge density is $\rho$. The average degree, standard deviation of degrees, and maximum degree are denoted by $\bar{d}$, $sd(d)$, and $\max(d)$, respectively. 
The degree scale parameter is denoted by $\gamma$, the degree assortativity coefficient by $r$, and the global clustering coefficient by $C_{\Delta}$.}
\label{tab:indian_vil_40}
\centering
\begin{tabular}{l|llllll}
Layer    & $V$ & $E$ & $\rho$ & $\bar{d}$ & $sd(d)$ & $\max(d)$ \\ \hline
borrow money   & 231    & 460    & 0.0173        & 3.9827      & 3.1759     & 17                    \\
give advice    & 231    & 275    & 0.0104        & 2.3810      & 2.3167     & 11            \\
help decision  & 231    & 290    & 0.0109        & 2.5108      & 2.3570     & 11                  \\
kero rice come & 231    & 442    & 0.0166        & 3.8268      & 2.8507     & 14                  \\
kero rice go   & 231    & 444    & 0.0167        & 3.8442      & 2.8502     & 14                   \\
lend money     & 231    & 413    & 0.0155        & 3.5758      & 2.8563     & 17                      \\
medic          & 231    & 398    & 0.0150        & 3.4459      & 2.7952     & 17                  \\
nonrel         & 231    & 449    & 0.0169        & 3.8874      & 3.2631     & 16                     \\
rel            & 231    & 229    & 0.0086        & 1.9827      & 1.6124     & 7                     \\
temple company & 231    & 55     & 0.0021        & 0.4762      & 0.8486     & 6                        \\
visit come     & 231    & 527    & 0.0198        & 4.5628      & 3.1505     & 13                 \\
visit go       & 231    & 524    & 0.0197        & 4.5368      & 3.0010     & 12                  \\ \hline 
\end{tabular}
\begin{tabular}{l|llll}
Layer    & $\gamma$ & $r$ & $C_{\Delta}$ & diameter \\ \hline
Borrow money   & 5.1879             & -0.0494             & 0.1991               & 10   \\
give advice    & 5.5089             & -0.0475             & 0.1474               & 14       \\
help decision  & 4.9165             & 0.0462              & 0.1755               & 13       \\
kero rice come & 5.1702             & 0.0378              & 0.2898               & 12       \\
kero rice go   & 5.4734             & 0.0602              & 0.2909               & 11       \\
lend money     & 6.5959             & -0.0379             & 0.1933               & 9        \\
medic          & 6.5971             & -0.0305             & 0.2115               & 10       \\
nonrel         & 4.3684             & -0.0381             & 0.1142               & 9        \\
rel            & 5.4921             & 0.1197              & 0.3149               & 23       \\
temple company & 4.9032             & 0.0409              & 0.1111               & 8        \\
visit come     & 12.5626            & -0.0932             & 0.2276               & 8        \\
visit go       & 5.2070             & -0.1320             & 0.2243               & 8       \\ \hline 
\end{tabular}
\end{table}
\begin{figure}[!h]
    \centering
    \includegraphics[width=0.5\linewidth, trim={1cm, 1.5cm, 1cm, 1.5cm}]{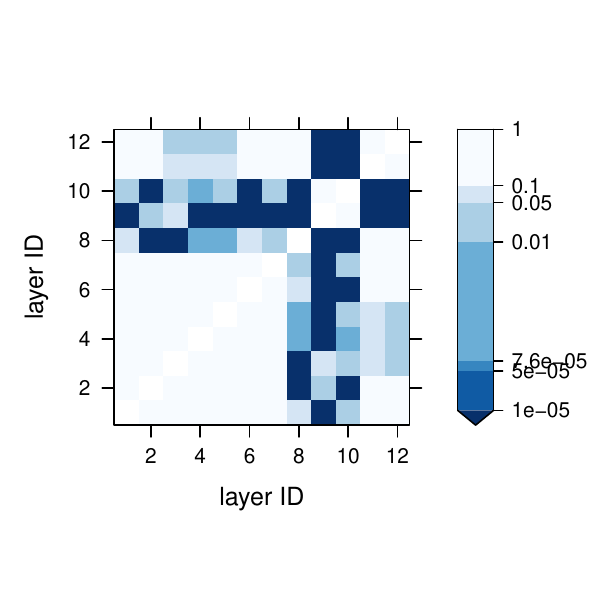}
    \caption{$p$-values from pair-wise tests by \cite{shao2022Higherorder} using triangle densities. Cell colours indicate $p$-values for each pair of distinct layers, with darker cells representing smaller $p$-values for the corresponding pair.}
    \label{fig:indian_vil_40_net_comparison}
\end{figure}

Before fitting the homogeneous multi-network histogram, we first identify layers with similar structure. We apply a two-sample network comparison test \cite{shao2022Higherorder}, where the null hypothesis is that the normalized network moments are identical ($H_0: (\rho^{(1)})^{-s} \mu_1 = (\rho^{(2)})^{-s} \mu_2$), with $\mu_\ell$ denoting a network statistic (e.g., triangle density) and $s$ the number of edges involved. Pairwise comparisons are performed across the 12 layers using triangle densities ($s=3$), controlling the family-wise error rate at 0.05 via Bonferroni correction. Figure~\ref{fig:indian_vil_40_net_comparison} displays the resulting pairwise $p$-value matrix. We find insufficient evidence to conclude that the following nine layers differ: borrow money, give advice, help decision, kero rice come, kero rice go, lend money, medic, visit come, and visit go. Therefore, we use the homogeneous multi-network histogram to estimate a common graphon for these nine layers.

% \bibliographystyle{comnet}
% \bibliography{multinethist/ref}

\end{appendix}
%
% once the .bbl file has been generated then place the text in your article.

% To get the unnumbered reference style the author should use [unnumbib]
%as an option in the document class.  For example: \documentclass[unnumbib]{comnet}

% \begin{thebibliography}{99}

% \bibitem{Rottmann:2010a}
% {\sc Rottmann-Matthes, J.} (2011a) Linear stability of traveling
% waves in nonstrictly hyperbolic PDES.\break {\em J. Dynam.
% Differential Equations}, \textbf{23}, 365--393.

% \bibitem{Rottmann:2011a}
% {\sc Rottmann-Matthes, J.} (2011b) Stability and freezing of
% nonlinear waves in first-order hyperbolic PDEs. Preprint
% 11-016, CRC 701, Bielefeld University.

% \bibitem{Rottmann:2011b}
% {\sc Rottmann-Matthes, J.} (2012) Stability of
% parabolic-hyperbolic traveling waves. Preprint 12-005, CRC
% 701, Bielefeld University.

% \bibitem{RowleyKevrekidisMarsdenLust:2003}
% {\sc Rowley, C. W., Kevrekidis, I. G., Marsden, J. E. \& Lust,
% K.} (2003) Reduction and reconstruction for self-similar
% dynamical systems. {\em Nonlinearity}, \textbf{16},
% 1257--1275.

% \end{thebibliography}
\end{document}